\documentclass{article}

\usepackage[T1]{fontenc}
\usepackage{amsmath,amsfonts,amsthm}
\usepackage{mathtools}
\usepackage{thmtools}
\usepackage{enumitem}
\usepackage[normalem]{ulem}
\usepackage[letterpaper,top=2cm,bottom=2cm,left=3cm,right=3cm,marginparwidth=1.75cm]{geometry}
\usepackage{authblk}

\usepackage{multirow}
\usepackage{comment}
\usepackage{complexity}
\usepackage{subcaption}
\usepackage{float}
\usepackage{nicefrac}

\usepackage[dvipsnames]{xcolor}
\usepackage[colorlinks=true,pdfpagemode=UseNone,urlcolor=RoyalBlue,linkcolor=RoyalBlue,citecolor=OliveGreen,pdfstartview=FitH,linktocpage]{hyperref}

\usepackage{dsfont}
\usepackage{xspace}
\usepackage{tikz}
\usepackage[skins]{tcolorbox}
\usepackage[nameinlink, capitalise]{cleveref}
\crefname{section}{Section}{Sections}
\crefname{subsection}{Subsection}{Subsections}
\usepackage{csquotes}
\usepackage[backend=biber, style=alphabetic, maxbibnames=10, minalphanames=3, maxalphanames=4]{biblatex}
\DeclareSourcemap{
  \maps[datatype=bibtex, overwrite]{
    \map{
      \step[fieldset=editor, null]
    }
  }
}
\DeclareSourcemap{
  \maps[datatype=bibtex]{
    \map{
       \step[fieldset=series, null]
    }
  }
}

\declaretheorem{theorem}
\declaretheorem[sibling=theorem]{lemma}
\declaretheorem[sibling=theorem]{claim}
\declaretheorem[sibling=theorem]{claim*}

\declaretheorem[sibling=theorem]{fact}

\declaretheorem[sibling=theorem]{definition}

\theoremstyle{definition}

\usepackage{figchild}

\newcommand{\eps}{\epsilon}
\renewcommand{\epsilon}{\varepsilon}

\newcommand{\ind}{\mathds{1}}

\usepackage{float}
\usepackage{tikz}
\usetikzlibrary{calc} 
\usepackage{array}
\usepackage{multirow}

\usepackage[noend]{algpseudocode}
\usepackage{amsmath, amssymb}
\usepackage{amsthm}
\usepackage{comment}
\usepackage{mdframed}
\usepackage{xspace}
\usepackage{multicol}
\usepackage{tabularx}
\usepackage{makecell}

\usepackage{enumitem}

\declaretheoremstyle[
    bodyfont=\normalfont  
]{upright}

\declaretheorem[
    style=upright,  
    refname={Algorithm,Algorithms},
    Refname={Algorithm,Algorithms},
    name={Algorithm}
]{algorithm}
\usepackage{graphicx}

\newcommand{\dist}{d}

\newcommand{\Sf}{S_{\mathsf{free}}}
\newcommand{\Sr}{S_{\mathsf{reg}}}
\newcommand{\Ff}{F_{\mathsf{free}}}

\newcommand{\calI}{\mathcal{I}}

\newcommand{\calB}{\mathcal{B}}
\newcommand{\calC}{\mathcal{C}}

\newcommand{\calN}{\mathcal{N}}
\newcommand{\calQ}{\mathcal{Q}}

\newcommand{\calH}{\mathcal{H}}
\newcommand{\calM}{\mathcal{M}}
\newcommand{\calP}{\mathcal{P}}
\newcommand{\calX}{\mathcal{X}}

\newcommand{\calU}{\mathcal{U}}
\newcommand{\bcalU}{\widetilde{\mathcal{U}}}
\newcommand{\tmu}{\widetilde{\mu}}
\newcommand{\calR}{\mathcal{R}}

\newcommand{\bcalR}{\widetilde{\mathcal{R}}}
\newcommand{\calL}{\mathcal{L}}

\newcommand{\clcost}{\text{cost}}
\newcommand{\closedclcost}{\text{closedcost}}
\newcommand{\increase}{\text{cost-inc}}

\newcommand{\clients}{D}
\newcommand{\clientspure}{D^*_{\texttt{pure}}}
\newcommand{\clientscheap}{D^*_{\texttt{cheap}}}
\newcommand{\clientsexpensive}{D^*_{\texttt{exp}}}
\newcommand{\facilities}{F}
\newcommand{\opt}{\text{OPT}}
\newcommand{\cost}{\text{cost}}
\newcommand{\sopt}{\text{opt}}
\newcommand{\optlpfl}{\text{opt}_{\text{LP}}}

\newcommand{\Maxdist}{M}
\newcommand{\dummyset}{\Lambda}
\newcommand{\core}{C^{\text{core}}}

\newcommand{\logadaptalg}{\textsc{LogAdaptiveAlgorithm}\xspace}

\newcommand{\JMSalg}{\textsc{Greedy Algorithm}\xspace}
\newcommand{\mergealg}{\textsc{MergeSolutions}\xspace}
\newcommand{\completesol}{\textsc{CompleteSolution}\xspace}
\newcommand{\completesolx}{\textsc{CompleteSolution}}
\newcommand{\completesequence}{\textsc{CompleteSequence}\xspace}
\newcommand{\calLexp}{\calL_{\text{exp}}}

\usepackage[colorinlistoftodos,textsize=tiny,textwidth=2cm,color=green!50!gray]{todonotes}

\newlength\hwhatcwidth
\newlength\hwhatuwidth
\newcommand*\halfwidehat[1]{%
  \settowidth\hwhatcwidth{\ensuremath{#1}}%
  \settowidth\hwhatuwidth{\ensuremath{\mathrm{#1}}}%
  \addtolength\hwhatuwidth{-\hwhatcwidth}
  \makebox[0pt][l]{
    \kern\dimexpr0.25\hwhatcwidth-0.667\hwhatuwidth\relax
    \ensuremath{\widehat{\vphantom{#1}\rule{0.5\hwhatcwidth}{0pt}}}
  }#1
}

\newtcolorbox{myframe}[2][]{%
  enhanced,colback=white,colframe=black,coltitle=black,
  sharp corners,boxrule=0.4pt,
  fonttitle=\itshape,
  attach boxed title to top left={yshift=-0.3\baselineskip-0.4pt,xshift=2mm},
  boxed title style={tile,size=minimal,left=0.5mm,right=0.5mm,
    colback=white,before upper=\strut},
  title=#2,#1
}

\title{An Improved Greedy Approximation\\ for (Metric) $k$-Means\thanks{R. Gao was supported by a Stanford Graduate Fellowship. F. Grandoni was partially supported by the SNF Grants 200021-200731 and 200021-236706. E. Lee was partially supported by the NSF Grant 2236669. E. van. Wijland. 
was partially supported by the French PEPR integrated projects EPIQ (ANR-22-PETQ-0007).}
}

\author[1]{Moses Charikar}
\author[2]{Vincent Cohen-Addad}
\author[1]{Ruiquan Gao}
\author[3]{Fabrizio Grandoni}
\author[4]{Euiwoong Lee}
\author[5]{Ernest van Wijland}
\affil[1]{Stanford University, \url{{moses, ruiquan}@cs.stanford.edu}}
\affil[2]{Google Reserach, \url{cohenaddad@google.com}}
\affil[3]{IDSIA, USI-SUPSI, \url{fabrizio.grandoni@gmail.com}}
\affil[4]{University of Michigan, \url{euiwoong@umich.edu}}
\affil[5]{Université Paris-Cité, \url{ernest.vanwijland@irif.fr}}

\date{}

\newcommand{\avg}{\text{avg}}

\newcommand{\ho}{\hat{\opt}}

\newcommand{\apxLSkmeans}{25}

\begin{document}

    \maketitle


\begin{abstract}
\noindent 
Clustering is a basic task in data analysis and machine learning, and the optimization of clustering objectives are well-studied optimization problems; amongst these, the $k$-Means objective is arguably the most well known.
Given a collection of points in a metric space, the goal is to partition them into $k$ clusters, each with an associated center, so as to minimize the sum of squared distances of points to their cluster centers.  
In this paper, we present a polynomial-time $3+2\sqrt{2}+\eps<5.83$-approximation algorithm for $k$-Means in general metrics. This substantially improves on the current-best $(9+\eps)$-approximation in [Ahmadian, Norouzi-Fard, Svensson, Ward - FOCS'17, SICOMP'20], and even slightly improves on the $5.92$-approximation in [Cohen-Addad, Esfandiari, Mirrokni, Narayanan - STOC'22] for the Euclidean special case.

A natural approach for $k$-Means is to leverage Lagrangian Multiplier Preserving (LMP)  approximations for the facility location problem.
The previous best results for $k$-Means build upon an adaptation of an LMP $3$-approximation for facility location with metric connection costs in [Jain, Vazirani - J.ACM'01] based on a primal-dual method, rather than on the improved LMP greedy $2$-approximation for the same problem in [Jain, Mahdian, Markakis, Saberi, Vazirani - J.ACM'03]. The barrier to using the improved LMP algorithm was that no adaptation of this algorithm and its analysis to the case of squared metric connection costs was known (since squared distances violate triangle inequality). Our main contribution is overcoming this barrier by providing such an adaptation. This new LMP approximation algorithm is then combined with the framework recently introduced in [Cohen-Addad, Grandoni, Lee, Schwiegelshohn, Svensson - STOC'25] for the related (metric) $k$-Median problem.
\end{abstract}

\thispagestyle{empty}
\newpage

\tableofcontents
\thispagestyle{empty}

\newpage
\setcounter{page}{1}

\section{Introduction}
\label{sec:intro}

In a generic clustering problem we are given a collection of points together with a dissimilarity measure (i.e.~distance function) between pairs of points. The high-level goal is to partition the points into a ``small'' number of \emph{clusters} so that similar points are clustered together while dissimilar ones are clustered separately. One of the most fundamental and best-studied clustering problems is \emph{$k$-Means}. Here we are given a collection $D$ of $n$ points (or \emph{clients}) and a collection $F$ of \emph{centers} (or \emph{facilities}), as well as a integer $k>0$. We are also given metric distances $d:(D\cup F) \times (D\cup F)\rightarrow \mathbb{R}_{\geq 0}$. Our goal is to select a set $S$ of $k$ centers (the \emph{open} centers) so as to minimize the sum of the squared distances from each client to the closest open center\footnote{It can equivalently be phrased as the problem of partitioning the given points ($D$) into $k$ disjoint clusters and finding the best center (amongst $F$) for each cluster, to minimize the $k$-Means objective.}, i.e.,
$$
\sum_{j\in D}d^2(j,S)=\sum_{j\in D}\min_{i\in S}d^2(j,i).
$$
Observe that each feasible solution $S$ naturally induces $k$ clusters where the cluster associated with $i\in S$ is given by the clients that are closer to $i$ than to any other center in $S$ (breaking ties arbitrarily). We next use $\opt_k$ to denote a reference optimal solution and $\sopt_k$ to denote its cost. When $k$ is clear from the context, we simply use $\opt$ and $\sopt$.

$k$-Means is NP-hard and well-studied in terms of approximation algorithms. It is impossible to approximate it (in polynomial time) better than a factor $1+8/e \approx 3.94$ \cite{JMS02}. The current-best (polynomial-time) approximation factor for this problem is $9+\eps$ for any constant $\eps>0$ by Ahmadian, Norouzi-Fard, Svensson, and Ward \cite{AhmadianNSW20}. Their algorithm is based on a primal-dual Lagrangian Multiplier Preserving (LMP) $9$-approximation algorithm for facility location with squared metric connection costs (similar in spirit to a classical primal-dual $3$-approximation algorithm by Jain and Vazirani \cite{JaV01} for metric connection costs),
plus a careful way to combine so-called bi-point solutions that introduces a $1+\eps$ factor only in the approximation (rather than a larger constant factor as in prior work). This improved on an earlier $\apxLSkmeans$-approximation by Gupta and Tangwongsan \cite{LSGupta} based on local search: we will need this result later. Interestingly a $(1+8/e+\eps)$-approximation can be obtained in FPT time \cite{Cohen-AddadG0LL19}.

Very often in practice one considers the Euclidean special case of $k$-Means (which in the literature is often just called $k$-Means), where the clients are $n$ points in the ($d$-dimensional) Euclidean space, with the respective distances. In this case one is allowed to select any point in the Euclidean space as a center. However known reductions \cite{matouvsek2000approximate, de2003approximation, feldman2007ptas} 
allow one to focus on a set of candidate centers $F$ of size $n^{O_\eps(1)}$ while introducing a factor $1+\eps$ in the approximation for any constant $\eps>0$. Therefore one can (essentially) see the Euclidean case as a special case of the metric one. The best-known approximation factor for the Euclidean case is $5.92$ by Cohen-Addad, Esfandiari, Mirrokni, Narayanan \cite{Cohen-AddadEMN22}. This improves on a sequence of increasingly better approximations for the problem: a $(9+\eps)$-approximation by Kanungo, Mount, Netanyahu, Piatko, Silverman, and Wu \cite{KanungoMNPSW04}, a $6.36$-approximation by Ahmadian et al. \cite{AhmadianNSW20}, and a $6.13$-approximation by Grandoni, Ostrovsky, Rabani, Schulman, and Venkat \cite{GrandoniORSV22}.

Our main result is a $5.83$-approximation for (metric) $k$-Means. More precisely, we get the following:
\begin{theorem}\label{thr:main}
For any constant $\eps>0$, there is a polynomial-time randomized $(3+2\sqrt{2}+\eps)$-approximation for (metric) $k$-Means.    
\end{theorem}
We remark that the above result not only substantially improves the current best $(9+\eps)$-approximation for the general metric case \cite{AhmadianNSW20}, but even slightly the current best $5.92$-approximation for the Euclidean case \cite{Cohen-AddadEMN22}. We overview our approach in Section \ref{sec:overview}.

\subsection{Related Work}
\label{sec:related}

$k$-Means belongs to the family of $k$-clustering problems where the target number of clusters $k$ is fixed. Other famous examples are $k$-Center and $k$-Median. In $k$-Center one wishes to select a set $S$ of $k$ centers so as to minimize the maximum distance from any client to $S$, i.e., the objective function to minimize is $\max_{j\in D}d(j,S)$. This problem admits a simple greedy $2$-approximation which is best possible unless $P=NP$ \cite{Gon85,HoS86}. $k$-Median is defined like $k$-Means, except that here one wishes to minimize the sum of the distances rather than squared distances, i.e., the objective function is $\sum_{j\in D}d(j,S)$. $k$-Median is very close to   
$k$-Means in terms of results and techniques. $k$-Median is hard to approximate below a factor $1 + 2/e$~\cite{JMS02}. For general metric distances, the first constant approximation was achieved by Charikar, Guga, Tardos and Shmoys \cite{CharikarGTS99}. After a very long sequence of improvements \cite{JaiV01,JainMMSV03,AryaGKMMP04,LiS16,ByrkaPRST17,Cohen-AddadLS23,GowdaPST23}, the current-best $(2+\eps)$-approximation for this problem was very recently achieved by Cohen-Addad, Grandoni, Lee, Schwiegelshohn, and Svensson \cite{CGLSS25stoc}. This is also the best result for the Euclidean case, improving on an earlier $2.406$-approximation \cite{Cohen-AddadEMN22}.

\section{Overview of our Approach}
\label{sec:overview}

At a very high level, our approach looks similar in spirit to the one leading to a recent $(2+\eps)$-approximation for the related $k$-Median problem by Cohen-Addad, Grandoni, Lee, Schwiegelshohn, and Svensson \cite{CGLSS25stoc}. In more detail, we exploit a combination of two different algorithms. The first one is a bicriteria approximation algorithm with the desired approximation factor, which opens $O(\log n/\eps^2)$ more centers than the $k$ allowed ones. 
\begin{theorem}\label{thr:mainBicriteria}
For any constant $\eps\in (0,1/6)$, there is a polynomial-time algorithm for $k$-Means that returns a solution containing at most $k+O(\log n/\eps^2)$ centers and of cost at most $(3+2\sqrt{2}+\eps)\sopt$.  
\end{theorem}
The second algorithm is a radically different $(5+\eps)$-approximation algorithm that only works for instances which are \emph{stable} in the following sense. We say that an instance of $k$-Means is $\beta$-stable if $\sopt_{k-1}\geq (1+\beta)\sopt_k$, where $\sopt_{k-1}$ is the optimal cost for the same instance but with the target number of centers being $k-1$.
\begin{theorem}\label{thr:mainStable}
For any constants $\eps,\zeta>0$, there is a polynomial-time randomized algorithm for $k$-Means that with high probability returns a solution of cost at most $(5+\eps)\sopt$ assuming that the input instance is $(\zeta/\log n)$-stable.   
\end{theorem}
It is relatively easy to derive Theorem \ref{thr:main} from the above two theorems. 
\begin{proof}[Proof of Theorem \ref{thr:main}]
We compute a set of feasible solutions, and return the cheapest one. Let $\Delta$ be the maximum extra number of centers computed by the algorithm from Theorem \ref{thr:mainBicriteria} (this number is independent from $k$). Let 
$k'=\max\{1,k-\Delta\}$. One solution is obtained by computing the optimum solution with one center if $k'=1$, and otherwise by running
the algorithm from Theorem \ref{thr:mainBicriteria} with target number of centers being $k'$ (notice that it returns a solution with at most $k$ centers, hence feasible). Furthermore, we run the algorithm from Theorem \ref{thr:mainStable} for every integer $k''\in (k',k]$ (thus obtaining solutions with $k''\leq k$ centers, hence feasible).

If 
$\sopt_{k'}$ is not much larger than $\sopt_k$ (more precisely
$\sopt_{k'}\leq (1+\eps)\sopt_k$), 
the first solution is  $(1+\eps)(3+2\sqrt{2}+\eps)$ approximate. Otherwise, notice that there exists an integer $k''\in (k',k]$ such that $\sopt_{k''}\leq (1+\eps)\sopt_{k}$ and $\sopt_{k''-1}\geq (1+\beta)\sopt_{k''}$ for $\beta\in \Omega(\eps/\Delta)=\Omega(\eps^3/\log n)$. For that value of $k''$ the corresponding instance is $\beta$-stable, hence the respective solution has cost at most $(5+\eps)\sopt_{k''}\leq (1+\eps)(5+\eps)\sopt_k$. The claim follows by rescaling $\eps$ by a constant factor. 
\end{proof}

It remains to describe how the above two theorems are obtained. The algorithm from Theorem \ref{thr:mainStable} is rather complex and its analysis is highly non-trivial (see Section \ref{sec:centerremoval}). However, it is a relatively easy adaptation of the $(2+\eps)$-approximation algorithm in \cite{CGLSS25stoc} for a similar notion of stable $k$-Median instances. The main difference is that we have to carefully use an approximate form of triangle inequality (see Lemmas \ref{lem:apxTriangleInequality2} and \ref{lem:apxTriangleInequality3}). This also explains why we get a higher approximation factor $5+\eps$.

The main contribution of this paper is our proof of Theorem \ref{thr:mainBicriteria}. In more detail, we consider the related facility location problem (with uniform opening costs). Recall that here, instead of a bound $k$ on the number of centers, we are given a uniform facility cost $f$. We are now allowed to open an arbitrary number of centers/facilities $S$, and the objective function is the total cost of the open facilities plus the squared distance from each client to the closest (open) facility in $S$, i.e., 
$$
\cost_{FL}(S):=f|S|+\sum_{j\in D}d^2(j,S)\,.
$$
We say that an algorithm for the above problem is LMP $\Gamma$-approximate if it produces a feasible solution $S$ such that
\begin{equation}
\Gamma\cdot f|S|+\sum_{j\in D}d^2(j,S)\leq \Gamma\cdot \sopt_{LP}(f).
\label{eq:LMP}
\end{equation}
where $\sopt_{LP}(f)\leq \sopt$ is the optimal cost of a standard LP relaxation for the problem (see Section \ref{sec:preliminaries}). 
In other words, the solution cost at most $\Gamma \cdot \sopt$ even if we increase the facility cost of $S$ by a factor $\Gamma$.

Recall that the current-best $(9+\eps)$-approximation for (Metric) $k$-Means by Ahmadian et al. \cite{AhmadianNSW20} builds upon a primal-dual LMP $9$-approximation for facility location (with squared metric connection costs), which is inspired by a classical primal-dual $3$-approximation by Jain and Vazirani \cite{JaV01} for the case of metric connection costs. Our main contribution is a \emph{greedy} LMP $(3+2\sqrt{2})$-approximation for facility location with squared metric connection costs which is a non-trivial (but relatively simple) variant of the classical greedy LMP $2$-approximation (JMMSV) for facility location with metric connection costs by Jain, Mahdian, Markakis, Saberi, and Vazirani \cite{JainMMSV03}. It is instructive to recall how JMMSV works in order to see the differences.

\paragraph{The JMMSV Greedy Algorithm.}
The algorithm has a variable $\alpha_j$ per client $j$ (initialized to $0$), a set $S$ of open facilities (initialized to $\emptyset$), and a set $A$ of \emph{active} clients (initialized to $D$ ). At each point of time, the \emph{bid} $bid(j,i)$ of client $j$ towards facility $i$ is $[\alpha_j-d(j,i)]^+$ if $j$ is active and $[d(j,S)-d(j,i)]^+$ otherwise\footnote{We let $[a]^+:=\max\{0,a\}$.}. The variables $\alpha_j$ of active clients are increased uniformly until one of the following events happens:
\begin{enumerate}[label=(\alph*)]
    \item For some client $j$ and $i\in S$, $\alpha_j\geq d(j,i)$. In that case $j$ is removed from $A$ and we say that $j$ is connected to $i$;
    \item For some (not open) facility $i\notin S$, one has $\sum_{j\in D}bid(j,i)=f$.\footnote{To be precise, this version of JMMSV establishes an equivalent notion of approximation, but not exactly the LMP 2-approximation defined in \Cref{eq:LMP}; the latter is achieved when we run the algorithm with scaled opening cost $\hat f = 2f$, which our actual algorithms do with different scaling factors. In this overview, let us ignore this subtlety and conflate these two notions of approximations while just using $f$.}
    In that case we open $i$, i.e., we add $i$ to $S$. Furthermore, each $j\in A$ with $\alpha_j\geq d(j,i)$ is removed from $A$ (and connected to $j$). Also, each inactive client $j\in D-A$ with $bid(j,i)>0$ is reconnected to $i$.
\end{enumerate}
The intuition behind the bids is as follows. On the one hand, the active clients $j$ (which are not yet connected to an open facility) are willing to pay the difference $\alpha_j-d(j,i)$ (if positive) to open $i$, while reserving $d(j,i)$ for their own connection cost to $i$. On the other hand, the inactive clients $j\in D-A$ which are already connected to the closest facility in $S$, are willing to offer $d(j,S)-d(j,i)$ (if positive) towards the opening of facility $i$: if then $i$ is actually opened, $j$ reconnects to $i$ which is closer (while altogether still spending $\alpha_j$ in total). This also motivates the term \emph{greedy}. 
Notice that at any point of time active clients $j$ satisfy $\alpha_j<{d}(j,S)$.

\paragraph{Counterexample for the Na\"ive Extension.}
\begin{figure}[H]
    \centering
    \begin{tikzpicture}[x=0.9cm, y=0.9cm, >=latex, every text node part/.style={align=center}]
        \node[fill=red, circle, minimum size=5pt, inner sep=0pt] (c1) at (0pt,0pt) {};
        \node[] () at (0pt, 10pt) {\small $1$};
        \node[fill=red, circle, minimum size=5pt, inner sep=0pt] (c2) at (50pt,0pt) {};
        \node[] () at (50pt, 10pt) {\small $2$};
        \node[fill=blue, minimum size=7pt, inner sep=0pt] (f1) at (25pt,0pt) {};
        \node[] () at (25pt, 10pt) {\small $i$};
        \node[fill=blue, minimum size=7pt, inner sep=0pt] (f2) at (250pt,0pt) {};
        \node[] () at (251.5pt, 10pt) {\small $i'$};
        \path[->, thick] (-30pt,-15pt) edge (290pt, -15pt);
        \path[] (0pt, -15pt) edge (0pt, -12pt);
        \node[] () at (-3pt, -23pt) {\small $-0.5$};
        \path[] (25pt, -15pt) edge (25pt, -12pt);
        \node[] () at (25pt, -23pt) {\small $0$};
        \path[] (50pt, -15pt) edge (50pt, -12pt);
        \node[] () at (50pt, -23pt) {\small $0.5$};
        \path[] (250pt, -15pt) edge (250pt, -12pt);
        \node[] () at (250pt, -23pt) {\small $\sqrt{f/2}$};
    \end{tikzpicture}
    \caption{\small Counterexample for the na\"ive variant of JMMSV algorithm. The example is on a line. 
    Red circles represent the key clients, while blue squares represent the key facilities. The facility opening cost is $f$.}
    \label{fig:counter-example}
\end{figure}
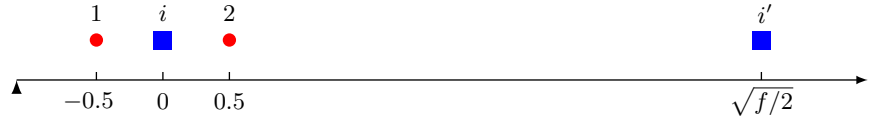

A na\"ive idea is to adapt the above algorithm to the case of squared distances by simply replacing $d(\cdot,\cdot)$ with $d^2(\cdot,\cdot)$. It is instructive to see why this attempt fails miserably (which also explains why prior work \cite{AhmadianNSW20,Cohen-AddadEMN22} use variants of the Jain-Vazirani algorithm \cite{JaV01} instead, despite the worse approximation factor). Consider the example in Figure \ref{fig:counter-example}, where we have two facilities $i$ and $i'$, and two clients $1$ and $2$, all on a line, with distances specified in the figure. The facility opening cost is $f$. Consider the following (hypothetical but valid) execution of the algorithm.

\begin{itemize}
\item Time $f/2+0.25-\eps$: facility $i'$ is open by client $2$ and other clients not in the figure. 
\begin{itemize}
\item At that time, client $1$ and $2$'s bid to $i$ was $f/2 - \eps$ each (so $i$ was very close to be opened), but after the opening of $i'$, client $2$ lowers its bid to $i$ to $(\sqrt{f/2}-0.5)^2 - 0.5^2= f/2-\sqrt{f/2}$. 
\end{itemize}
\item Time $f/2+\sqrt{f/2} + 0.25$: facility $i$ is open as the bid from client $1$ becomes $f/2+\sqrt{f/2} + 0.25 - 0.25 = f/2 + \sqrt{f/2}$. 
\item To prove an LMP $\Gamma$-approximation, the dual constraint corresponding to a solution where $1$ and $2$ are connected to $i$ (paying $f$ for opening and $2 \cdot 0.5^2 = 0.5$ for connecting) requires that $\alpha_1+\alpha_2 \leq f + 0.5\Gamma$. 
Our $\alpha_1+\alpha_2$ is $f + \sqrt{f/2} + 0.5-\eps$, which is only an LMP $\Omega(\sqrt{f})$-approximation. 
\end{itemize}
In this example, client $2$ contributes to the opening of some center far away from it, with a contribution much less than its $\alpha$ value. Or equivalently, the connection cost $d^2(2,i')$ is strictly but only slightly smaller than $\alpha_2$. This is the key (and essentially the only) reason why this na\"ive variant of JMMSV does not work; one can see that any other assumption between $\alpha_2$ and $d^2(2, i')$ in the above example (realized by moving $i'$ and clients not in the figure)
would lead to a good LMP approximation as follows.

\begin{itemize}
\item If $d^2(2,i')=\alpha_2$, 
then $2$'s bid to $i$ is $\alpha_2 - 0.25$ even after the opening of $i'$, so the fact that $i$ was open at time $\alpha_1-\eps$ (for infinitesimally small $\eps > 0$) implies $\alpha_1+\alpha_2 \leq f + 0.5$.
\item If $d^2(2,i')\leq (1-\delta)\alpha_2$ for some constant $\delta > 0$, we could simply use an approximate version of the triangle inequality to prove a good LMP approximation. 
For example, if $\delta=1/2$, we would have $\alpha_1 \leq d^2(1,i') \leq 2d^2(2,i') + 4(d^2(1,i)+d^2(2,i)) \leq \alpha_2 + 2$. Since $\alpha_2 \leq f/2+0.25$ in this example, this would give us $\alpha_1+\alpha_2 = f + O(1)$. 
\end{itemize}

\paragraph{Our Algorithm.}
    This inspired us to construct a different variant of the JMMSV algorithm in the following spirit: if a client can contribute to the opening of a facility, its $\alpha$ value must be much higher than its connection cost to the facility (e.g., at least twice as high as it). In more detail, our LMP algorithm uses the following different bidding strategy. Let $\gamma>1$ be a parameter to be fixed later. The bid of active clients is $[\alpha_j-\gamma d^2(j,i)]^+$. In particular notice that $j$ cannot contribute to the opening of $i$ if $\alpha_j<\gamma d^2(j,i)$ (while one might naturally expect it for $\alpha_j>d^2(j,i)$). However, whenever $\alpha_j\geq d^2(j,i)$ for some already open facility $i\in S$ and active clients $j$, we still connect $j$ to $i$ and make $j$ inactive. In the above case, if $\gamma d^2(j,i)>\alpha_j$,  
we say that $j$ is \emph{indirectly connected} to $i$ since $j$ did not contribute to the opening of $i$ with a positive bid. The inactive clients $j$ that are connected to an open facility $i$ whose opening they contributed to, are called \emph{directly connected}. The bids for directly and indirectly connected clients in the later steps are different: they are $[\alpha_j-\gamma d^2(j,i)]^+$ for the indirectly connected ones and $[\gamma d^2(j,S)-\gamma d^2(j,i)]^+$ for the directly connected ones. When an indirectly connected client contributes to the opening of a facility, it becomes directly connected. For the above process, we are able to show that the final solution $S$ satisfies:
$$
\sum_{j\in D}\alpha_j\geq f|S|+\sum_{j\in D}d^2(j,S),
$$
i.e., the cost of the solution is upper bounded by the sum of the variables $\alpha_j$. Furthermore, we can establish the following inequality showing that $\alpha_j$'s, after scaling, form a dual-feasible solution. 
$$
\sum_{j\in D} \left[\alpha_j
-
\left(\gamma+2+\frac{2}{\gamma-1}\right)d^2(j,i)\right]^+
\leq f,\quad \forall i\in F.
$$
The above equations together imply that the overall algorithm is an LMP $(\gamma+2+\frac{2}{\gamma-1})$-approximation for the considered facility location problem.
Fixing $\gamma=1+\sqrt{2}$ gives the claim. 

\paragraph{Opening $k$ Centers.}
With the above LMP approximation for facility location at hand, we can now derive the result in Theorem \ref{thr:mainBicriteria} following an approach again close to \cite{CGLSS25stoc}. A standard way to derive a $\rho_{kMeans}$-approximation algorithm for $k$-Means from an LMP $\rho_{FL}$-approximation for facility location with squared metric connection costs is to perform a binary search over the uniform facility cost $f$ so as to obtain a solution with $k_1<k$ open facilities for some facility cost $f_1$, and a solution with $k_2>k$ open facilities for some facility cost $f_2<f_1$, where the difference between $f_1$ and $f_2$ is very tiny. This is called a \emph{bi-point} solution. Then the two solutions are combined together to obtain a solution opening $k$ facilities. This combination is typically costly, hence leading to a factor $\rho_{kMeans}$ substantially larger than $\rho_{FL}$\footnote{An exception is the approach in \cite{AhmadianNSW20} which loses only a factor $(1+\eps)$. However their approach is tailored to the specific LMP algorithms considered in that work and hard to adapt to others, as the authors openly admit.}. We instead carefully modify the mentioned LMP $(3+2\sqrt{2})$-approximation for facility location (while increasing by $\eps$ the approximation factor) so that the variables $\alpha_j$ are increased in a logarithmic number of rounds only (with a small multiplicative increase at each round). Exploiting this fact and the walking-between-solutions framework in \cite{CGLSS25stoc}, we are able to construct a bi-point solution where $k_2\leq k+O(\log n/\eps^2)$. This leads to Theorem \ref{thr:mainBicriteria}. Again, the details are rather technical but they are similar in spirit to \cite{CGLSS25stoc}.

\section{Preliminaries}
\label{sec:preliminaries}

\paragraph{LP Relaxations and Basic Assumptions.}
Given $(D, F, d)$, the standard LP relaxations for $k$-Means and facility location with squared distances and uniform opening cost $f$ are as follows.

\begin{multicols}{2}
\begin{align*}
\min \quad &  \sum_{i\in F,j\in D}d^2(i,j)x_{i,j} & (LP_{km})\\
\mbox{s.t.} \quad &  \sum_{i\in F}x_{i,j}\geq 1 & \forall j\in D\\
& y_i - x_{i,j}\geq 0 & \forall j\in D,i\in F\\
&  \sum_{i\in F} y_i \leq k  \\
& x, y \geq 0.
\end{align*}
\columnbreak
\begin{align*}
\vspace{0.5em} \\
\min \quad & \sum_{i\in F,j\in D}d^2(i,j)x_{i,j} + f\cdot \sum_{i\in F}y_i & (LP_{FL}(f))\\
\mbox{s.t.} \quad & \sum_{i\in F}x_{i,j}\geq 1  & \forall j\in D \\
&  y_i - x_{i,j}\geq 0 & \forall j\in D,i\in F \\
& x, y \geq 0.
\end{align*}
\end{multicols}
\vspace{-1.5em}
Let $\sopt_{LP}(k)$ and $\sopt_{LP}(f)$ be the optimal values for $LP_{km}$ and $LP_{FL}(f)$ respectively. 
{Recall that} $[a]^+ := \max(a, 0)$. 
The dual for $LP_{FL}(f)$ is:
\begin{align*}
\max \quad & \sum_{j\in D}\alpha_j & (DP_{FL}(f))\\
\mbox{s.t.} \quad & \sum_{j\in D} [\alpha_j - d^2(i,j)]^+ \leq f & \forall i\in F \\
& \alpha \geq 0.
\end{align*}

Thanks to standard reductions, we can assume that  distances are integers in $[1,n^3/\eps]$ while loosing a factor $1+O(\eps)$ in the approximation.
\begin{lemma}
\label{lem:aspectratio}
For any constants $\eps > 0$ and $\alpha>1$, given a polynomial-time $\alpha$-approximation algorithm for $k$-Means on instances with distances in $\{1,\ldots,n^3/\eps\}$, there exists a polynomial-time $\alpha(1+O(\eps))$-approximation algorithm for $k$-Means on general instances. 
\end{lemma}
\begin{proof}
Consider any input instance $(\clients{,} \facilities,d)$ of $k$-Means. Assume that $n:=|\clients|$ is large enough w.r.t. $\alpha$, otherwise the problem can be solved by brute force in polynomial time.  
We guess $M:=\max_{j\in \clients}d(j,\opt)$, where $\opt$ is some optimal $k$-Means solution, by trying all the polynomially-many possibilities. If $M=0$, the problem can be solved optimally in polynomial time, hence assume $M>0$. Next, for each $(i,j)\in \facilities\times \clients$, define  $d'(j,i)=d'(i,j)=\max\{1,\lceil\frac{d(i,j)}{M}\frac{n}{\eps}\rceil\}$ if $d(i,j)\leq M$. Set all the remaining $d'(a,b)$, $a\neq b$, to $\frac{n^3}{\eps}$. Finally, replace the $d'(i,j)$'s with the corresponding metric closure. We run the given $\alpha$-approximation algorithm on the $k$-Means instance $(D{,} F,d')$, hence obtaining a solution $S$. It is easy to check that $S$ is a good enough approximation for the input instance.  
\end{proof}

\paragraph{Triangle Inequalities for Squared Distances.}
Given three distances $d,d_1,d_2$ with $d\leq d_1+d_2$, the following lemma provides a useful upper bound on $d^2$ as a function of $d_1$ and $d_2$. The next lemma provides a similar upper bound in an analogous setting where $d\leq d_1+d_2+d_3$.
\begin{lemma}\label{lem:apxTriangleInequality2}
 Let $\gamma>1$ be any constant. For any $x,y$ we have $\gamma x^2+\frac{\gamma}{\gamma-1}y^2\geq (x+y)^2$.    
\end{lemma}
\begin{proof}
Let us show for which values of $b$ the inequality $\gamma x^2+by^2\geq (x+y)^2$ holds. This is equivalent to requiring that the quadratic form $(\gamma-1) x^2+(b-1)y^2-2xy$ is non negative. This happens iff $\gamma-1\geq 0$ (which is satisfied by assumption) and $(\gamma-1)(b-1)-1\geq 0$, which is satisfied for $b\geq \frac{\gamma}{\gamma-1}$.
\end{proof}
\begin{lemma}\label{lem:apxTriangleInequality3}
 Let $\gamma>1$ be any constant. For any $x$, $y$, $z$, we have $\gamma x^2+\left(2+\frac{2}{\gamma-1} \right)(y^2+z^2)\geq (x+y+z)^2$.    
\end{lemma}
\begin{proof}
The proof of this lemma consists of a rearrangement of the LHS and applying the inequality of arithmetic and geometric means: 
\begin{alignat*}{5}
\gamma\cdot x^2+\left(2+\frac{2}{\gamma-1}\right)(y^2+z^2) & =x^2+y^2+z^2&&+\left(\frac{\gamma-1}{2}x^2+\frac{2}{\gamma-1}y^2\right) \\
&&&+\left(\frac{\gamma-1}{2}x^2+\frac{2}{\gamma-1}z^2\right)+(y^2+z^2)\\
& \geq x^2+y^2+z^2&&+2xy+2xz+2yz 
\\
&=(x+y+z)^2&&. 
\end{alignat*}
{The inequality follows from the arithmetic-geometric mean inequality $a^2+b^2\geq 2ab$ applied three times.}
\end{proof}

\section{Greedy Algorithm}
\label{subsec:modifiedJMS}

We formally present the new greedy algorithm for facility location with squared distances outlined in \Cref{sec:overview}. 
Let $\gamma=1+\sqrt{2}$ and $\Gamma = \gamma+2+\frac{2}{\gamma-1}\approx 5.828$ be the desired approximation ratio.
Given an instance $(D, F, d, f)$ for facility location, let $\hat f=\Gamma f$. The following is our new greedy algorithm\footnote{{As standard for primal-dual-based algorithms, we describe the algorithm in a more intuitive continuous version. It is easy to convert it into an equivalent discrete algorithm.}} achieving an LMP $\Gamma$-approximation {(see Algorithm \ref{fig:algGreedy})}. 
\begin{figure}[h!]
\begin{center}
\begin{minipage}{1.0\textwidth}
\begin{mdframed}[hidealllines=true, backgroundcolor=gray!15]

\begin{algorithm}[\JMSalg] \ \\[0.2cm]

\textbf{{Initialization:}} Set $S\leftarrow \emptyset$ and $\alpha_j\leftarrow 0$ for every $j\in D$. Let $A \leftarrow D, IC \leftarrow \emptyset, DC \leftarrow \emptyset$\\

{While $A \neq \emptyset$:}\\[-0.5cm]
\begin{quote}
Increase the $\alpha$-value of every $j \in A$ uniformly, until the following holds:
\begin{quote}
    (1) There exists an unopened facility $i\not\in S$ such that
    $$\sum_{j\in A \cup IC}[\alpha_j - \gamma d^2(i, j)]^+ + \sum_{j\in DC}[\gamma d^2(j,S)-\gamma d^2(i,j)]^+ \geq \hat f.$$
    If such event occurs, open $i$, i.e. add it to $S$.
\end{quote}
\begin{quote}
    (2) Some $j \in A$ has $\alpha_j \geq d^2(j, S)$.
\end{quote}
\end{quote}
\begin{quote}
Update $A, IC, DC$ so that
\begin{itemize}
    \item $A=\{j\in D: \alpha_j < d^2(j, S)\}$.
    \item $IC=\{j\in D: d^2(j, S)\leq \alpha_j < \gamma d^2(j, S)\}$.
    \item $DC=\{j\in D: \gamma d^2(j, S) \leq \alpha_j\}$.
\end{itemize}
\end{quote}\label{fig:algGreedy}
\end{algorithm}
\end{mdframed}
\end{minipage}
\end{center}
\end{figure}

We say that the clients in $A$ are \emph{active}, that the clients in $IC$ are \emph{indirectly connected}, and that the clients in $DC$ are \emph{directly connected}. 
The following standard lemma proves that these $\alpha$ values are enough to {\em pay} the connection costs and the (augmented) opening costs. 

\begin{lemma}[Approximation Guarantee]
At the end of the execution of \JMSalg, we have $\sum_{j\in D}\alpha_j\geq \sum_{j\in D}d^2(j, S) + |S|\hat f$.
\label{lem:approx-original}
\end{lemma}
\begin{proof}
    {It is sufficient to show that}, at the end of the execution, we have
    \begin{gather}
        \sum_{j\in DC}\alpha_j \geq \sum_{j\in DC}\gamma d^2(j, S) + \sum_{i\in S}\hat f\,.
        \label{ineq:approxgamma}
    \end{gather}
    Indeed, at the end of the execution, $A=\emptyset$, and for every $j\in IC$, $\alpha_j\geq d^2(j, S)$. Furthermore, $\gamma>1$, therefore, \eqref{ineq:approxgamma} implies the lemma. To this end, we prove that at any point in the execution of \JMSalg, \eqref{ineq:approxgamma} holds.
    
    The equality is initially true since $DC=\emptyset$ and $S = \emptyset$. Furthermore, since whenever the $\alpha$-value of a client $j$ reaches $d^2(j, S)$ it stops growing, no client is added to $DC$ outside of when a facility is opened.
    
    We now consider what happens when we open a facility $i$, i.e., add it to $S$. 
    Let $(\alpha,S, A, IC, DC)$ be the state right before opening $i$, and set $DC' = \{ j \in A\cup IC : \alpha_j \geq \gamma d^2(i, j) \}$, and $X=\{j\in DC:d^2(j, i)<d^2(j,S)\}$.
     The change of cost of the right-hand side of~\eqref{ineq:approxgamma} is at most 
    \begin{align*}
     \hat f &+   \sum_{j\in DC'} \gamma d^2(i,j) + \sum_{j \in X} \gamma (d^2(i,j) - d^2(j, S))\,.
    \end{align*}
    Since the algorithm decided to open $i$, we also have
    \begin{align*}
      \hat f &\leq \sum_{j\in DC'}(\alpha_j - \gamma d^2(i, j)) + \sum_{j\in X}(\gamma d^2(j,S)-\gamma d^2(i,j))\,.
    \end{align*}
    We thus get that the change of cost of the right-hand side is at most $\sum_{j\in DC'} \alpha_j$, which is the change of the left-hand side.
\end{proof}



Let $(\alpha_j)_{j\in D}$ be the final vector of $\alpha$ values. 
We following lemma {is the core of our analysis: it} implies that $\alpha / \Gamma$ is a feasible solution for $DP_{FL}(f)$. Combined with \Cref{lem:approx-original}, {it shows that} our solution $S$ satisfies 
\[
\sum_{j \in D} d^2(j, S) + \Gamma |S| f \leq \sum_{j \in D} \alpha_j \leq \Gamma \cdot \sopt_{LP}(f) \leq \Gamma \cdot \sopt_{FL}(f), 
\]
which finishes the proof of the LMP $\Gamma$-approximation.

\begin{lemma}[Dual Feasibility]
For every facility $i$, we have $\sum_{j\in D} [\alpha_j - \Gamma d^{2}(i,j)]^+ \leq \hat{f}$. 
\label{lem:dual-feasibility-original}
\end{lemma}
\begin{proof}
Consider a facility $i$. Let $D^*=\{j\in D: \alpha_j> \Gamma d^2(i, j)\}$.
Since $\alpha_j \leq \Gamma d^{2}(i,j)$ for any $j\notin D^*$, 
it remains to show that 
\begin{align}
    \sum_{j\in D^*} \alpha_j \leq \hat f + \Gamma\cdot \sum_{j\in D^*} d^{2}(i,j)~. \label{eqn:dual-feasible-equiv}
\end{align}
Let $s=|D^*|$ be the size of $D^*$. Without loss of generality, we assume that $D^*=[s]$ and it is sorted in non-decreasing order of $\alpha_j$ values, namely, $\alpha_1\leq \ldots \leq \alpha_s$. 

For each $j\in D^*$, let $DC^j$ (resp., $IC^j$) be the set of clients $j'$ in $[j-1]$ such that $j'\in DC$ (resp., $j'\in IC\cup A$) at time $\alpha_j-\eps$, where the value of $\eps>0$ can be an arbitrary positive number such that $\alpha_j-\eps > \Gamma\cdot d^2(j,i)$ and such that no client becomes connected or no facility becomes open between $\alpha_j-\eps$ (inclusive) and $\alpha_j$ (exclusive). 
Similarly, let $S^j$ be the set $S$ at time $\alpha_j-\eps$.
Then, for any $j\in D^*$ and $j'\in DC^j$, the squared distance between $j$ and the facility $i'$ that $j'$ connects to at time $\alpha_j-\eps$ should be strictly larger than $\alpha_j-\eps$. 
This is because otherwise we will stop growing $\alpha_j$ at time $\alpha_j-\eps$. 
Since $\eps$ here can be arbitrarily small, we have $d^2(j,i') \geq \alpha_j$. 
Further, applying \cref{lem:apxTriangleInequality3}, we have
\begin{align}
    \forall j\in [s], j'\in DC^j, \quad \alpha_j &\leq d^2(j,i')  \leq \left(d(j,i) + d(j',i) + d(j',i')\right)^2
    \notag
    \\
    & \leq \gamma \cdot d^2(j',i') + \left(2+\frac{2}{\gamma-1}\right)\cdot (d^2(j,i) + d^2(j',i))
    \notag
    \\
    & {=} \gamma \cdot d^2(j',S^j) + \left(2+\frac{2}{\gamma-1}\right)\cdot (d^2(j,i) + d^2(j',i))\,.
    \label{eqn:bound-of-alpha_i}
\end{align}
{In the equality we used that $d(j',i')=d(j',S^j)$ by the definition of $i'$ and $S^j$.}

Note that $\alpha_j-\eps > \Gamma\cdot d^2(j,i) > \gamma \cdot d^2(j,i)$ for any $j\in D^*$.
At time $\alpha_j-\eps$, facility $i$ is unopened, because otherwise client $j$ should be directly connected to $i$.
Therefore, we have
\begin{align*}
    \sum_{j'\in DC^j} \left[\gamma d^2(j',S^i) - \gamma d^2(j',i)\right]^+
    + \sum_{j'\in IC^j} \left[\alpha_{j'} - \gamma d^2(j',i)\right]^+ 
    + \sum_{j\leq j'\leq s} \left[\alpha_{j} - \eps - \gamma d^2(j',i)\right]^+ < \hat f~.
\end{align*}
Again, since this $\eps$ can be arbitrarily small and $[v]^+\geq v$ for any value of $v$, for any $j\in D^*$, we have the following inequality:
\begin{align*}
    \sum_{j'\in DC^j} \left(\gamma d^2(j',S^i) - \gamma d^2(j',i)\right)
    + \sum_{j'\in IC^j} \left(\alpha_{j'} - \gamma d^2(j',i)\right) 
    + \sum_{j\leq j'\leq s} \left(\alpha_{j} - \gamma d^2(j',i)\right) \leq \hat f~.
\end{align*}
By applying \cref{eqn:bound-of-alpha_i}, we get
\begin{alignat*}{1}
    \sum_{j'\in DC^j} \left(\alpha_j - \left(2+\frac{2}{\gamma-1}\right) \cdot \left(d^2(j,i)+d^2(j',i)\right)  - \gamma d^2(j',i)\right) &+ \sum_{j'\in IC^j} (\alpha_{j'} - \gamma d^2(j',i)) 
    \\
    &+ \sum_{j\leq j'\leq s} (\alpha_j - \gamma d^2(j',i)) \leq \hat f~,
\end{alignat*}
and we can simplify it as (using the definition of $\Gamma := \gamma + 2 + \frac{2}{\gamma-1}$)
\begin{align}
    \label{eqn:lambda-j}
    & (s-j+1+|DC^j|) \cdot \alpha_j + \sum_{j'\in IC^j} \alpha_{j'} \leq \hat f + \gamma \cdot \sum_{j'\in D^*} d^2(j',i) + \left(\Gamma-\gamma\right) \cdot \sum_{j'\in DC^j} (d^2(j,i)+d^2(j',i)) ~. 
    \tag{$\text{UB}_j$}
\end{align}

Next, we show that all \cref{eqn:lambda-j} (for $j\in D^*$) together imply our objective \cref{eqn:dual-feasible-equiv}.
Let $A_{jj'}$ define the coefficient of $\alpha_{j'}$ in \cref{eqn:lambda-j}. That is, we have 
    \begin{align*}
        A_{jj'}=\begin{cases}
            \ind(j'\in IC^j) & \text{if } j'<j;\\
            s-j+1+|DC^j| & \text{if }j'=j;\\
            0  & \text{if }j'>j.
        \end{cases}
    \end{align*}
    We can rewrite \cref{eqn:lambda-j} as 
    \begin{align*}
        \label{eqn:lambda-j_bis}
        \sum_{j'\in [s]} A_{jj'} \alpha_{j'} \leq \hat f + \gamma \cdot \sum_{j'\in [s]} d^2(j',i)
        + \left(\Gamma-\gamma \right) \cdot \sum_{j'\in [j-1]} (1-A_{jj'})(d^2(j,i) + d^2(j',i)).\tag{$\text{UB}'_j$}
    \end{align*}
    Let $\beta_1,\ldots,\beta_s$ be some coefficients such that $\sum_{j\in [s]} \beta_j \cdot \sum_{j'\in [s]} A_{jj'}\alpha_{j'} = \sum_{j\in [s]} \alpha_j$, {for every possible value of the quantities $\alpha_j$}. {In other words, we impose}
    \begin{equation}\label{eqn:formulabetaj}    
{\forall j\in [s]: 1 = \sum_{j'\in [s]}\beta_{j'}A_{j'j}=\sum_{j'\geq j}\beta_{j'}A_{j'j},}
    \end{equation}
{where the second equality comes from the fact} that $A_{j'j}=0$ for any $j>j'$. Thus we can easily compute each $\beta_j$ one by one, using the values of $\beta_{j+1}, \dots, \beta_s$:
    \begin{align}
        \label{eqn:recurrence-for-beta}
        \beta_j = \frac{1}{A_{jj}} \cdot \left(1 - \sum_{j'>j} \beta_{j'} A_{j'j}\right)~.
    \end{align}   
{Let us show that} $\beta_{j}\geq 0$ {for all $j\in [s]$}. {Using \cref{eqn:recurrence-for-beta} and} since $A_{j'j}\in \{0,1\}$ for any $j'>j$, it suffices to show that $\sum_{j'>j} \beta_{j'}\leq 1$.
    Indeed, the following stronger fact holds.
    \begin{fact}
        \label{fact:ub-for-sum-beta}
        Fix any $j\in \{0\} \cup [s]$. We have $\sum_{j'>j} \beta_{j'} \leq \frac{s-j}{s-j+|DC^{j}|}$, where we define $DC^0:=\emptyset$.
    \end{fact}
    \begin{proof}
Using \cref{eqn:formulabetaj}, we get
\begin{align}
        \label{eqn:sum-of-A*beta}
        {s-j=\sum_{\ell>j}\sum_{j'\in [s]}\beta_{j'}A_{j'\ell}=\sum_{j'>j}\beta_{j'}\left(\sum_{\ell>j}A_{j'\ell}\right),}
\end{align}
{where in the last equality we used the fact that, for $j'\leq j<\ell$, $A_{j'\ell}=0$.}

    
    Recall the definition of $A_{j'\ell}$, where we have for any $j'>j$,
    \begin{align*}
        A_{j'j'} &= s-j'+1+|DC^{j'}| = s-j'+1 + \sum_{j<\ell<j'} \ind(\ell \in DC^{j'}) + \sum_{\ell\leq j} \ind(\ell \in DC^{j'})~,
        \\
        A_{j'\ell} &= \ind(\ell \in IC^{j'}) \hspace{10em} \forall j<\ell<j'~.
    \end{align*}
    For any $j<\ell<j'$, we have either $\ell\in DC^{j'}$ or $\ell \in IC^{j'}$ (note that $IC^{j'}$ includes {active}
    clients with index less than $j'$ at time $\alpha_{j'}-\eps$). {Thus, $\sum_{j<\ell<j'}(\ind(\ell \in DC^{j'})+\ind(\ell \in IC^{j'}))=(j'-1)-j$.}    
    Hence, for any $j'>j$, we have 
    \[
        \sum_{\ell>j} A_{j'\ell} = s-j'+1 + (j'-1)-j + \sum_{\ell\leq j} \ind(\ell\in DC^{j'}) \geq s-j+|DC^{j}| ~,
    \]
    where the inequality results from the fact that $DC^{j'} \cap[j-1] \supseteq DC^{j}$ for $j'>j$. 
    Applying this lower bound to \cref{eqn:sum-of-A*beta}, we get $s-j\geq (s-j+|DC^j|)\sum_{j'>j} \beta_{j'}$, which is equivalent to $\sum_{j'>j} \beta_{j'}\leq \frac{s-j}{s-j+|DC^j|}$. 
    \end{proof}

    Hence, we have $\beta_j\geq 0$ for any $j\in [s]$ according to our earlier discussions. Further, we have 
    \begin{align}
        & \sum_{j\in [s]} \alpha_j = \sum_{j\in [s]} \beta_j \cdot \left(\sum_{j'\in [s]} A_{jj'}\alpha_{j'}\right)
        \notag
        \\
        \overset{\cref{eqn:lambda-j_bis}}{\leq} & \left(\hat f + \gamma\cdot \sum_{j\in [s]} d^2(j,i)\right)  \left(\sum_{j\in [s]} \beta_j\right) + \left(\Gamma-\gamma\right) \left(\sum_{j\in [s]} \beta_j \cdot \sum_{j'\in [j-1]} (1-A_{jj'})(d^2(j,i) + d^2(j',i))\right)
        \notag
        \\
        \overset{\cref{fact:ub-for-sum-beta}}{\leq} & \left(\hat f + \gamma\cdot \sum_{j\in [s]} d^2(j,i)\right) + \left(\Gamma-\gamma\right) \left(\sum_{j\in [s]} \beta_j \cdot \sum_{j'\in [j-1]} (1-A_{jj'})(d^2(j,i) + d^2(j',i))\right)\ .
        \label{eqn:sum-of-alpha}
    \end{align}
    Finally, to finish the proof, we show that 
    \begin{align}
       \sum_{j\in [s]} \beta_j \cdot \sum_{j'\in [j-1]} (1-A_{jj'})(d^2(j,i) + d^2(j',i)) \leq \sum_{j\in [s]} d^2(j,i)~.
        \label{eqn:additional-term-in-sum-of-alpha}
    \end{align}
    Rearranging the LHS of the inequality, we get 
    \begin{align*}
        \text{LHS of \cref{eqn:additional-term-in-sum-of-alpha}} &= \sum_{j\in [s]} \sum_{j'\in [j-1]} \beta_j \cdot (1-A_{jj'}) \cdot d^2(j,i) + \sum_{j\in [s]} \sum_{j'\in [j-1]} \beta_{j}\cdot (1-A_{jj'})\cdot d^2(j',i)
        \\
        &= \sum_{j\in [s]} \left(\beta_j \cdot \sum_{j'\in [j-1]} (1-A_{jj'}) + \sum_{j'>j} \beta_{j'} \cdot (1-A_{j'j})\right) \cdot d^2(j,i)
        \\
        &\leq \sum_{j\in [s]} \left(\beta_j \cdot \sum_{j'\in [j-1]} (1-A_{jj'}) + \sum_{j'>j} \beta_{j'}\right) \cdot d^2(j,i)  
        \\
        &\overset{\cref{fact:ub-for-sum-beta}}{\leq} \sum_{j\in [s]} \left(\beta_j \cdot \sum_{j'\in [j-1]} (1-A_{jj'}) + \frac{s-j}{s-j+|DC^j|} \right) \cdot d^2(j,i)\ ,
    \end{align*}
    {where in the first inequality we used the fact that $A_{j'j}\in \{0,1\}$ for $j'>j$.}
    By \cref{eqn:recurrence-for-beta}, we have $\beta_j\leq \frac{1}{A_{jj}} = \frac1{s-j+1+|DC^j|}$.
    Note that, for any $j'<j$, $A_{jj'}=0$ if and only if $j'\notin IC^j$, i.e., $j'\in DC^j$. 
    Therefore, $\sum_{j'\in [j-1]} (1-A_{jj'}) = |DC^j|$. 
    We can further upper bound the above LHS by
    \begin{align*}
        \text{LHS of \cref{eqn:additional-term-in-sum-of-alpha}} &\leq \sum_{j\in [s]} \left(\frac{|DC^j|}{s-j+1+|DC^j|} + \frac{s-j}{s-j+|DC^j|} \right) \cdot d^2(j,i)
        \\
        &\leq \sum_{j\in [s]} \left(\frac{|DC^j|}{s-j+|DC^j|} + \frac{s-j}{s-j+|DC^j|} \right) \cdot d^2(j,i) = \sum_{j\in [s]} d^2 (j,i)\ .
    \end{align*}
    Putting \cref{eqn:sum-of-alpha,,eqn:additional-term-in-sum-of-alpha} together, we get 
    \begin{align*}
        \sum_{j\in [s]} \alpha_j \leq \hat f + \gamma \cdot \sum_{j\in [s]} d^2(j,i) + \left(\Gamma-\gamma\right) \cdot \sum_{j\in [s]} d^2(j,i) = \hat f + \Gamma \cdot \sum_{i\in [s]} d^2(j,i)~. & \qedhere
    \end{align*}
\end{proof}

\section{$5.83$-Approximation with $O(\log n/\eps^2)$ Extra Centers}
Let $\gamma=1+\sqrt{2}$ and $\Gamma = \gamma+2+\frac{2}{\gamma-1}\approx 5.828$. In this section, we present a polynomial-time algorithm that opens $k + O(\log n/ \eps^2)$ many centers and
returns a solution of cost at most $(\Gamma + \eps)\sopt$, proving \Cref{thr:mainBicriteria}.

This section is close to Section 3 of~\cite{CGLSS25stoc} that used the classical LMP 2-approximation~\cite{JainMMSV03} for a $(2+\eps)$-approximation for $k$-Median. Here, we start from the LMP $\Gamma$-approximation (with squared distances) presented in \Cref{subsec:modifiedJMS} instead.
\Cref{subsec:log_adaptive} presents a modified algorithm \logadaptalg that runs in $O(\log n / \eps^2)$ adaptive steps (that we call {\em phases}) and \Cref{subsec:log_adaptivity_analysis} analyzes the algorithm. 
Finally, building on additional setups in \Cref{subsec:robust_analysis}, \Cref{subsec:merging_algorithm} presents our final algorithm \mergealg achieving the guarantee of \Cref{thr:mainBicriteria}. 

\subsection{Log Adaptive Algorithm}
\label{subsec:log_adaptive}

We present the modified greedy algorithm that works in $O(\log n / \eps^2)$ phases. 
{At high level}, instead of considering the continuous time where $\alpha$ values of the active clients are increased at unit rate, {the idea is to} consider discrete time steps (that we call {\em phases}), where the $i$th phase $(i \geq 0)$ corresponds to time $(1+\eps^2)^i$. Then, since we assumed that the minimum nonzero distance is $1$ and the maximum distance is $\poly(n)$, the number of phases is $O(\log(n)/\eps^2)$.

We will ensure that the dual feasibility (\Cref{lem:dual-feasibility-original}) will be satisfied without any slack, which could be violated if $\alpha$ values of the active clients discretely jump from one value to a strictly larger value. In order to deal with this issue, \cite{CGLSS25stoc} uses two {ideas}, namely (1) scaling down the metric by a factor of $(1-\delta)$ with $\delta = 3\eps$ and (2) allowing clients extremely close to a facility $i$ to continuously increase their $\alpha$ values in order to open $i$. We use a similar definition adapted to our new greedy algorithm for squared distances and its analysis. {A key ingredient is the following technical definition of facilities that can be opened.}
\begin{definition}[openable]
\label{def:openable}
A facility $i$ is \emph{openable} with respect to the algorithm's state $(\alpha,S,\theta)$ if there are increased dual values $(\tau_j)_{j\in A}$ of active clients that satisfy the following conditions:
\begin{itemize}
    \item Only nearby clients are increased:
    $$\tau_j=\alpha_j, \text{ for every } j\in A\setminus B(i,\sqrt{\eps\theta}).$$

    \item Nearby clients are only increased slightly:
    $$\alpha_j\leq\tau_j\leq \min\{(1-\delta)d^2(j,S), (1+\eps^2)\theta\}\text{ for every } j\in A\cap B(i,\sqrt{\eps\theta}).$$

    \item Facility $i$ is paid for:
    $$\sum_{j\in A}[\tau_j - (1-\delta)\gamma d^2(i, j)]^+ + \sum_{j\in IC}[\alpha_j - (1-\delta)\gamma d^2(i, j)]^+ +(1-\delta)\sum_{j\in DC}[\gamma d^2(j,S)-\gamma d^2(i,j)]^+ \geq \hat f$$

    \item Dual feasibility: for every client $k\in A$ and facility $i_0$
    \begin{align*}
    \sum_{j\in DC}\left[\tau_k - \left(2+\frac{2}{\gamma-1}\right)\cdot (d^2(i_0,k) + d^2(i_0,j))-\gamma d^2(i_0,j)\right]^+ &+ \sum_{j\in IC} [\alpha_{j}-\gamma d^2(i_0,j)]^+ \\
    &+ \sum_{j\in A}[\tau_{j}-\gamma d^2(i_0,j)]^+\leq \hat f
    \end{align*}
\end{itemize}
\end{definition}
The openability can be easily checked by solving a linear program. 
\begin{lemma}
    There is a polynomial-time algorithm that, given the state of the algorithm and a facility $i$, either outputs $(\tau_j)_{j \in A}$ that satisfy all conditions of openability or certifies that $i$ is not openable.
\end{lemma}
\begin{proof}
Fix $i$ and $(\alpha, S, A, IC, DC, \theta)$ and consider a linear program with variables $(\tau_j)_{j \in A}$. We claim that each of the four bullets in \Cref{def:openable} can be expressed as linear constraints. The first two bullets are immediate. Indeed, as $\tau_j = \alpha_j$ for every $j \in A - B(i, \eps \theta)$, we can treat them as constants and only have $(\tau_j)_{j \in A \cap B(i, \sqrt{\eps \theta})}$ as the true variables. 

Then, the third bullet, 
\begin{gather*}
    \sum_{j\in A} [\tau_j - (1-\delta)\gamma d^2(i,j)]^+ + \sum_{j\in IC\cup DC} [\min(\alpha_j, (1-\delta)\gamma d^2(j,S)) - (1-\delta)\gamma d^2(i,j)]^+ \geq \hat f\,
\end{gather*}
becomes a linear inequality in the true variables, since the only terms involving the true variables are $[\tau_j - (1 - \delta)\gamma d^2(i,j)]^+$ for $j \in A \cap B(i, \sqrt{\eps \theta})$, and for those $j$, 
$[\tau_j - (1 - \delta)\gamma d^2(i,j)]^+$ is indeed equal to $(\tau_j - (1 - \delta)\gamma d^2(i,j))$ without $[\cdot]^+$. 

Finally, the fourth bullet {can be similarly expressed as} a linear inequality in the true variables {as the third bullet}.
\end{proof}

\Cref{alg:logadaptive} is our log-adaptive algorithm. It works in $O(\log n/\eps^2)$ phases where, in the $i$th phase, every active client $j$ has $\alpha_j = (1+\eps^2)^{i}$ and the algorithm opens openable facilities in an arbitrary order by increasing $\alpha_j$'s to $\tau_j$'s. However, when $i$ is open, the only scenario $j$ with $\tau_j > \alpha_j$ is when $j \in B(i, \sqrt{\eps \theta})$, so $j$ {is} immediately put into $DC$ after $i$ is open.

\begin{figure}[h!]
\begin{center}
\begin{minipage}{1.0\textwidth}
\begin{mdframed}[hidealllines=true, backgroundcolor=gray!15]

\begin{algorithm}[\logadaptalg] \ \\[0.2cm]
\label{alg:logadaptive}
\noindent{\bf Initialization:} Set $\theta=1$, $S=\emptyset$, $A=D$, $DC=\emptyset$, $IC=\emptyset$, and $\alpha_j=\theta$ for every $j\in A=D$.\\

\noindent While $A\neq \emptyset$:

\begin{enumerate}
    \item While there is an unopened facility $i$ that is openable:
    \begin{itemize}
        \item Compute $(\tau_j)_{j\in A}$ so that $i$ is openable with $(\tau_j)_{j\in A}$.
        \item Add $i$ to $S$. Set $\alpha_j\leftarrow \tau_j$ for every $j\in A$.
        \item For every $j\in D$:
        \begin{itemize}
            \item If $\alpha_j\geq (1-\delta)\gamma d^2(j,S)$, move $j$ to $DC$.
            \item Otherwise, if $j\in A$ and $\alpha_j\geq (1-\delta)d^2(i,j)$, move $j$ to $IC$.
        \end{itemize}
    \end{itemize}

    \item Perform the following step and move to the next phase:
\begin{itemize}
\item For every remaining $j\in A$, set $\alpha_j\leftarrow \min\{(1+\eps^2)\theta, (1-\delta)d^2(j, S)\}$, and move it to $IC$ when $\alpha_j=(1-\delta)d^2(j,S)$. Update $\theta\leftarrow(1+\eps^2)\theta$.
\end{itemize}   
\end{enumerate}
\end{algorithm}
\end{mdframed}
\end{minipage}
\end{center}
\end{figure}

\subsection{Analysis}
\label{subsec:log_adaptivity_analysis}

In this subsection, we show that \Cref{alg:logadaptive} achieves an LMP ${\frac{\Gamma}{1-\delta}=}\Gamma + O(\eps)$ approximation. {This is a direct consequence of the following two lemmas, proved in \Cref{subsubsec:approximation-log} and \Cref{subsubsec:dual-feasibility-log}, resp.} Let $(\alpha^*_j)_{j \in D}$ be the final $\alpha$ values. 
\begin{lemma}
    We have $\sum_{j\in D}\alpha^*_j \geq \sum_{j\in D}(1-\delta)d^2(j,S) + \sum_{i\in S}\hat f$.
\label{lem:approximation-log}
\end{lemma}
\begin{lemma}
    For every facility $i$,
    $$\sum_{k\in D}[\alpha_k^* -\Gamma d^2(i, k)]^+ \leq \hat f.$$
\label{lem:dual-feasible-log}
\end{lemma}
{In more detail, Lemma \ref{lem:approximation-log} shows that the dual values $\frac{\alpha^*_j}{1-\delta}$ can pay the cost of the solution with facility opening cost $\frac{\Gamma}{1-\delta}f$.} Lemma \ref{lem:dual-feasible-log} implies that $\alpha^*/\Gamma$ is a feasible solution for the {dual of the} facility location LP with facility cost $f$. {We conclude as desired that
$$
\frac{\Gamma}{1-\delta}\optlpfl(f)\geq \sum_{j\in D}d^2(j,S)+\frac{\Gamma}{1-\delta}|S|f.
$$}


\subsubsection{Approximation Guarantee}
\label{subsubsec:approximation-log}

\begin{proof}[Proof of Lemma \ref{lem:approximation-log}]

{It is sufficient to show that, at any point of the algorithm, one has}
    \begin{gather}
        \sum_{j\in DC} \alpha_j \geq \sum_{j\in DC}(1-\delta)\gamma d^2(j, S)  + \sum_{i\in S} \hat f\,.
        \label{eq:IH_approx_guarantee}
    \end{gather}
{Indeed, at the end of the algorithm all the clients are in $IC$ or $DC$. For the ones} in $IC$, one trivially has ${\alpha^*_j}\geq (1-\delta)d^2(j, S)$. {For the remaining ones it is sufficient to apply the above inequality and the fact} that $\gamma\geq 1$.

We prove \eqref{eq:IH_approx_guarantee} by induction on the steps of the algorithm. The {in}equality is initially true since $A= D$ (thus $IC\cup DC = \emptyset$) and $S = \emptyset$. {For the inductive step,} we next analyze each {one} of the two stages separately. 

    \paragraph{Stage 1.} Consider what happens when we open a facility $i$, i.e., {we} add it to $S$. 
    Let $(\alpha, S, A, IC, DC, \theta)$ be the state right before opening $i$, and let $A'$ and $IC'$ be the clients that are moved to $DC$ from $A$ and $IC$ respectively by the opening, {i.e.,} $A' = \{ j \in A : \tau_j \geq (1 - \delta)\gamma d^2(i, j) \}$ and $IC'=\{j\in IC:\alpha_j\geq (1-\delta)\gamma d^2(i, j)\}$. {Let also} $X=\{j\in DC:d^2(j, i)<d^2(j,S)\}$.
     The change of cost of the right-hand side of~\eqref{eq:IH_approx_guarantee} is at most 
    \begin{align*}
     \hat f &+   \sum_{j\in A'\cup IC'} (1-\delta) \gamma d^2(i,j) + \sum_{j \in X}  (1 - \delta)\gamma (d^2(i,j) - d^2(j, S))\,.
    \end{align*}
    Since $i$ is {paid for} (the third bullet of \Cref{def:openable}), we also have
    \begin{align*}
      \hat f &\leq \sum_{j\in A'}(\tau_j - (1-\delta)\gamma d^2(i, j)) + \sum_{j\in IC'} (\alpha_j - (1-\delta)\gamma d^2(i, j)) + (1-\delta)\sum_{j\in X}(\gamma d^2(j,S)-\gamma d^2(i,j))\,.
    \end{align*}
{In the above inequality we used the fact that the only positive terms in the left-hand side of the definition of paid for correspond to the sets $A'$, $IC'$, and $X$.}    
    We thus get that the change of cost of the right-hand side is at most $\sum_{j\in A'} \tau_j+\sum_{j\in IC'}\alpha_j$, which is the change of the left-hand side.
    
    \paragraph{Stage 2.}
    No facility is open at this stage, and no client is added to $DC$ either.
\end{proof}

\subsubsection{Dual Feasibility}
\label{subsubsec:dual-feasibility-log}
Here we prove \Cref{lem:dual-feasible-log}. First, we prove the following lemma, showing {intuitively} that no facility is overpaid during the course of our algorithm. 

\begin{claim}[No Over-Bidding]
\label{clm:nooverbidding}
    At any point in the algorithm, for every facility $i_0$,
    $$\sum_{j\in DC}[{(1-\delta)}\gamma d^2(j, S) - \gamma d^2(i_0,j)]^+ + \sum_{j\in A\cup IC}[\alpha_j - \gamma d^2(i_0,j)]^+\leq \hat f.$$
\end{claim}
\begin{proof}
{Let us prove the claim by induction on the steps of the algorithm. The claim is initially true since all clients are in $A$ and the initial value $1$ of the variables $\alpha_j$ makes the left-hand side equal to $0$ (since we can impose w.l.o.g. that the minimum distance of a client to a facility is $1$ by \cref{lem:aspectratio}). For the inductive step, we consider the two stages separately:}

\paragraph{Stage 1:} Suppose that facility $i'$ is open, {and let} $(\tau_j)_{j\in A}$ {be the associated values}. {This implies that} (1) the $\alpha$-values of clients in $B(i', \sqrt{\eps\theta})\cap A$ are potentially increased, and (2) some clients might become direct{ly} connect{ed}. {Let us} see how the constraint is impacted by these changes. First, observe that a client that was directly connected before the opening of $i'$ can only lower its contribution. For a client $j$ that was active or indirectly connected and that now becomes direct{ly} connect{ed}, its contribution lowers from $[\alpha_j-\gamma d^2(i_0, j)]^+$ to $[{(1-\delta)}\gamma d^2(j,S) - \gamma d^2(i_0,j)]^+$, because $S$ includes $i'$.

\paragraph{Stage 2:} Assume towards contradiction that at the end of the phase, if we increase $\alpha_j$ of every $j\in A$ to $\min((1+\eps^2)\theta, (1-\delta)d^2(j,S))$, the claim is violated {by some facility $i$}. We select a "minimal" such counter-example in the following way: let $\tau'\leq (1+\eps^2)\theta$ be the smallest value so that $\alpha_j':=\min(\tau', (1-\delta)d^2(j,S))$ satisfies
$$\sum_{j\in A}[\alpha_j'-\gamma d^2(i, j)]^+ + \sum_{j\in IC}[\alpha_j - \gamma d^2(i,j)]^+ + \sum_{j\in DC}[(1-\delta)\gamma d^2(j,S)-\gamma d^2(i,j)]^+=\hat f,$$
for the considered $i$. We remark that $\tau'\geq \theta$ since the constraint is satisfied at the end of Phase 1. 
{We also observe that $i$ is not yet open. Indeed, otherwise}
there would be no client $j\in A$ for which $\alpha'_j-\gamma d^2(i,j)$ is strictly positive. Hence the increase in $\alpha$-values cannot cause $i$ to violate the constraint.

We now show that this $i$ with $\tau_j=\alpha_j'$ for $j\in A\cap B(i,\sqrt{\eps\theta})$ and $\tau_j=\alpha_j$ for $j\in A\setminus B(i,\sqrt{\eps\theta})$ satisfies the conditions of openability for $i$. In other words, $i$ is an openable facility (that was not yet opened), which contradicts the completion of the first stage. We verify the conditons of openability (\Cref{def:openable}) one-by-one. The first two bullets are satisfied by the definition of $(\tau_j)_{j\in A}$. For the third bullet, we first observe that, for any $j\in A\setminus B(i,\sqrt{\eps\theta})$, one has
\begin{equation}\label{eqn:clm:nooverbidding}
\alpha_j-(1-\delta)\gamma d^2(i,j)\geq (1+\eps^2)\alpha_j - \gamma d^2(i,j)\geq \alpha_j' - \gamma d^2(i,j)\ ,    
\end{equation}
where we use the facts that $d^2(i,j)\geq \eps\theta$, $\alpha_j=\theta$, and $\gamma\geq 1$. Then
\begin{align*}
    &\sum_{j\in A}[\tau_j - (1-\delta)\gamma d^2(i,j)]^+ + \sum_{j\in IC}[\alpha_j - (1-\delta)\gamma d^2(i,j)]^+ \\
    & \;\;\;\; + \sum_{j\in DC}[(1-\delta)\gamma d^2(j,S) - (1-\delta)\gamma d^2(i,j)]^+\\
    = & \sum_{j\in A\cap B(i,\sqrt{\eps\theta})}[\alpha_j' - (1-\delta)\gamma d^2(i,j)]^+ + \sum_{j\in A\setminus B(i,\sqrt{\eps\theta})}[\alpha_j - (1-\delta)\gamma d^2(i,j)]^+\\
     & \;\;\;\; + \sum_{j\in IC}[\alpha_j - (1-\delta)\gamma d^2(i,j)]^+ + \sum_{j\in DC}[(1-\delta)\gamma d^2(j,S)-(1-\delta)\gamma d^2(i,j)]^+\\
    \overset{\eqref{eqn:clm:nooverbidding}}{\geq} & \sum_{j\in A\cap B(i,\sqrt{\eps\theta})}[\alpha_j' - \gamma d^2(i,j)]^+ + \sum_{j\in (A\setminus B(i,\sqrt{\eps\theta}))}[\alpha'_j - \gamma d^2(i,j)]^+\\
    & \;\;\;\; + \sum_{j\in IC}[\alpha_j - \gamma d^2(i,j)]^+ + \sum_{j\in DC}[(1-\delta)\gamma d^2(j,S)-\gamma d^2(i,j)]^+\\
    \geq & \hat f
\end{align*}

It remains to verify the fourth bullet of openability. Suppose toward contradiction that there is a facility $i_0$ and $k\in A$ such that
\begin{align*}
    & \sum_{j\in DC}\left[\tau_k - \left(2+\frac{2}{\gamma-1}\right)\cdot (d^2(i_0,k) + d^2(i_0,j))-\gamma d^2(i_0,j)\right]^+ \\
    & \hspace{3cm}+ \sum_{j\in IC} [\alpha_{j}-\gamma d^2(i_0,j)]^+ + \sum_{j\in A}[\tau_{j}-\gamma d^2(i_0,j)]^+ > \hat f.
\end{align*}

Let $j\in DC$ and $i_1\in S$ such that $j$ connects to $i_1$. By \Cref{lem:apxTriangleInequality3},
\begin{equation}\label{eqn:clm:nooverbidding2}
\gamma d^2(i_1,j)+ \left(2+\frac{2}{\gamma-1}\right) (d^2(i_0,k) + d^2(i_0,j)) \geq (d(i_1,j)+d(i_0,k) + d(i_0,j))^2\ .  
\end{equation}
Then
\begin{align*}
    &(1-\delta)\gamma d^2(j, S) - \gamma d^2(i_0,j)\\
    &=(1-\delta)\gamma d^2(i_1,j) - \gamma d^2(i_0,j)\\
    & \overset{\eqref{eqn:clm:nooverbidding2}}{\geq} (1-\delta)\left((d(i_1,j) + d(i_0,k) + d(i_0,j))^2 - \left(2+\frac{2}{\gamma-1}\right)\cdot (d^2(i_0,k) + d^2(i_0,j))\right)- \gamma d^2(i_0,j)\\
    &\geq (1-\delta)d^2(i_1,k) - (1-\delta)\left(2+\frac{2}{\gamma-1}\right)\cdot (d^2(i_0,k) + d^2(i_0,j))- \gamma d^2(i_0,j)\\
    &\geq \tau_k - \left(2+\frac{2}{\gamma-1}\right)\cdot (d^2(i_0,k) + d^2(i_0,j))- \gamma d^2(i_0,j)\ ,
\end{align*}
{where in the last inequality we used} $\tau_k\leq (1-\delta)d^2(k,S)$.

Hence, we have:
\begin{align*}
&\sum_{j\in A}[\alpha_{j}'-\gamma d^2(i_0,j)]^+ + \sum_{j\in IC}[\alpha_{j} - \gamma d^2(i_0,j)]^+ + \sum_{j\in DC}[(1-\delta)\gamma d^2(j,S)-\gamma d^2(i_0,j)]^+\\
\geq & \sum_{j\in A}[\tau_{j}-\gamma d^2(i_0,j)]^+ + \sum_{j\in IC}[\alpha_{j} - \gamma d^2(i_0,j)]^+ \\
& \hspace{3cm}+ \sum_{j\in DC}\left[\tau_k - \left(2+\frac{2}{\gamma-1}\right)\cdot (d^2(i_0,k) + d^2(i_0,j))- \gamma d^2(i_0,j)\right]^+\\
> & \hat f
\end{align*}
which contradicts the minimality of $\tau'$.
\end{proof}

Let $(\alpha_j^*)_j$ be the final vector of $\alpha$-values. We say that {a} facility $i$ is frozen {(by $i_0$)} if there exists some open facility $i_0$ such that $d(i,i_0)\leq (1+\sqrt{\Gamma\eps}+3\eps)\sqrt{\theta/\Gamma}$. {Notice that a facility $i$ can be frozen by $i$ itself.} {We next focus on a specific facility $i$, with the goal of proving Lemma \ref{lem:dual-feasible-log}.}
Let $D^*=\{j:\alpha_j^*>\Gamma d^2(i,j)\}$.
Similarly to the proof of \Cref{lem:dual-feasibility-original}, we {define the set $D^*$ for $i$ and} assume w.l.o.g. that $D^*=[s]$
and that {$\alpha^*_j\leq \alpha^*_{j+1}$ for $j=1,\ldots,s-1$.}

Next, for each client $j\in [s]$:
\begin{itemize}
    \item If $j$ becomes inactive strictly before $i$ is frozen, we define $DC^j$ as the set of clients in $[j-1]$ that are directly connected at the beginning of the stage $j$ becomes inactive, $IC^j = [j-1]\setminus DC^j$, and $S^j$ as the set of opened facilities at the beginning of the stage $j$ becomes inactive.
    \item Else, we define $DC^j$ as the set of clients in $[j-1]$ that are directly connected right before $j$ becomes inactive, $IC^j = [j-1] \setminus DC^j$, and $S^j$ as the set of opened facilities right before $j$ becomes inactive.
\end{itemize}

\begin{claim}
\label{clm:notfrozen}
    For any $k\in D^*$ that becomes inactive strictly before $i$ becomes frozen, 
    \begin{align*}
    & \sum_{j\in {DC^k}}\left(\alpha_k^*-\left(2+\frac{2}{\gamma-1}\right)(d^2(i,k)+d^2(i,j))-\gamma d^2(i,j)\right)\\
    & \hspace{2cm}+\sum_{j\in {IC^k}}(\alpha_j^*-\gamma d^2(i,j)) 
    + \sum_{{k\leq j\leq s}}(\alpha_k^*-\gamma d^2(i,j))\leq \hat f.
    \end{align*}
\end{claim}
\begin{proof}
{We distinguish the following cases} depending on the stage {when} $k$ becomes inactive.

\paragraph{Stage 1.} $k$ becomes inactive because of the opening of some ${i'}$. It suffices to handle the case that ${i'}$ does not freeze $i$. The opening of ${i'}$ might have strictly increased the $\alpha$-value of the clients in $B({i'},\sqrt{\eps\theta})$, but no such client $j$ will be in $D^*$, since they immediately become inactive while
\begin{align*}
d^2(i,j) &\geq \left(d(i,{i'}) - d({i'},j)\right)^2\\
&> \left((1+\sqrt{\Gamma\eps}+3\eps)\sqrt{\theta/\Gamma} -\sqrt{\eps\theta}\right)^2\\
&\geq (1+\eps)^2 \theta/\Gamma \\
&\geq \alpha_j^* /\Gamma.
\end{align*}

Applying this argument to every facility open in this phase (which did not freeze $i$), we can conclude that, right before ${i'}$ is open, every $j\in D^*$ has $\alpha_j\leq \theta$ and all active ones have $\alpha_j=\theta$. \Cref{clm:nooverbidding} applied right before ${i'}$ is open ensures that
$$
\sum_{j\in DC^k}[(1-\delta)\gamma d^2(j, S) - \gamma d^2({i'},j)]^+ + \sum_{j\in {IC^k}}[\alpha_j - \gamma d^2({i'},j)]^+ +\sum_{k\leq j\leq s}[\alpha_j - \gamma d^2({i'},j)]^+\leq \hat f.
$$

Consider $j\in DC^k$ that connects to {the} open facility $i_1$. One has
\begin{align*}
    \alpha_k^* &\leq (1-\delta)d^2(i_1,k) \leq (1-\delta)(d(i_1,j) +  d(i,k) + d(i,j))^2\\
    &\overset{\text{Lem. \ref{lem:apxTriangleInequality3}}}{\leq} (1-\delta)\gamma\cdot d^2(i_1,j) + (1-\delta)\left(2+\frac{2}{\gamma-1}\right)\cdot (d^2(i,k) + d^2(i,j))\\
    & {\leq} (1-\delta)\gamma d^2(j,S) + \left(2+\frac{2}{\gamma-1}\right)\cdot (d^2(i,k)+d^2(i,j))
\end{align*}

Hence,
\begin{align*}
& \sum_{j\in {DC^k}}\left(\alpha_k^*-\left(2+\frac{2}{\gamma-1}\right)\left(d^2(i,k)+d^2(i,j)\right)-\gamma d^2(i,j)\right)\\
&\hspace{2cm}+\sum_{j\in {IC^k}}(\alpha_j-\gamma d^2(i,j))
+ \sum_{{k\leq j\leq s}}(\alpha_j-\gamma d^2(i,j))\leq \hat f.
\end{align*}

The claim follows using that for any $j\in IC^k$, $\alpha_j^*{=} \alpha_j$, and that for any $j\in A$, $\alpha_k^*\leq \alpha_j$.

\paragraph{Stage 2.} Suppose that $k$ becomes inactive by the increase of the $\alpha$-values at the end of a phase. Let $A', IC', DC'$, and $S'$ denote the respective sets of clients and open facilities at the end of this phase (immediately after $k$ becomes inactive). Since Stage 2 does not open any facilities, no clients are moved to $DC$, meaning $S' = S^k$ and $DC^k \subseteq DC'$.

By \Cref{clm:nooverbidding}, at the end of this phase, the following bound holds for $i$:
$$ \sum_{j\in DC'}[(1-\delta)\gamma d^2(j, S') - \gamma d^2(i,j)]^+ + \sum_{j\in A'\cup IC'}[\alpha_j - \gamma d^2(i,j)]^+ \leq \hat f. $$

We bound the contribution of each client $j \in D^*$ in the target inequality of the claim by its corresponding term in the above bound. Because $x \leq [x]^+$ for any real $x$, we can drop the positive part brackets for the upper bound. We divide the clients into the three groups:

\begin{itemize}
    \item For $j \in DC^k$: Because $DC^k \subseteq DC'$ and $S^k = S'$, we apply the exact same triangle inequality argument as in Stage 1. This bounds its contribution by $(1-\delta)\gamma d^2(j, S') - \gamma d^2(i,j)$, which corresponds to the first sum of \Cref{clm:nooverbidding}.
    
    \item For $j \in IC^k$: Client $j$ belongs to $A' \cup IC'$ at the end of the phase. If $j$ was already inactive before this phase, its $\alpha$-value is permanently fixed by the algorithm, meaning $\alpha_j^* = \alpha_j$. If $j$ becomes inactive in this exact phase, it is moved to $IC'$, meaning its current value $\alpha_j$ is now its final value $\alpha_j^*$. If $j$ is still active ($j \in A'$), then $\alpha_j = (1+\eps^2)\theta$, and since $j \in [k-1]$, our tie-breaking rule dictates $\alpha_j^* \leq \alpha_k^* \leq (1+\eps^2)\theta = \alpha_j$. In all cases, $\alpha_j^* \leq \alpha_j$, so its contribution is bounded by the corresponding term in the second sum.
    
    \item For $k \leq j \leq s$: These clients become inactive at the same time as $k$ or later. If $j$ remains active ($j \in A'$), its value is exactly $\alpha_j = (1+\eps^2)\theta$; because $k$ became inactive in this phase, we know $\alpha_k^* \leq (1+\eps^2)\theta = \alpha_j$. Alternatively, if $j$ becomes inactive in the exact same phase as $k$, it is moved to $IC'$, meaning $\alpha_j$ has reached its final value $\alpha_j^*$. By our tie-breaking rule, since $k \leq j$, we have $\alpha_k^* \leq \alpha_j^* = \alpha_j$. In both scenarios, $\alpha_k^* \leq \alpha_j$. Since these clients belong to $A' \cup IC'$, their contribution is bounded by the second sum.
\end{itemize}

Combining these three cases term-by-term shows that the left-hand side of our desired inequality is bounded by the left-hand side of \Cref{clm:nooverbidding}, completing the proof.
\end{proof}

\begin{claim}
\label{clm:frozenstay}
    Assume that $i$ is frozen by $i_0$ {and let} $(\tau_j)_{j\in A}$ {be the values associated to the opening of $i_0$}, $A$ being the set of active clients right before $i_0$ is opened. Then, for every $j\in D^*\cap A$, $\alpha_j^*=\tau_j$.
\end{claim}
\begin{proof}
Consider $j\in A$. First, note that if $\tau_j\geq (1-\delta)d^2(i_0,j)$ then $j$ is removed from $A$ when $i_0$ is opened and so $\alpha_j^*=\tau_j$.

In the other case, when $\alpha_j\leq \tau_j<(1-\delta)d^2(i_0,j)$, we show that $j\not\in D^*$. Consider the (future) time right after $j$ is removed from the active clients. The value {$\alpha_j$ at that time} (which is equal to the final $\alpha_j^*$) cannot be strictly greater than $(1-\delta)d^2(i_0,j)$ whether it is increased in stage 1 or in stage 2, since $i_0\in S$. By the triangle inequality:
$$d(i, j)\geq d(i_0, j) - d(i_0, i) = \left(1-\frac{d(i_0, i)}{\sqrt{\theta}}\cdot \frac{\sqrt{\theta}}{d(i_0, j)}\right)d(i_0, j)$$
where $\theta$ is the value when $i_0$ was open. As $i_0$ freezes $i$, we have $d(i_0, i)\leq (1+\sqrt{\Gamma\eps}+3\eps)\sqrt{\theta/\Gamma}$, and we have $\theta=\alpha_j\leq (1-\delta)d^2(i_0, j)$. Plugging those bounds {in} the above inequality yields,
\begin{align*}
d(i, j) &\geq \left(1-\frac{(1+\sqrt{\Gamma\eps}+3\eps)}{\sqrt{\Gamma}}\cdot\sqrt{1-\delta}\right)d(i_0, j) \\ 
{\Rightarrow}  \quad  
 \sqrt{\Gamma}d(i, j) &\geq \left(\sqrt{\Gamma} - (1+\sqrt{\Gamma\eps}+3\eps)\sqrt{1-\delta}\right)d(i_0, j){> d(i_0,j)}\ ,
\end{align*}
{where the last inequality follows from} $\sqrt{\Gamma}>2$ {and assuming that $\eps$ is small enough}. Thus $\alpha_j^*\leq d^2(i_0, j) < \Gamma d^2(i, j)$, so $j$ cannot be in $D^*$.
\end{proof}

The {last} two claims together with the fourth bullet point of \Cref{def:openable} imply the following fact:

\begin{fact}
\label{fct:bidbound}
    For any $k \in D^*$,
    \begin{align*}
    & \sum_{j\in {DC^k}}\left(\alpha_k^*-\left(2+\frac{2}{\gamma-1}\right)(d^2(i,k)+d^2(i,j))-\gamma d^2(i,j)\right)\\
    &\hspace{2cm}+\sum_{j\in {IC^k}}(\alpha_j^*-\gamma d^2(i,j)) + \sum_{k\leq j{\leq s}}(\alpha_k^*-\gamma d^2(i,j))\leq \hat f.
    \end{align*}
\end{fact}
\begin{proof}
    Consider the point in the execution when $k$ becomes inactive. If this occurs strictly before $i$ is frozen, then \Cref{clm:notfrozen} {directly implies the claim}.
    

    Else, $k$ becomes inactive when $i$ is frozen by the opening of some facility $i_0$: {let} $(\tau_j)_{j\in A^k}$ {be the values associated with the opening of $i_0$}. By the fourth bullet of \Cref{def:openable}, we have
    \begin{align*}    
    & \sum_{j\in DC^k}\left[\tau_k - \left(2+\frac{2}{\gamma-1}\right)\cdot (d^2(i,k) + d^2(i,j))-\gamma d^2(i,j)\right]^+\\
    &\hspace{2cm}+ \sum_{j\in IC^k} \left[\alpha_{j}-\gamma d^2(i,j)\right]^+ + \sum_{{k\leq j\leq s}}\left[\tau_{j}-\gamma d^2(i,j)\right]^+\leq \hat f.
    \end{align*}

    By \Cref{clm:frozenstay}, $\alpha_k^*= \tau_k$ and for every $j\in D^*$ {that is active at the considered time}, $\alpha_j^*= \tau_j$. Furthermore, for every $j\in IC^k$, $\alpha_j^*=\alpha_j$. {The claim follows} using that, for any $j\in D^*$ {that is active at the considered time}, $j \geq k$ implies $\alpha_j^*\geq \alpha_k^*$.
\end{proof}

We now are ready to prove Lemma \ref{lem:dual-feasible-log}. {The following proof is almost identical to the final part of the proof of Lemma \ref{lem:dual-feasibility-original}, the main difference being that we use \Cref{fct:bidbound} here.}

\begin{proof}[Proof of Lemma \ref{lem:dual-feasible-log}]

By definition of $D^*$, {the claim} is equivalent to
$$
\sum_{k\in D^*}\alpha_k^*\leq \hat f + \Gamma\cdot \sum_{k\in D^*}d^2(i, k).
$$

\Cref{fct:bidbound} implies that for any $k\in D^*$,
\begin{equation}
\label{eq-cnt}
\tag{$\text{UB}_k$}
(s-k+1 + |{DC^j}|)\cdot \alpha_k^*
+ \sum_{j\in {IC^k}}\alpha_{{j}}^*
\leq \hat f
+ \gamma\cdot \sum_{j\in [s]}d^2(i, j)
+ \left(2+\frac{2}{\gamma-1}\right)\cdot
\sum_{j\in {DC^k}} (d^2(i, k) + d^2(i, j))
\end{equation}

{The rest of the proof is identical to the final part of the proof of \Cref{lem:dual-feasibility-original}, hence we omit it}.
\end{proof}

\subsection{Walking Between Two Solutions: Setup} 
\label{subsec:robust_analysis}

Given the description of the \logadaptalg, we present our final algorithm \mergealg, resulting in a solution that opens $k+O(\log n{/\eps^2})$ centers by maintaining two partial solutions that lead to opening at least $k$ and at most $k$ centers respectively and {\em gradually merge them} into a single solution that opens exactly $k+O(\log n{/\eps^2})$ facilities, exploiting the $O(\log n{/\eps^2})$ adaptivity of \logadaptalg. 

Before detailing the merging procedure, which we call \mergealg, it is helpful to observe that we may generalize the analysis to make it more ``robust''. We introduce this more robust analysis in this subsection, where we also define the notation that will be used later to describe \mergealg. 

Starting from the two solutions obtained in the standard way whose facility costs differ by at most $\eta := 2^{-n}$, 
the two solutions we maintain will have an exponentially small difference $\eta$ in one of their parameters. 
(This parameter will be either the facility cost $f$ or some distance related to a {\em free facility} that will be introduced later.)
For this reason, when we merge the two solutions, it may occur that some facilities in one solution are almost completely paid for, rather than fully paid for. To address this, we generalize the definition of "openable" to "$\eta$-openable," where the only difference is in the third condition: we now require each facility that is opened to be paid $\hat{f} - \eta \gamma n$ instead of $\hat{f}$ (as in Definition~\ref{def:openable}). 
We also note that being openable only depends on the time $\theta=(1+\eps^2)^p$ and the sequence $\calH=(H_1, \dots , H_p)$
of the sets of facilities open by each stage 1 of the algorithm that was executed so far. Indeed, letting $S=\cup_{i\leq p}H_i$, we can obtain $A$, $IC$ and $DC$ as follows: $A=\{j\in D: \theta < (1-\delta)d^2(j, S)\}$, $IC=\{j\in D\setminus A:\theta<(1-\delta)\gamma d^2(j, S)\}$ and $DC=D\setminus (IC\cup A)$. Furthermore, the $\alpha$-values of the indirectly connected clients only depend on $\calH$.
This allows us to simplify notation and we have the following generalization of Definition~\ref{def:openable}.

\begin{definition}
    Consider the time $\theta$ and let $\calH$ be the sequence of sets of opened facilities. The variables $S$, $A$, $IC$, $DC$ and $\alpha$ are deduced from them. We say that a facility $i\in F$ is \emph{$\eta$-openable} (with respect to $\theta$ and $\calH$) if there are $(\tau_j)_{j\in A}$ of active clients that satisfy the following conditions.
    \begin{itemize}
        \item Only nearby clients are increased:
        \begin{gather*}
             \tau_j = \alpha_j \qquad \mbox{for every $j \in A - B(i, \sqrt{\eps \theta})$.}
        \end{gather*}  
        \item Nearby clients are only increased slightly:  
        \begin{gather*}
            \alpha_j \leq \tau_j \leq \min\{(1-\delta) d^2(j,S), (1+\eps^2)\theta\} \qquad \mbox{for every $j\in A \cap B(i, \sqrt{\eps \theta})$.}
        \end{gather*}
        \item Facility $i$ is paid for (up to the error parameter $\eta$): 
\begin{align*}
            & \sum_{j\in A}[\tau_j - (1-\delta)\gamma d^2(i, j)]^+ + \sum_{j\in IC}[\alpha_j - (1-\delta)\gamma d^2(i,j)]^+ \\
            &\quad +(1-\delta)\sum_{j\in DC}[\gamma d^2(j,S)-\gamma d^2(i,j)]^+ \geq \hat f - n \gamma \eta.
\end{align*}
        \item Dual feasibility: for every facility $i_0$ and $k \in A$,
        \begin{align*}
            & \sum_{j\in DC}\left[\tau_k - \left(2+\frac{2}{\gamma-1}\right)(d^2(i_0,k) + d^2(i_0,j))-\gamma d^2(i_0,j)\right]^+ + \sum_{j\in IC} [\alpha_{j}-\gamma d^2(i_0,j)]^+ \\
            &\quad + \sum_{j\in A}[\tau_{j}-\gamma d^2(i_0,j)]^+\leq \hat f.
        \end{align*}
    \end{itemize}
\label{def:robust_openable}
\end{definition}

We also introduce the concept of \emph{free facilities}. During the process of merging the two solutions, we introduce free facilities, whose opening costs are not necessarily paid. While we will introduce at most $O(\log n{/\eps^2})$ such free facilities, we do not restrict the number of free copies a regular facility may have. Let $\Sf \subseteq S$ denote the multiset of free facilities that have been opened.

Moreover, for each free copy $\tilde{i} \in \Sf$ of $i \in F$, we define a parameter $u(\tilde{i})$, which represents the distance between the freely opened facility $\tilde{i}$ and its original copy $i$. So for any point $x$ in the metric space (either a facility or a client), the distance is given by $d(\tilde{i}, x) = \sqrt{u(\tilde{i}) + d^2(i, x)}$. {We remark that the space remains metric after this extension. In particular, for any two points $x_1,x_2$ other than $\tilde{i}$, one has $d(x_1,x_2)\leq d(\tilde{i},x_1)+d(\tilde{i},x_2)$.}

We further introduce the following terminology: as mentioned, we refer to the facilities in $\Sf$ as \emph{free facilities}, the facilities in $F$ and $\Sr = S - \Sf$ as \emph{regular} facilities, and when we simply mention "facilities," we are referring to their union.

Finally, we will use the letters $i$ and $h$ to denote facilities. Specifically, we use $i \in F$ to denote a regular facility and use $\tilde{i}$ to denote a free copy of $i$. When referring to a facility that can be either free or regular, we use the letter $h$.

We now introduce the definition of valid sequences, which aims to capture the sequence of facilities that are opened during the while-loop of stage 1 in \logadaptalg.
Let us say $\calH'=(H_1',\dots,H_p')$ is a super-sequence of $\calH=(H_1,\dots,H_p)$ and denote it $\calH'\supseteq \calH$ if they have a common prefix of length $p-1$, i.e. for all $i<p$, $H_i'=H_i$, and $H_p'$ is a superset of $H_p$. Given two sequences $\calH_1=(H_1,\dots,H_r)$ and $\calH_2=(H_{r+1}, \dots,H_l)$, we define $\calH_1\oplus\calH_2=(H_1,\dots,H_l)$ as the concatenation of the two sequences.
\begin{definition} [$\eta$-valid sequence] 
    Consider the time $\theta$ and let $\calH$ be the sequence of sets of opened facilities. 
A sequence $\langle h_1, \ldots, h_\ell\rangle$ of facilities is $\eta$-valid (with respect to $\theta$ and $\calH$), if the following conditions hold.
\begin{itemize}
    \item Each $h_t$ is either free or $\eta$-openable with respect to the time $\theta$ and {some} super-sequence $\calH' \supseteq \calH \oplus (\langle h_1, \ldots, h_{t-1}\rangle)$ of the sequence of sets of opened facilities, where $\oplus$ is just the concatenation.
    \item  The sequence is maximal, i.e., there is no openable facility with respect to the time $\theta$ and  $\calH \oplus (\langle h_1, \ldots, h_\ell \rangle)$.
\end{itemize}
\label{def:valid_sequence}
\end{definition}
{The reader might wonder how it is possible to efficiently check whether a sequence is $\eta$-valid. We will not need to do that since during the execution of our algorithm a convenient $\calH'$ will be always available when needed.}


We remark that \logadaptalg generates $0$-valid sequences in each phase. Specifically, it generates sequences $\calH = (H_1, H_2, \ldots, H_L)$, where each $H_p$ is $0$-valid (with respect to $\theta = (1+\eps^2)^{p-1}$ and {$(H_1,\ldots,H_{p-1})$})
and contains only regular facilities. The set $S$ of opened facilities is such that $A = \emptyset$ at the final time $\theta$.

For a general sequence of {sets of facilities} $\calH = (H_1, H_2, \ldots, H_L)$, let us say it consists of $\eta$-valid sequences if, for every $p \in [L]$, $H_p$ is an $\eta$-valid sequence with respect to $\theta = (1+\eps^2)^{p - 1}$ and 
{$(H_1,\ldots,H_{p-1})$}. 
Furthermore, we say that $\calH$ is a \emph{solution} if the set $S=\cup_{q\leq p}H_q$ of opened facilities ensures that $A = \emptyset$ at the end, meaning that at the time $\theta = (1+\eps^2)^L$ of the $(L+1)$st phase, the set of clients $\{ j \mid \theta \geq (1-\delta) d^2(j,S) \}$ equals all clients in $D$. 
Otherwise, we say that $\calH$ is a partial solution.

The results in \Cref{subsec:log_adaptivity_analysis}, namely \Cref{lem:approximation-log} and \Cref{lem:dual-feasible-log},
can thus be stated as if $\calH$ is a solution of $0$-valid sequences consisting of regular facilities then the set $S$ of opened facilities satisfies
\[
    \sum_{j\in D} d^2(j, S) \leq \frac{\Gamma}{1-\delta} \cdot \left( \optlpfl(f) - f \cdot |S| \right)\,.
\]
We now allow {to open facilities that are paid up to an amount} $\hat f- \eta\gamma n$ instead of $\hat f$, and for free facilities {that can be opened without paying them}. This leads to the following generalization.
\begin{theorem}
Consider a solution $\calH$ of $\eta$-valid sequences and let $S$ be the opened facilities. Further, let $\Sf$ and  $\Sr = S - \Sf$ be the free {facilities} and regular facilities, respectively. Then  
\[
    \sum_{j\in D} d^2(j, S) \leq \frac{\Gamma}{1-\delta} \cdot \left( \optlpfl(f) - (f-\eta n) \cdot |\Sr| \right)\,.
\]
\label{thm:robust}
\end{theorem}

In the next subsection we will merge two solutions into one with $|\Sr| =k$ and $|\Sf| =O(\log n/\eps^2)$. To our help, we will have two basic routines $\completesol$ and $\completesequence$:
\begin{itemize}
    \item $\completesol_{f,u}(\calH)$ takes a partial solution $\calH$ as input and returns a solution by running \logadaptalg {from the while-loop of stage $1$, starting with the open facilities in $\calH$ and} using the opening cost $f$ and the distances $u$ for free facilities already opened.
    \item $\completesequence_{f,u}(\calH = (H_1, \ldots, H_{p-1}), H_p = \langle h_1, \ldots, h_\ell\rangle )$ where $\calH$ is  partial solution and $H_p$ is an $\eta$-valid sequence except that it may not be maximal. $\completesequence$ then {completes $H_p$ into a maximal sequence by adding the open facilities obtained by running} the while-loop of stage $1$ in \logadaptalg starting with $H_p$ (with respect to opening cost $f$ and distances $u$ for the free facilities already opened).
\end{itemize}
We note that all facilities added by \completesol and \completesequence are $0$-openable and regular. Additionally, for intuition, observe that, with the above notation, the set of facilities opened by $\completesol_{f,u}(\emptyset)$ is equivalent to that opened by \logadaptalg. (There is no free facility so $u$ is meaningless here.) 
To simplify notation, we only use $\completesequence$ when the partial solution $\calH$ is clear from the context, so we simply write $\completesequence_{(f,u)}(\langle h_1, \ldots, h_\ell \rangle)$. Furthermore, we omit $(f,u)$ when the parameters are clear from the context.

\subsection{\mergealg}
\label{subsec:merging_algorithm}

We now present our procedure \mergealg that {\em walks} between two solutions. We will maintain two solutions $(\calH, f, u)$ and $(\calH', f', u')$ whose numbers of open regular facilities {\em sandwich $k$}. (Formally, let us say $a$ and $b$ sandwich $k$ if $a<k<b$ or $b<k<a$.) We will gradually make them closer until one of them opens exactly $k$ regular facilities. 
We start from the two initial solutions obtained using the standard method. Note that the minimum nonzero distance is $1$ and the maximum squared distance is $\Maxdist = \poly(n)$, which implies that $\sopt \leq n\Maxdist$. 

\paragraph{Initialization.} Binary search on $f$ in the interval $[1/n^2, 4n\Maxdist]$ and run $\completesol_{f, u}(\emptyset)$ (with empty $u$).
When $f = 1/n^2$, every client $j \in D$ satisfies $d(j, S) = 0$ if $d(j, F) = 0$ and $d^2(j, S) \leq d^2(j, F) + 2/n^2$ otherwise, as each client can open its closest facility by itself when $\alpha _j = d^2(j, F) + 2/n^2$. If this solution opens at most $k$ centers, then it is already an $(1+2/n)$-approximation. On the other hand, when $f = 4n\Maxdist$, the algorithm will open exactly one facility, since the first facility is open when $\theta \geq 2\Maxdist$, so every client becomes immediately inactive when it is open. 
Therefore, the binary search yields two solutions $(\calH, f, u)$ and $(\calH', f', u)$ where $(\calH, f, u)$ opens less than $k$ facilities and $(\calH', f', u)$ opens more than $k$ facilities where ${f'} \leq {f}  \leq {f'} + \eta$.

\paragraph{General Step.}
\mergealg works on phase $p = 1, 2, \dots$ iteratively, and ensures that the two solutions are identical up to the $p$-th phase. Formally, \mergealg takes two solutions $(\calH, f, u)$ and $(\calH', f', u')$ as well as the current phase $p$ as input. The following promises on them will be satisfied whenever we call \mergealg. 
\begin{itemize}
    \item Both consist of $\eta$-valid sequences.
    \item The numbers of open regular facilities by the two solutions sandwich $k$. 
    \item Exactly one parameter, which we call the {\em difference parameter}, differs by at most $\eta$ between the two solutions, i.e, if $f\neq f'$ then $|f-f'| \leq \eta$ and $u(\tilde i) = u'(\tilde i)$ for every $\tilde i \in \Ff$; otherwise if $f=f'$ there is exactly one free facility $\tilde i\in \Ff$ so that $u(\tilde i) \neq u'(\tilde i)$ and $|u(\tilde i) - u'(\tilde i)| \leq \eta$. 
    \item $\calH = (H_1, H_2, \ldots, H_L)$ and $\calH' = (H'_1, H'_2, \ldots, H'_{L'})$ have a common prefix of length $p-1$, i.e., $H_i = H'_i$ for $i=1,2, \ldots, p-1$. Furthermore, the remaining sequences $H_p, \ldots, H_L$ and $H'_p,\ldots, H'_{L'}$ only contain regular facilities and are $0$-valid with respect to the parameters $f, u$ and $f', u'$, respectively. (This implies that $\calH$ and $\calH'$ open the same set of free facilities.)
\end{itemize}

Given the promises, the goal of \mergealg is to 
\begin{itemize}
\item[(1)] either find a solution of $\eta$-valid sequences that opens exactly $k$ regular facilities and $O(\log n{/\eps^2})$ free facilities, or
\item[(2)] produce two new solutions that satisfy the above promises for phase $p + 1$.
\end{itemize}
Since the maximum squared distance between any two points is $\Maxdist \leq \poly(n)$ and $f \leq 4n\Maxdist$, 
before $\theta$ reaches $6M$, at least one facility will open and every client will become inactive. 
Therefore, for some threshold $p^* = O(\log n / \eps^2)$, once $\calH$ and $\calH'$ agree on the first $p^*$ phases, they open exactly the same set of facilities which violates the promise that the numbers of open regular facilities sandwich $k$. Therefore, case (1) must happen for some $p \leq p^*$, which yields a solution of $\eta$-valid sequences that opens exactly $k$ regular facilities and $O(\log n{/\eps^2})$ free facilities. 
Combined with \Cref{thm:robust}, this proves \Cref{thr:mainBicriteria}. Indeed, we then have a solution consisting of open centers $S$ with $|S| = k  + O(\log(n)/\eps^2)$ {and $k$ regular facilities. Then,} by Theorem~\ref{thm:robust},
\begin{align*}
    \sum_{j\in D} d^2(j, S) &\leq \frac{\Gamma}{1-\delta} \cdot \left( \optlpfl(f) - (f-\eta n) \cdot k \right) \leq (\Gamma+O(\eps)) \sopt\,,
\end{align*}
where, for the last inequality, we used that $\eta = 2^{-n}$ and $\sopt\geq 1$ (since any distance is at least $1$). \Cref{thr:mainBicriteria} then follows by scaling $\eps$ appropriately.

In the remainder of this section, we present \mergealg for two solutions and a given phase $p$. It consists of three subroutines, presented in Section~\ref{sec:parameters_identical}~\ref{sec:common_prefix}, and~\ref{sec:extra_facilities} respectively. Each subroutine achieves (1) or (2) above (which ends the whole procedure or lets us move on to the next phase), or sets up the stage for the next subroutine. (The final third subroutine always achieves (1) or (2).)

\subsubsection{Making Parameters Identical}
\label{sec:parameters_identical}
By renaming,  we assume that if $f\neq f'$ is the difference parameter then $f' < f$, and if the difference parameter is $u(\tilde i) \neq u'(\tilde i)$ then $u(\tilde i) < u'(\tilde i)$.

\begin{claim}
    $H'_1, H'_2, \ldots, H'_p$ are $\eta$-valid sequences with respect to parameters $f, u$.
    \label{claim:1ststep_etavalid}
\end{claim}
\begin{proof}
Suppose first that $f\neq f'$ is the difference parameter.
As the first $p-1$ sequences are identical, we have that $H'_1, \ldots, H'_{p-1}$ are $\eta$-valid with respect to parameters $f, u$. We further have that $H'_p$ is $\eta$-valid with respect to parameters $f, u$ since it was $0$-valid with respect to parameters $f', u'$ by assumption and the only difference is that $f$ is slightly bigger than $f'$ by at most $\eta$.

Now assume $u(\tilde i) < u'(\tilde i)$ is the difference parameter. Then by the same arguments as above, we have $H'_1, \ldots, H'_{p-1}$ are $\eta$-valid with respect to $f,u$. Moreover, $H'_p$ was $0$-valid with respect to $f',u'$ by assumption and $u'(\tilde i) - \eta \leq u(\tilde i) \leq u'(\tilde i)$. 
{Next, we show that} $H'_p$ is $\eta$-valid with respect to $f,u$.

Let $H'_p = \langle h_1, \dots, h_{\ell} \rangle$ and consider $h_t$ for some $t \in [\ell]$. 
{With an abuse of notation, we will use $\calH$ to denote the sequence of sequences $(H'_1, \ldots, H'_{p-1})$ in the rest of the proof.}
{Note that $h_t$} is $0$-openable with respect to $f', u', \theta = (1 + \eps^2)^{p - 1}$ and some super-sequence $\calH' \supseteq \calH \oplus (\langle h_1, \dots, h_{t-1} \rangle)$. Let $S'$ and $S$ be the set of facilities open in $\calH'$ and $\calH$, respectively. 
{We note that $S$ can be viewed as the set of facilities opened at the beginning of the stage while $S'$ can be viewed as the set of facilities opened before $h_t$ with respect to the super-sequence $\calH'$.}

Let $(A', IC', DC', \{ \alpha'_j \}_{j \in IC'}, d'(\cdot,\cdot))$ be the state of the algorithm using $f', u'$ {with opened facilities $S'$,}
and $(A, IC, DC, \{ \alpha_j \}_{j \in IC'}, d(\cdot,\cdot))$ be the state using $f, u$ {with opened facilities $S'$.}
{We refer to these two scenarios as the old and new setting, resp.}
As the only difference is $u(\tilde i) \in [u'(\tilde i) - \eta, u'(\tilde i)]$, the distances satisfy $d^2(x,y) \le d'^2(x,y) \le d^2(x,y) + \eta$ {for any point $x,y$ in the space}. We have the following: 

\begin{itemize}
    \item $A \subseteq A'$ and $DC \supseteq DC'$. To see this, $j \in A$ implies $\theta < (1-\delta)d^2(j,{S'}) \le (1-\delta)d'^2(j,{S'})$. Hence $j \in A'$. Also, $j \in DC'$ {implies that} $\alpha'_j$ reached $(1-\delta)\gamma d'^2(j,{S'})$. Since $d^2(j,{S'}) \le d'^2(j,{S'})$, the threshold $(1-\delta)\gamma d^2(j,{S'})$ is reached no later in the new setting, so $j \in DC$. 
    \item $A' \cap DC = \emptyset$. {Indeed, assume by contradiction that} $j \in A' \cap DC$. Then $j \in A'$ implies $\theta < (1-\delta)d'^2(j,{S'}) \le (1-\delta)(d^2(j,{S'}) + \eta)$. Meanwhile, $j \in DC$ implies $(1-\delta)\gamma d^2(j,{S'}) \le \alpha_j \le \theta$. Combining these gives $\theta < \theta/\gamma + (1-\delta)\eta$, which is impossible since $\gamma > 1$, $\theta \ge 1$, and $\eta = 2^{-n}$. 
\end{itemize}
Thus, the sets $A, A' \cap IC, IC \cap IC', IC' \cap DC,$ and $DC'$ form a valid partition of $D$.

{Recall} that $h_t$ was $0$-openable with respect to the {old} setting with $(\tau'_j)_{j \in A'}$. Let 
$(\tau_j)_{j \in A}$ be defined as $\tau_j = \min(\tau'_j, (1-\delta)d^2(j, {S'}))$ for $j \in A$. 
Let us {show that} $h_t$ was $\eta$-openable with respect to the {super-sequence $\calH'$ and the} {new} setting 
with $(\tau_j)_{j \in A}$, {by} checking the conditions in \Cref{def:robust_openable} {one by one}.
\begin{enumerate}
    \item {The first bullet,} $\tau_j = \alpha_j$ for every $j \in A \setminus B(h_t, \sqrt{\eps \theta})$: For such $j$, we have $d'^2(h_t, j) \geq d^2(h_t, j) > \eps \theta$, so $j \in A' \setminus B'(h_t, \sqrt{\eps \theta})$. The validity in the {old} setting gives $\tau'_j = \alpha'_j$. Because $j \in A \subseteq A'$, $j$ is active in both settings, so $\alpha_j = \alpha'_j = \theta$, giving $\tau'_j = \theta$. Thus, $\tau_j = \min(\theta, (1-\delta)d^2(j, {S'}))$. Since $j \in A$, $(1-\delta)d^2(j, {S'}) > \theta$, so $\tau_j = \theta = \alpha_j$.
    
    \item {The second bullet,} $\alpha_j \leq \tau_j \leq \min\{(1-\delta)d^2(j,{S'}),(1+\eps^2)\theta\}$ for every $j \in A \cap B(h_t, \sqrt{\eps \theta})$: 
    Since $j \in A \subseteq A'$, validity in the {old} setting guarantees $\tau'_j \leq (1+\eps^2)\theta$. {Hence} by definition {of $\tau_j$}, $\tau_j \leq \min((1+\eps^2)\theta, (1-\delta)d^2(j,{S'}))$. For the lower bound, $j \in A$ implies $\alpha_j = \theta < (1-\delta)d^2(j,{S'})$, and $j \in A'$ implies $\tau'_j \geq \alpha'_j = \theta$. Thus, $\tau_j \geq \theta = \alpha_j$.
    
    \item {The third bullet,} 
    \begin{align*}
    \sum_{j\in A}[\tau_j - (1-\delta)\gamma d^2(h_t, j)]^+ + \sum_{j\in IC}[\alpha_j - (1-\delta)\gamma d^2(h_t, j)]^+ \\
    +(1-\delta)\sum_{j\in DC}[\gamma d^2(j,S)-\gamma d^2(h_t, j)]^+ \geq \hat f - n \gamma \eta ~:
    \end{align*}
    We argue that the contribution of each client $j$ to the left-hand side decreases by at most $\gamma \eta$ when switching {from the old to} the new setting. 
    \begin{itemize}
    \item For $j \in A {\subseteq A'}$, $\tau_j < \tau'_j$ implies $\tau_j = (1-\delta)d^2(j, {S'}) \geq (1-\delta)(d'^2(j, {S'}) - \eta) \geq \tau'_j - \eta$. Also $d^2(h_t,j) \le d'^2(h_t,j)$. Its contribution decreases by at most $\eta \le \gamma \eta$.
    \item For $j \in A' \cap IC$, its new contribution uses $\alpha_j$ and its old {one} uses $\tau'_j$. Since $j \in IC$, $\alpha_j {\geq} (1-\delta)d^2(j, {S'}) \geq (1-\delta)(d'^2(j, {S'}) - \eta) \geq \tau'_j - \eta$ (because $\tau'_j \le (1-\delta)d'^2(j, {S'})$ for all active clients $j \in A'$). Its contribution decreases by at most $\eta \le \gamma \eta$.
    \item For $j \in IC \cap IC'$, its old contribution uses $\alpha'_j$. Since $\alpha_j$ is determined either by the same phase time $\theta$ or by $(1-\delta)d^2(j,{S'})$ in Stage 2, and $d^2(j,{S'}) \ge d'^2(j,{S'}) - \eta$, we have $\alpha_j \ge \alpha'_j - \eta$. Its contribution decreases by at most $\eta \le \gamma \eta$.
    \item For $j \in IC' \cap DC$, its new contribution is $(1-\delta)[\gamma d^2(j,{S'}) - \gamma d^2(h_t, j)]^+$ and its old {one} is $[\alpha'_j - \gamma d'^2(h_t, j)]^+$. Because $j \in IC'$, $\alpha'_j < (1-\delta)\gamma d'^2(j,{S'})$. Thus $(1-\delta)\gamma d^2(j,{S'}) \geq (1-\delta)\gamma(d'^2(j,{S'}) - \eta) > \alpha'_j - \gamma \eta$. Its contribution decreases by at most $\gamma \eta$.
    \item For $j \in DC'$, since $d^2(x,y) \le d'^2(x,y) \le d^2(x,y) + \eta$ {for any $x,y$ in the metric space}, its contribution decreases by at most $(1-\delta)\gamma \eta \le \gamma \eta$.
    \end{itemize}
    Summing over at most $n$ clients gives a maximum decrease of $n \gamma \eta$. 
    The property follows.

    \item {The fourth bullet,} for every regular facility $i_0$ and $k \in A$,
    \begin{align}
    & \sum_{j\in DC}\left[\tau_k - \left(2+\frac{2}{\gamma-1}\right)\cdot (d^2(i_0,k) + d^2(i_0,j))-\gamma d^2(i_0,j)\right]^+\nonumber \\    
    &\quad + \sum_{j\in IC} \left[\alpha_{j}-\gamma d^2(i_0,j)\right]^+ + \sum_{j\in A}\left[\tau_{j}-\gamma d^2(i_0,j)\right]^+\leq \hat f~:\label{eqn:dualFeasibility:new}
    \end{align}
    Since $k \in A \subseteq A'$, the fourth condition of $\eta$-openability must hold for $k$ and the regular facility $i_0$ in the old setting. Because $i_0$ is regular, $d'(i_0, j) = d(i_0, j)$. {Thus one has}:
    \begin{align}
    \sum_{j\in DC'}\left[\tau_k' - \left(2+\frac{2}{\gamma-1}\right)\cdot (d^2(i_0,k) + d^2(i_0,j))-\gamma d^2(i_0,j)\right]^+ & +\nonumber \\ \sum_{j\in IC'} \left[\alpha_{j}'-\gamma d^2(i_0,j)\right]^+ + \sum_{j\in A'}\left[\tau_{j}'-\gamma d^2(i_0,j)\right]^+ &\leq \hat f.\label{eqn:dualFeasibility:old}
    \end{align}
    {It is sufficient to show that} the {(\emph{new})} contribution of each client $j$ {to the left-hand side of \eqref{eqn:dualFeasibility:new} is upper bounded by its (\emph{old}) contribution to the left-hand side of \eqref{eqn:dualFeasibility:old}.} Note $\tau_k \leq \tau'_k$ by definition of $\tau_k$.
    \begin{itemize}
        \item For $j \in DC'$: Its new contribution uses $\tau_k$ while {its} old {one} uses $\tau'_k$. Since $\tau_k \le \tau'_k$, the new contribution is at most the old {one}.
        \item For $j \in A$: Its new contribution uses $\tau_j$ {while its} old {one} uses $\tau'_j$. Since $\tau_j \le \tau'_j$, the new contribution is at most the old {one}.
        \item For $j \in IC \cap IC'$: Its new contribution uses $\alpha_j$  {while its} old {one} uses $\alpha'_j$. Since $d^2{(x,y)} \le d'^2{(x,y)}$ {for any $x,y$ in the metric space (and we are considering the same super-sequence, $\calH'$, in both the new and old settings)}, $j$ becomes indirectly connected no later in the new setting, giving $\alpha_j \le \alpha'_j$. {Thus the} new contribution is at most the old {one}.
        \item For $j \in A' \cap IC$: Its new contribution is $[\alpha_j - \gamma d^2(i_0, j)]^+$ {while its} old {one} is $[\tau'_j - \gamma d^2(i_0, j)]^+$. Since $j \in IC$ in the new setting but $j \in A'$ in the old setting, we have $\alpha_j \le \theta \le \tau'_j$.
        \item For $j \in IC' \cap DC$: Its new contribution is $\big[\tau_k - (2+\frac{2}{\gamma-1})\cdot(d^2(i_0,k)+d^2(i_0,j)) - \gamma d^2(i_0,j)\big]^+$ {while its} old {one} is $[\alpha'_j - \gamma d^2(i_0,j)]^+$. Let $h \in {S'}$ such that $d(j, {S'})=d(h, j)$. By \Cref{lem:apxTriangleInequality3}:
        \begin{equation}\label{claim:1ststep_etavalid:eqn1}
        \gamma d^{2}(h,j)+\left(2+\frac{2}{\gamma-1}\right)(d^2(i_0,k)+d^2(i_0,j))\geq (d(h,j)+d(i_0,k)+d(i_0,j))^2.    
        \end{equation}
        Since $k \in A$, $\tau_k \leq (1-\delta)d^2(k,{S'}) \leq (1-\delta)d^2(k,h)$. Since $j \in DC$ in the new setting, $\alpha_j \ge (1-\delta)\gamma d^2(j, {S'}) = (1-\delta)\gamma d^2(h, j)$. 
        {As discussed above in the case $j\in IC\cap IC'$,}
        client $j$ becomes indirectly connected no later in the new setting, ensuring $\alpha'_j \ge \alpha_j$. Thus:
        \begin{align*}
            \alpha_j' &\geq \alpha_j \geq (1-\delta)\gamma d^2(h, j)\\
            &\overset{\eqref{claim:1ststep_etavalid:eqn1}}{\geq} (1-\delta)\left( (d(h, j) + d(i_0, k) + d(i_0, j))^2 - \left(2+\frac{2}{\gamma-1}\right)\cdot (d^2(i_0, k) + d^2(i_0, j)) \right)\\
            &\geq (1-\delta) d^2(h, k) - \left(2+\frac{2}{\gamma-1}\right)\cdot (d^2(i_0, k) + d^2(i_0, j))\\
            &\geq \tau_k - \left(2+\frac{2}{\gamma-1}\right)\cdot (d^2(i_0, k) + d^2(i_0, j)).
        \end{align*}
        This implies $\alpha'_j - \gamma d^2(i_0,j) \ge \tau_k - \left(2+\frac{2}{\gamma-1}\right)\cdot (d^2(i_0, k) + d^2(i_0, j)) - \gamma d^2(i_0,j)$, so the new contribution {of $j$} is at most {its} old one.
    \end{itemize}
\end{enumerate}
    {This completes the proof that $h_t$ is $\eta$-openable with respect to the super-sequence $\calH'$ and the new setting $(\tau_j)_{j\in A}$. Because we choose arbitrary $t\in [\ell]$, it guarantees that $H'_p$ is $\eta$-valid with respect to $f,u$ according to our~\cref{def:valid_sequence}.}
    {The claim follows.}
\end{proof}

\paragraph{Outcome of this section of \mergealg.}
Let 
\[
    \calH'' =  \completesol_{(f,u)}(H'_1, \ldots, H'_p)\,,
\]
which is a valid call by the above claim.
If $\calH''$ opens $k$ regular facilities, we return this solution. Otherwise, 
as $(\calH, f,u)$ and $(\calH', f', u')$ sandwich $k$, we must have that either $(\calH, f, u)$ and $(\calH'', f, u)$, or $(\calH'', f, u)$ and $(\calH', f', u')$ sandwich $k$.
On the one hand, if $(\calH'', f, u)$ and $(\calH', f',u')$ sandwich $k$, then we made progress in that we now have two solutions with a common prefix of length $p$ (instead of $p-1$) and only one parameter differs by at most $\eta$. In this case, we are done with phase $p$ and can continue \mergealg with these two solutions and phase $p+1$.

On the other hand, if $(\calH, f,u)$ and $(\calH'', f, u)$ sandwich $k$, then these two solutions have the same parameters and their first $\eta$-valid sequences that differ are $H_p$ and $H'_p$. In this case, we rename $\calH''$ to $\calH'$ and proceed to the next step.

\subsubsection{Maximize Common Prefix of $H_p$ and $H'_p$}
\label{sec:common_prefix}
We are given two solutions $(\calH, f,u)$ and $(\calH', f, u)$ on the same parameters $f,u$ that sandwich $k$. Moreover, if we let $\calH = (H_1, \ldots, H_L)$ and $\calH' = (H'_1, \ldots, H'_{L'})$, then $H_i = H'_i$ for $i= 1, \ldots, p-1$ and $H_p \neq H'_p$. In addition $H_p$ and $H'_p$ only contain regular facilities. 

In this step of \mergealg, we increase the common prefix of  $H_p$ and $H'_p$ by ``slowly'' modifying $H'_p$. As $H_p$  only contain{s} regular facilities, we let $H_p = \langle i_1, i_2, \ldots, i_\ell \rangle$. Moreover, as the parameters $(f,u)$ are clear from the context, we omit them in the following when, e.g., calling the procedures \completesol and \completesequence. 

\begin{mdframed}[hidealllines=true, backgroundcolor=gray!15]
\vspace{-5mm}
\paragraph{Transforming $H'_p$  step-by-step so that $H_p$ becomes a prefix.}\ \\

\noindent Repeat the following until all of $H_p$ becomes a prefix of $H'_p$.\\

\begin{itemize}
        \item Let $q$ be the largest index so that $H'_p = \langle i_1, i_2, \ldots, i_q, h'_{q+1}, \ldots, h'_{\ell_1}\rangle$, i.e., $q$ equals the length of the common prefix $i_1, i_2, \ldots, i_q$ of $H_p$ and  $H'_p$.
            \item Update $H'_p$ by inserting a free copy of $i_{q+1}$, $\tilde i_{q+1}$, at the end of $H'_p$ with $u(\tilde i_{q+1}) = 0$. That is, we get the sequence
            \begin{gather*}
                H'_p = \langle i_1, i_2, \ldots, i_q, h'_{q+1}, \ldots, h'_{\ell_1}, \tilde{i}_{q+1} \rangle,.
            \end{gather*}

            \item
            For $t = q+1 , \ldots, \ell_1$,
            \begin{itemize}
                \item  Update $H'_p$ by first removing $h'_t$ and then completing the sequence by a call to \completesequence. More explicitly, if 
\begin{align*}
                  H'_p = &\langle i_1, \ldots, i_q, h'_t, h'_{t+1} \ldots, h'_{\ell_1}, \tilde i_{q+1}, h'_{\ell_1+1}, \ldots, h'_{\ell_2}\rangle\\
\intertext{then we delete $h'_t$ to obtain}
                    H''_p = &\langle i_1, \ldots, i_q,  h'_{t+1} \ldots, h'_{\ell_1}, \tilde i_{q+1}, h'_{\ell_1+1}, \ldots, h'_{\ell_2}\rangle\\
\intertext{and finally update $H'_p = \completesequence(H''_p)$ which will be of the form}
              &\langle i_1, \ldots, i_q, h'_{t+1}, \ldots, h'_{\ell_1}, \tilde i_{q+1}, h'_{\ell_1+1}, \ldots, h'_{\ell_2}, \ldots, h'_{\ell_3}\rangle\,,
\end{align*}
                where  $h'_{\ell_2+1},\ldots, h'_{\ell_3}$ are the facilities  added by \completesequence (which could be empty).
            \end{itemize}

            \item After removing $h'_{q+1}, \ldots, h'_{\ell_1}$, we have
            \begin{align*}
                H'_p &= \langle i_1, i_2, \ldots, i_q, \tilde{i}_{q+1}, h'_{x}, \ldots, h'_{y}\rangle\,.\\
            \intertext{We update  $H'_p$ by replacing $\tilde{i}_{q+1}$ with its regular copy $i_{q+1}$, i.e., we obtain}
                H'_p &= \langle i_1, i_2, \ldots, i_q, {i}_{q+1}, h'_{x}, \ldots, h'_{y}\rangle\,.
            \end{align*}
\end{itemize}
\end{mdframed}
We have the following claims.

\begin{claim}
    Throughout the updates to $H'_p$, $H'_1, \ldots, H'_p$ remain  $\eta$-valid sequences, and $H'_p$ contains at most one free facility.
    \label{claim:2ndstep_valid_free}
\end{claim}
\begin{proof}
    As we never update $f$ or the value $u(\cdot)$ of a free facility in $H'_1, \ldots, H'_{p-1}$, these sequences remain $\eta$-valid. 

    We continue to argue that $H'_p$ remains $\eta$-valid throughout its updates. By Claim~\ref{claim:1ststep_etavalid}, this is true before any updates to $H'_p$. We now analyze the three types of updates to $H'_p$. 
    \begin{itemize}
        \item Updating  $H'_p$ by inserting a free copy of $i_{q+1}$ maintains that it is $\eta$-valid as every regular facility stays $\eta$-openable in the sequence (they have the same prefix as before) and the sequence remains maximal.
        \item Updating $H'_p$ by removing $h'_t$ and calling \completesequence.  Removing a $h'_t$ can only decrease the prefix of regular facilities in the sequence, and so each regular facility still has a super-sequence of facilities for which it is $\eta$-openable. Finally, the call \completesequence ensures that the sequence is maximal.
        \item Updating $H'_p$ by replacing $\tilde i_{q+1}$ with its regular copy $i_{q+1}$. First note that, since $u(\tilde i_{q+1}) = 0$, this does not impact the following regular facilities are $\eta$-openable (with respect to a super-sequence). By the same argument, the sequence stays maximal (as $\tilde{i}_{q+1}$ with $u(\tilde i_{q+1}) = 0$ and  $i_{q+1}$ has the same impact on following facilities).  Finally, to see that $i_{q+1}$ is $\eta$-openable  (with respect to a potential super-sequence) notice that the prefix $i_1, \ldots, i_q$ is the same as that in $H_p$ which is an $\eta$-valid sequence. 
    \end{itemize}

    We complete the proof of the claim by arguing that $H'_p$ contains at most one free facility at any point of time. This follows from the facts that (1) $H'_p$ initially contains no free facility, (2) in one repetition, the free facility $\tilde i_{q+1}$ is inserted into $H'_p$, and (3) $\tilde i_{q+1}$ is replaced by its regular copy $i_{q+1}$ before continuing to the next repetition.
\end{proof}

\newcommand{\Hbefore}{\ensuremath{H^{\text{before}}_p}\xspace}
\newcommand{\Hafter}{\ensuremath{H^{\text{after}}_p}\xspace}

\paragraph{Outcome of this section of \mergealg.} We distinguish three cases.

\paragraph{Case 1.} If at some point $H'_p$ is such that $\completesol(H'_1, H'_2, \ldots, H'_p)$ opens exactly $k$ facilities, then we return that solution.

\paragraph{Case 2.} Otherwise,  suppose there is an update to $H'_p$ so that if we let \Hbefore and \Hafter denote $H'_p$ before and after this update, respectively, then \completesolx$(H_1, H_2, \ldots, \Hbefore)$ and \completesolx$(H_1, H_2, \ldots, \Hafter)$ sandwich $k$. There are three kinds of updates to $H'_p$. However, note that the update where $\tilde i_{q+1}$ is replaced by its regular copy $i_{q+1}$ cannot satisfy the conditions of this case. Indeed, as $u(\tilde i_{q+1}) = 0$, we have
\[
    \completesol(H_1, \ldots, \Hbefore) = \completesol(H_1, \ldots, \Hafter) 
\]
in this case, which contradicts that they sandwich $k$. 

We proceed to analyze the two other kind of updates.
If the update to $H'_p$ is to insert a free facility $\tilde{i}_{q+1}$ with $u(\tilde{i}_{q+1}) = 0$. Then
\begin{align*}
    \Hbefore&= \langle i_1, i_2, \ldots, i_q, h'_{q+1}, \ldots, h'_{\ell_1} \rangle\\
    \Hafter &= \langle i_1, i_2, \ldots, i_q, h'_{q+1}, \ldots, h'_{\ell_1}, \tilde{i}_{q+1} \rangle
\end{align*}
Note that the solution \completesolx$(H_1, H_2, \ldots, \Hbefore)$ is equivalent to \completesolx$(H_1, H_2, \ldots, \Hafter)$ if we set $u(\tilde{i}_{q+1}) = 10\Maxdist$ (except for the free facility $\tilde i_{q+1}$). We can thus do binary search on $u(\tilde{i}_{q+1})$ to find two values $u'(\tilde{i}_{q+1})$ and $u(\tilde{i}_{q+1})$ with $|u'(\tilde{i}_{q+1}) - u(\tilde{i}_{q+1})| \leq \eta$ so that the solutions
$\completesol_{(f,u')}(H_1, \ldots, \Hafter)$ and $\completesol_{(f,u)}(H_1, \ldots, \Hafter)$  sandwich $k$ (where $u$ and $u'$ are identical except for $\tilde{i}_{q+1}$). 
We have thus finished phase $p$, and we can restart \mergealg with the two solutions and phase $p+1$, since they have the first $p$, instead of $p-1$, $\eta$-valid sequences in common and exactly one difference parameter. 
Their maximality follows from that of $\Hbefore$. 

Finally, consider when the update to $H'_p$ is of the second type, i.e., $u(\tilde{i}_{q+1}) = 0$ and 
\begin{align*}
    \Hbefore  = 
                  &\langle i_1, \ldots, i_q, h'_t, h'_{t+1} \ldots, h'_{\ell_1}, \tilde i_{q+1}, h'_{\ell_1+1}, \ldots, h'_{\ell_2}\rangle\\
    \Hafter = &\langle i_1, \ldots, i_q, h'_{t+1}, \ldots, h'_{\ell_1}, \tilde i_{q+1}, h'_{\ell_1+1}, \ldots, h'_{\ell_2}, \ldots, h'_{\ell_3}\rangle
\end{align*}
In this case, $\completesol(H_1, \ldots, \Hbefore)$ and $\completesol(H_1, \ldots, \Hafter)$ are two solutions with the same parameters $f, u$, that sandwich $k$ and we proceed to the next step with the following properties of $\Hbefore$ and $\Hafter$:
\begin{itemize}
    \item Both sequences contain one free facility, $\tilde{i}_{q+1}$.
    \item $\Hbefore$ and $\Hafter$ are identical except that (1) $\Hbefore$ contains a regular facility $h'_t$ not present in $\Hafter$, and (2) $\Hafter$ contains a (potentially empty) suffix of regular facilities not present in $\Hbefore$.
\end{itemize}

\paragraph{Case 3.} None of the previous cases apply. Then,  completing the solution with the initial $H'_p$ and the final $H'_p$ cannot sandwich $k$. Since then (assuming the first case does not apply) there must be an update to $H'_p$ that results in two solutions that sandwich $k$, since the the completion of $H_p$ and the completion of the initial $H'_p$ sandwich $k$. The above implies that if we let $\Hbefore = H_p$ and $\Hafter$ equal the final $H'_p$ then \completesolx$(H_1, \ldots, \Hbefore)$ and \completesolx$(H_1, \ldots, \Hafter)$ sandwich $k$. Further note that since $\Hafter$ equals $H'_p$ after the final updates $\Hbefore$ appears as a prefix of $\Hafter$ and both sequences only contain regular facilities.   To summarize, 
in this case, \completesolx$(H_1, \ldots, \Hbefore)$ and $\completesol(H_1, \ldots, \Hafter)$ are two solutions with the same parameters $f, u$, that sandwich $k$ and we proceed to the next step with the following properties of $\Hbefore$ and $\Hafter$:
\begin{itemize}
    \item Both sequences contain no free facilities. 
    \item $\Hbefore$ and $\Hafter$ are identical except that $\Hafter$ contains a suffix of regular facilities not present in $\Hbefore$.
\end{itemize}

\subsubsection{Removing Extra Facilities}
\label{sec:extra_facilities}
The input to this step consists of two solutions $\completesol(H_1, \ldots, \Hbefore)$ and \completesolx$(H_1,\ldots, \Hafter)$ that have the same parameters $f, u$ and sandwich $k$. Moreover,  $\Hbefore$ and $\Hafter$ both contain at most one free facility and they are identical except for potential two differences: 
\begin{itemize}
\item $\Hbefore$ may contain a facility not present in $\Hafter$;
\item $\Hafter$ contains a suffix of regular facilities (which may be empty) that is not present in $\Hbefore$.
\end{itemize}

To simplify notation in this section, we let $i_*$ denote the regular facility in $\Hbefore$ that is not present in $\Hafter$ if it exists (so $i_*$ corresponds to $h'_t$ in the previous section). We further let $h_1, \ldots, h_\ell$ denote the suffix of $\Hafter$ not present in $\Hbefore$.

We now transform $\Hafter$ into $\Hbefore$ one change at a time:

\begin{mdframed}[hidealllines=true, backgroundcolor=gray!15]
\vspace{-5mm}
\paragraph{Generation of $\eta$-valid sequences from $\Hafter$ to $\Hbefore$.}\ \\
\begin{itemize}
    \item If $\Hbefore$ contains the regular facility $i_*$ not present in $\Hafter$, then let $\tilde{i}_*$ be a free copy of $i_*$ with $u(\tilde{i}_*) =0$ and obtain $H^{0}_p$ from $\Hafter$ by adding $\tilde{i}_*$ at the end. Otherwise, we let $H^0_p$ equal $\Hafter$. 
    \item For $r= 1, \ldots, \ell$, obtain $H^r_p$ from $H^{r-1}_p$ by removing $h_r$. 
\end{itemize}
\end{mdframed}
We remark that $H^\ell_p$ is identical to $\Hbefore$ since we removed the whole suffix $h_1, \ldots h_\ell$ present in $\Hafter$, except for potentially the facility $i_*$. That is the only difference in $\Hbefore$ and $H^\ell_p$ may be the placement of $\tilde{i}_*$ and that $\tilde{i}_*$ is a free copy of the regular facility $i_*$ present in $\Hbefore$. However, since $u(\tilde{i}_*) = 0$, we have that the two sequences $\Hbefore$ and $H^\ell_p$ have an identical impact on \completesol, which is a fact we use in the proof of Claim~\ref{claim:3rdstep_outcome}.

\begin{claim}
    For $r= 0, 1 \ldots \ell$, $H^r_p$ is an $\eta$-valid sequence with at most two free facilities.  
\end{claim}
\begin{proof}
    The sequence $H^r_p$ has at most two free facilities because  $\Hafter$  has at most one free facility, and we potentially add one new free facility $\tilde i_*$.

    We now verify that $H^r_p$ is $\eta$-valid. For any regular facility $h'$ in the sequence, the prefix of $H^r_p$ before $h'$ is a subsequence of the prefix before $h'$ in $\Hafter$. Thus, as $h'$ was $\eta$-openable with respect to a super-sequence $\calH(h')$ in \Hafter, $h'$ must be $\eta$-openable with respect to the same super-sequence $\calH(h')$ in $H^r_p$. This ensures that the first property of Definition~\ref{def:valid_sequence} is satisfied. To verify maximality, notice that since we added $\tilde i_*$ with $u(\tilde i_*) =0$ in the end, we have that the set of active clients after opening up the facilities in $H^r_p$ is a subset of the active clients after opening up the facilities in $\Hbefore$. Hence, the maximality of $\Hbefore$ implies the maximality of $H^r_p$.
\end{proof}

\begin{claim}
    One of the following is true.
    \begin{itemize}
        \item  For one of the sequences $H^r_p$,  \completesolx$(H_1, H_2, \ldots, H^r_p)$ opens $k$ regular facilities.
        \item  $\completesol(H_1, H_2, \ldots, \Hafter)$ and $\completesol(H_1, H_2, \ldots, H^{0}_p)$ sandwich $k$.
        \item  \completesolx$(H_1, H_2, \ldots, H^r_p)$ and \completesolx$(H_1, H_2, \ldots, H^{r+1}_p)$ sandwich $k$ for {some}  $r\in \{0, 1, \ldots, \ell-1\}$. 
    \end{itemize}
    \label{claim:3rdstep_outcome}
\end{claim}
\begin{proof}
    Assume the first bullet point does not hold, i.e., the completion of no $H^r_p$ leads to the opening of $k$ regular facilities. Now notice that the second and third bullet points ask whether consecutive solutions sandwich $k$, starting with $\Hafter$ and ending with $H^{\ell}_p$.
    By definition, $H^{\ell}_p$
     equals  $\Hbefore$ except  potentially for the placement of $\tilde{i}_*$, {where}  $\tilde{i}_*$ is the free copy of $i_*$ present in $\Hbefore$. However, as  $u(\tilde{i}_*) = 0$, the two sequences have an identical impact when we complete the solution, which implies
    \[
    \completesol(H_1, \ldots, H^{\ell}_p) = \completesol(H_1, \ldots, \Hbefore).
    \]
    Thus, as $\completesol(H_1, \ldots, \Hafter)$ and $\completesol(H_1, \ldots, \Hbefore)$ sandwich $k$, either the second or third bullet point of the statement must be valid (since no solution opened exactly $k$ regular facilities).
\end{proof}

\paragraph{Outcome of this section of \mergealg.} We distinguish three cases based on the first bullet point that is true in the above claim.

\paragraph{Case 1.} If one of the sequences $H^r_p$ is such that $\completesol(H_1, H_2, \ldots, H^r_p)$ opens  $k$ regular facilities, then we return that solution.

\paragraph{Case 2.} Otherwise, if the solutions $\completesol(H_1, H_2, \ldots, \Hafter)$ and \completesolx$(H_1, H_2, \ldots, H^{0}_p)$ sandwich $k$, then $H^0_p$ is equal to $\Hafter$ except that $\tilde{i}_p$ was added in the end of $\Hafter$ with $u(\tilde{i}_p) = 0$. 
Note that the solution \completesolx$(H_1, H_2, \ldots, \Hafter)$ is equivalent to \completesolx$(H_1, H_2, \ldots, H^0_p)$ if we set $u(\tilde{i}_{q+1}) = 10\Maxdist$ (except the presence of $\tilde{i}_*$). We can thus do binary search on $u(\tilde{i}_*)$ to find two values $u'(\tilde{i}_*)$ and $u(\tilde{i}_*)$ with $|u'(\tilde{i}_*) - u(\tilde{i}_*)| \leq \eta$ so that the two solutions
$\completesol_{(f,u')}(H_1, \ldots, H^0_p)$ and $\completesol_{(f,u)}(H_1, \ldots, H^0_p)$ sandwich $k$  (where $u$ and $u'$ are identical except for $\tilde{i}_{q+1}$). 
We have thus finished phase $p$, and we can restart \mergealg with the two solutions and phase $p+1$, since they have the first $p$, instead of $p-1$, $\eta$-valid sequences in common and exactly one difference parameter.

\paragraph{Case 3.} Finally, consider the case when none of the above cases apply and (thus by the above claim)   there is an  $r\in \{0, 1, \ldots, \ell-1\}$ so that \completesolx$(H_1, H_2, \ldots, H^r_p)$ and \completesolx$(H_1, H_2, \ldots, H^{r+1}_p)$ sandwich $k$.
The only difference between $H^r_p$ and $H^{r+1}_p$ is that $H^r_p$ contains $h_{r+1}$ which is not present in $H^{r+1}$. Let $H_p$ be the sequence obtained by adding a free copy $\tilde i_{r+1}$ of $h_{r+1}$ at the end of $H^{r+1}$ (which clearly maintains that is $\eta$-valid), then 
\begin{align*}
    \completesolx(H_1, H_2, \ldots, H_p) &= \completesolx(H_1, H_2, \ldots, H^r_p) \mbox{ if $u(\tilde i_{r+1}) = 0$}\\
\intertext{and} 
    \completesolx(H_1, H_2, \ldots, H_p) &= \completesolx(H_1, H_2, \ldots, H^{r+1}_p) \mbox{ if $u(\tilde i_{r+1}) = 10\Maxdist$.}
\end{align*}
Hence, doing a binary search on $u(\tilde{i}_{r+1})$ as in the previous case finishes phase $p$. We can restart \mergealg with the two solutions and phase $p+1$, since they have the first $p$, instead of $p-1$, $\eta$-valid sequences in common and exactly one difference parameter.

\subsection{Analysis of the Robust Algorithm}

In this section, we prove \Cref{thm:robust}, showing that a solution $\calH$ of $\eta$-valid sequences yields an almost $\Gamma$-LMP approximation. The proof is almost identical to that for \Cref{alg:logadaptive}. There are only two minor changes:
\begin{enumerate}
    \item \Cref{def:robust_openable} of $\eta$-openability allows for a facility to be paid up to an $n\gamma\eta$-additive difference. This difference only appears in \Cref{lem:robustapprox} below, which is an analog of \Cref{lem:approximation-log}.

    \item \Cref{def:valid_sequence} of an $\eta$-valid sequence allows $h_t$ is $\eta$-openable with respect to a super sequence $\calH'\supseteq \calH\oplus(\langle h_1,\dots,h_{t-1}\rangle)$. This difference only appears in (the last part of) the proof of \Cref{lem:robustfeasible}, an analog of \Cref{lem:dual-feasible-log}.
\end{enumerate}
These two lemmas immediately imply \Cref{thm:robust}.

In order to use mostly the same analysis, one syntactic difference between the setting of \Cref{subsec:log_adaptivity_analysis} and the current setting that we have to reconcile is the the former analyzes the outcome of \Cref{alg:logadaptive} while the latter analyzes a solution $\calH=(H_1,\dots,H_L)$ of $\eta$-valid sequences. To use the same terminology from \Cref{subsec:log_adaptivity_analysis}, let us consider the {\em execution} of the solution $\calH$ where for $p=1,\dots,L$, we run the $p$-th phase according to the sequence $H_p$; in Stage 1, each facility in the sequence becomes open one-by-one with the corresponding $(\tau_j)_j$ values, and in Stage 2 at the end of the phase, $\theta\leftarrow (1+\eps^2)\theta$.

Then, an {\em atomic step} can still be defined in the same way as in \Cref{alg:logadaptive}; it corresponds to opening a facility in Stage 1, {possibly a free facility}, (and updating $\alpha$, $S$, $A$, $IC$, $DC$ accordingly) or increasing $\alpha$-values in Stage 2 simultaneously (and updating $IC$ accordingly). By ``at any point in the algorithm'', we mean any point in the algorithm's execution that is not in the middle of an atomic step.

Throughout the analysis, we let $\alpha_j^*$ be the final $\alpha$-value of client $j\in D$, and we let $S^*$ be the final set of opened facilities. We start by analyzing the approximation guarantee of the algorithm with respect to {$\alpha^*$}, and then we prove that $\alpha^*/\Gamma$ is a feasible solution to the dual.

\subsubsection{Approximation Guarantee}

{The following lemma is the analogue of Lemma \ref{lem:approximation-log}.}

\begin{lemma}
\label{lem:robustapprox}
    We have $\sum_{j\in D}\alpha_j^*\geq (1-\delta)\sum_{j\in D}d^2(j, S)+\sum_{i\in \Sr}(\hat f-n\gamma\eta)$.
\end{lemma}
\begin{proof}
{The proof is very similar to the proof of Lemma \ref{lem:approximation-log}. It is sufficient to show that, at any point of the algorithm, one has}
    \begin{gather}
        \sum_{j\in DC} \alpha_j \geq \sum_{j\in DC}(1-\delta)\gamma d^2(j, S)  + \sum_{i\in \Sr} (\hat f-n\gamma\eta)\,.
        \label{eq:robustIH_approx_guarantee}
    \end{gather}
{Indeed, at the end of the algorithm all the clients are in IC or DC. For the ones} in IC, we have $\alpha_j\geq (1-\delta)d^2(j, S)$. {For the remaining ones it is sufficient to apply the above inequality and the fact} that $\gamma\geq 1$.
    
    {We prove \eqref{eq:robustIH_approx_guarantee} by induction on the steps of the algorithm.} The equality is initially true since $A= D$ (thus $IC\cup DC = \emptyset$) and $S = \emptyset$. {For the inductive step,} we next analyze each one of the two stages separately. 

    \paragraph{Stage 1.} Consider what happens when we open a facility $h$, i.e., add it to $S$. 

    If $h$ is a free facility, the $\alpha$-values and $\Sr$ do not change, and the only change is {that} some client $j$ with $\alpha_j\geq (1-\delta)\gamma d^2(j, S)$ becomes {directly connected}. The amount of the increase of the left-hand side of \eqref{eq:robustIH_approx_guarantee} is at least that of the right-hand side.
    
    If $h$ is a regular facility $i$, {the argument is analogous to the proof of lemma \ref{lem:approximation-log} in the same case, the only difference being that $\hat{f}$ is replaced by $\hat{f}-n\gamma\eta$.}

    \paragraph{Stage 2.}
    No facility is open at this stage, and no client is added to $DC$ either.
\end{proof}

\subsubsection{Dual Feasibility}

We next prove {the following lemma, which is analogous to \Cref{lem:dual-feasible-log}.} 
\begin{lemma}
\label{lem:robustfeasible}
We have
\[
\sum_{k\in D}[\alpha_k^*-\Gamma d^2(i,k)]^+\leq \hat f.
\]
\end{lemma}
The following three claims and their proofs are identical to \Cref{clm:nooverbidding}, \Cref{clm:notfrozen} and \Cref{clm:frozenstay}, without any modification {and using a similar notation}.

\begin{claim}[No Over-Bidding]
\label{clm:robustnooverbidding}
    At any point in the algorithm, for every facility $i_0$,
    $$\sum_{j\in DC}[{(1-\delta)}\gamma d^2(j, S) - \gamma d^2(i_0,j)]^+ + \sum_{j\in A\cup IC}[\alpha_j - \gamma d^2(i_0,j)]^+\leq \hat f.$$
\end{claim}

As usual, fix an arbitrary facility $i\in F$. The setup in this paragraph is also identical to \Cref{subsubsec:dual-feasibility-log}. Let $D^*=\{j\in D:\alpha_j^*>\Gamma d^2(i, j)\}$ be those clients that contribute to the left-hand side. 
We also reuse the definitions of freezing and $IC^k, DC^k, A^k$ for {for $k\in D^*$ that becomes inactive strictly before $i$ becomes frozen}.

\begin{claim}
\label{clm:robustnotfrozen}
For any $k\in D^*$ that becomes inactive strictly before $i$ becomes frozen,
\begin{align*}
& \sum_{j\in {DC^k}}\left(\alpha_k^* - \left(2+\frac{2}{\gamma-1}\right)(d^2(i,k)+d^2(i,j))-\gamma d^2(i,j)\right)+\sum_{j\in {IC^k}}(\alpha_j^*-\gamma d^2(i,j))\\
&\quad +\sum_{{k\leq j\leq s}}(\alpha_k^*-\gamma d^2(i, j))\leq \hat f.
\end{align*}
\end{claim}

\begin{claim}
\label{clm:robustfrozenstay}
Assume that $i$ is frozen by $i_0$ {and let} $(\tau_j)_{j\in A}$ {be the values associated with the opening of $i_0$}. Then, for every $j\in D^*\cap A$, $\alpha_j^*=\tau_j$.
\end{claim}

The remaining part of the proof needs some changes from \Cref{subsubsec:dual-feasibility-log}, due to the fact that a facility may be openable with respect to a super-sequence instead of a sequence that the algorithm executes.
Let $i_0$ be the facility that freezes $i$, which opens in (stage 1) of phase $p$. Let $\theta = (1+\eps^2)^p$ be the $\theta$ value in phase $p$.
Let $\mathcal{H'} \supseteq \mathcal{H}$ be the super-sequence with respect to which $i_0$ is openable. As in \Cref{subsec:log_adaptivity_analysis}, we order the clients in $D^*$ according to the time they are removed from $A$ (in the algorithm) and break ties according to $\alpha^*$-values in increasing order. 
The only additional rule is for the clients who are in $A$ right before $i_0$ is open; let us call them $A$. They are, by \Cref{clm:robustfrozenstay}, the last group of clients in $D^*$. 
Let $A' \subseteq A$ be the clients who are active right before $i_0$ is open {according to} the supersequence $\mathcal{H'}$. Then for the order within $A$, we place 
$A \setminus A'$ first and $A'$ after, and within each group we sort by the $\alpha^*_j$ values. 
Similarly, let $IC'$ and $DC'$ be the set of indirectly and directly connected clients right before $i_0$ is open in the supersequence $\mathcal{H}'$, and $I'=A\setminus A'$.
{For $k\in D^*$ that becomes inactive when $i$ becomes frozen, we define $DC^k$ by $DC'\cap [k-1]$, and $IC^k$ by $(IC'\cup A')\cap [k-1]$.}
We prove the following analog of \Cref{fct:bidbound}.

\begin{fact}
\label{fct:robustbidbound}
    For any $k \in D^*$:
    \begin{align*}
    & \sum_{j\in DC^k}\left(\alpha_k^*-\left(2+\frac{2}{\gamma-1}\right)(d^2(i,k)+d^2(i,j))-\gamma d^2(i,j)\right)\\
    &\quad +\sum_{j\in IC^k}(\alpha_j^*-\gamma d^2(i,j)) + \sum_{k\leq j{\leq s}}(\alpha_k^*-\gamma d^2(i,j))\leq \hat f.
    \end{align*}
\end{fact}
\begin{proof}
    Let $k\in D^*$. Consider the point in the execution when $k$ becomes inactive. If this occurs strictly before $i$ is frozen, then by \Cref{clm:robustnotfrozen}, we have
    \begin{align*}
    & \sum_{j\in {DC^k}}\left(\alpha_k^*-\left(2+\frac{2}{\gamma-1}\right)(d^2(i,k)+d^2(i, j)) - \gamma d^2(i,j)\right)\\
    &\quad + \sum_{j\in {IC^k}}(\alpha_j^* - \gamma d^2(i, j)) + \sum_{{k\leq j\leq s}}(\alpha_k^*-\gamma d^2(i, j)) \leq \hat f.
    \end{align*}
    This is precisely the claim.

    Else, $k$ becomes inactive when $i$ is frozen by the opening of some facility $i_0$: {let} $(\tau_j)_{j\in A^k}$ {be the values associated with the opening of $i_0$}. {For simplicity, we use $\Gamma'=2+\frac{2}{\gamma-1}$ in the rest of this proof.}
    
    If $k\in D^*\cap A'$, by the fourth bullet of $\eta$-openability, we have
    \[
    \sum_{j\in DC'}\left[\tau_k - 
    {\Gamma'}
    \cdot (d^2(i,k) + d^2(i,j))-\gamma d^2(i,j)\right]^+ + \sum_{j\in IC'} [\alpha_{j}-\gamma d^2(i,j)]^+ + \sum_{j\in A'}[\tau_{j}-\gamma d^2(i,j)]^+\leq \hat f.
    \]
    By \Cref{clm:robustfrozenstay}, $\alpha_j^*=\tau_j$ for every $j\in D^*\cap {A'}$ {(and thus for $j=k$)}. 
    Furthermore, for $j \in IC'$, even when $\alpha_j$ values above are determined by the super-sequence $\mathcal{H'}$ instead of $\mathcal{H}$, the $\alpha_j$ values are the same in both executions; these sequences are identical until the $(p-1)$th phase, so any $j$ that became inactive in phase $(p-1)$ or earlier has the same $\alpha_j$ values. For $j$ that becomes inactive in (stage 1 of) phase $p$ in $\mathcal{H'}$, we have $\alpha_j = \theta$ in $\mathcal{H'}$, and we also have $\alpha^*_j = \theta$ as well, because if we had $\tau_j > \theta$ at any opening of a facility in phase $p$, it implies $j \in DC'$, which is contradiction.
    Hence
    \[
    \sum_{j\in DC'}\left[\alpha_k^* - 
    {\Gamma'}
    \cdot (d^2(i,k) + d^2(i,j))-\gamma d^2(i,j)\right]^+ + \sum_{j\in IC'} [\alpha_{j}^*-\gamma d^2(i,j)]^+ + \sum_{j\in A'} [\alpha_{j}^*-\gamma d^2(i,j)]^+ \leq \hat f.
    \]
    By our ordering, (1) $j \in DC'$ and $k \in A'$ imply $j < k$, and (2) $k \in A'$ and $j \geq k$ imply $\alpha^*_j \geq \alpha^*_k$. Therefore, the above inequality implies what we want:
    \[
    \sum_{j<k:j\in DC'}\left(\alpha_k^* - 
    {\Gamma'} \cdot (d^2(i,k) + d^2(i,j))-\gamma d^2(i,j)\right) + \sum_{j<k:j\not\in DC'} (\alpha_{j}^*-\gamma d^2(i,j)) + \sum_{j\geq k}(\alpha_{k}^*-\gamma d^2(i,j))\leq \hat f.
    \]

    If $k\in D^*\cap I'$, then $\alpha^*_k = \theta$ since, with respect to $\mathcal{H'}$, $k$ was inactive right before opening $i_0$.
    \Cref{clm:robustnooverbidding} applied to the execution with respect to $\calH'$ right before $i_0$ was opened ensures that
    \[
    \sum_{j\in DC'}[\gamma d^2(j,S) - \gamma d^2(i_0,j)]^+ + \sum_{j\in A'\cup IC'}[\alpha_j - \gamma d^2(i_0,j)]^+\leq \hat f.
    \]
    We conclude, using \Cref{lem:apxTriangleInequality3}, \Cref{clm:robustfrozenstay} and that for any $j\in A'\cap D^*$, $j \geq k$ implies $\alpha_j^*\geq \alpha_k^*$.
\end{proof}

We are now ready to prove \Cref{lem:robustfeasible}.
\begin{proof}[Proof of \Cref{lem:robustfeasible}]
By definition of $D^*$, this is equivalent to proving that
$$\sum_{k\in D^*}\alpha_k^*\leq \hat f + \Gamma\cdot \sum_{k\in D^*}d^2(i, k).$$

\Cref{fct:robustbidbound} implies that for any $k\in D^*$,
\begin{equation}
\label{eq-cnt}
\tag{$\text{UB}_k$}
(s-k+1 + |{DC^k}|)\cdot \alpha_k^*
+ \sum_{j\in {IC^k}}\alpha_j^*
\leq \hat f
+ \gamma\cdot \sum_{j\in [s]}d^2(i, j)
+ \left(2+\frac{2}{\gamma-1}\right)\cdot
\sum_{j\in {DC^k}} (d^2(i, k) + d^2(i, j))
\end{equation}
{The rest of the proof is identical to the final part of the proof of \Cref{lem:dual-feasibility-original}, hence we omit it}.
\end{proof}

{We conclude with the proof of \Cref{thm:robust}.}
\begin{proof}[Proof of \Cref{thm:robust}]
By \Cref{lem:robustapprox}, the dual values $\frac{\alpha^*_j}{1-\delta}$ can pay the cost of the solution with facility opening cost $\frac{\Gamma f-n\gamma \eta}{1-\delta}{\geq \frac{\Gamma}{1-\delta}(f-n\eta)}$ {for the regular facilities}. Lemma \ref{lem:robustfeasible} implies that $\alpha^*/\Gamma$ is a feasible solution to the dual of the facility location LP with facility cost $f$. We conclude that
$$
\frac{\Gamma}{1-\delta}\sopt_{LP}(f)\geq \sum_{j\in D}d^2(j,S)+|{S_{reg}}|\frac{\Gamma}{1-\delta}(f-n\eta).
$$
\end{proof}

\section{$(5+O(\sqrt{\eps}))$-Approximation for  $(\frac{\zeta}{\log n})$-{Stable} Instances}
\label{sec:centerremoval}
In this section, we give a $(5+O(\sqrt{\eps}))$-approximation algorithm for 
$(\frac{\zeta}{\log n})$-stable $k$-means instances. 
\begin{restatable}[]{theorem}{stableapprox}
For any constants $\eps, \zeta>0$, there exists a randomized polynomial-time algorithm that, given a $(\zeta/\log n)$-stable $k$-means instance,  returns a solution of cost at most $(5+O(\sqrt{\eps})\sopt$ with high probability. 
    \label{thm:mainadditivecenters}
\end{restatable}

To simplify notation, we let $\eps$ be the minimum of $\eps$ and $\zeta$ in the above theorem and further assume that $\eps<1/12$ is sufficiently small. With this notation, we present a $(5+O(\sqrt{\eps})$-approximation algorithm for $\eps/\log(n)$-stable instances. This implies the above theorem as any $\zeta/\log(n)$-stable instance is also $\zeta'/\log(n)$-stable with $\zeta' \leq \zeta$.

The algorithm and the proof of the above theorem follow several steps where we iteratively simplify the instance by adding more structure. Finally, we obtain enough structure to solve the problem by maximizing a submodular function subject to a partition matroid constraint. We have not optimized constants in favor of simplifying the description, and throughout this section we let $O_\eps(\log n))$ denote $O(\log(n)/\eps^{O(1)})$. An overview of this section and the algorithm is as follows:
\begin{mdframed}[hidealllines=true, backgroundcolor=gray!15]
\vspace{-5mm}
\paragraph{Overview of a $(5+O(\sqrt{\eps})$-Approximation for $k$-means on $(\eps/\log n)$-
Stable Instances $k, \clients, \facilities, \dist$.}\ \\
\begin{enumerate}
    \item We start, in \Cref{sec:localSearchAnalysis}, by running a standard Local Search on $(k, \clients, \facilities, \dist)$ to obtain a set of centers $S$. It is well-known that $S$ is a constant-factor approximation. Moreover, for a stable instance, we show that $S$ ``correctly'' implicitly identifies all but $O_\eps(\log n)$ centers of $\opt$. Specifically, this ensures that all but $O_\eps(\log n)$ clusters have the same cost as in $\opt$ up to a $(1+O(\sqrt{\eps}))$ factor. 
    \item In \Cref{sec:dsampleproc}, we use a technique inspired from the 
    classic $k$-means++ algorithm and its
    so called $D^2$-sampling procedure to
    sample clients proportional to their cost in solution $S$, and we show that if we take  $O_\eps(\log n)$ samples $W$ then with good probability we ``hit'' all but an insignificant fraction of clusters of high cost in $S$. 
    \item In \Cref{sec:ballguesses}, we then, guess a subset of the  sampled clients $W$ so that each hit cluster has exactly one client, which we refer to as the leader of that cluster. In addition, for each client $\ell$ that is a leader, we guess the (approximate) distance from $\ell$ to the optimal center in the optimal cluster it hits. This defines a set $\calB$ of balls with the guarantee that the optimal solution opens one center in each of these balls. Moreover, if we let $\ho$ be the optimal centers corresponding to these balls, we prove in \Cref{sec:dummycenters} that there is a swap of centers $M_O = S-S_0  \cup \ho$ that removes $S_0$ and adds the optimal centers $\ho$ so that the cost of the solution $M_O$  is good and in particular a $(4+O(\sqrt{\eps}))$-approximation.  
    \item The remaining steps are then focused on approximating this swap, i.e., intuitively to that of finding $S_0$ and $\ho$:  
   \begin{enumerate}
    \item We first find the set $\calQ$ of expensive centers of $S_0$ so that the {clusters of} centers $S_0-\calQ$ ha{ve} total cost $O(\eps \sopt)$ in \Cref{sec:removalofExpensive}. This allows us to think of the {clusters of} centers in $S_0 - \calQ$ as contracted, i.e., as single weighted points, and in particular, that all clients of these clusters are assigned to the same center when their center in $S_0$ is removed. 
     We then, in \Cref{sec:consistentlyassigning}, use this structure to partition $S_0 - \calQ$ into two sets $\calU$ and $\calR$ where  $\calU$ are those centers that are reassigned to other centers in $S - S_0$ and $\calR$ are those centers that are assigned to $\ho$. 

    \item The set $\calU$ is then approximated with a set $\bcalU$ (with close to identical properties) in \Cref{sec:removalofCheap}.
    \item At this stage, we have thus ``guessed'' all centers $\calQ \cup \bcalU$ to remove from $S$ except for those in $\calR$. This is achieved in \Cref{sec:submodularopt} where we reduce the problem of both finding $\calR$ and an approximate version of $\ho$ to that of submodular function optimization subject to a partition matroid (since we should open one center per ball in $\calB$). 
    \end{enumerate}
\end{enumerate}

\end{mdframed}

We remark that in each of Steps 3, 4a, 4b, and 4c, we will make guesses by enumerating polynomially $n^{\eps^{-O(1)}}$ many potential choices.  In the following, we describe the algorithm and its analysis by presenting each of these steps. When presenting the next step, we analyze the branch of the algorithm that took the successful guesses up to that point. The algorithm outputs the solution of the smallest cost among the $n^{\eps^{-O(1)}}$ constructed solutions obtained by enumerating all possible guesses (see \Cref{sec:stable:everythingtogether} for a more formal treatment where we put everything together).
 In particular, the final solution will have cost at most the solution that is output in the analyzed branch of ``correct'' guesses. 

We end this section by introducing basic definitions and concepts, followed by the definition of matroids and submodular functions. We also state the known result, an  $(1-1/e)$-approximation algorithm for maximizing monotone submodular functions subject to a matroid constraint, that we use in the last step (\Cref{sec:submodularopt}).


\paragraph{Definitions of leaders.} 
We use the following definitions throughout this section. Given a set of centers $A\subseteq \facilities$, we let $\clcost(A):=\sum_{p\in \clients}\dist^2(p,A)$ denote the cost of the corresponding clustering (where each client $p$ is assigned to the closest center in $A$ and pays the corresponding distance). Sometimes we will consider a (possibly suboptimal) assignment $\mu:\clients \rightarrow A$ of clients to centers in $A$, and let $\clcost(A,\mu):=\sum_{p\in \clients}\dist^2(p,\mu(p))$ be the corresponding cost. Let $\opt$ be a fixed optimum solution for the $k$-means instance, and let $\sopt = \clcost(\opt)$ be its cost. 
For a client $p\in \clients$, let $\opt_p:=\dist^2(p,\opt)$ 
be the cost paid by $p$ in $\opt$.

Let $C^*$ be an $\opt$ cluster with center $c^*\in \opt$. We let $\avg_{C^*, \opt}$ be the squared distance from $c^*$ to the $\eps |C^*|$th closest client of $C^*$ to $c^*$.
We further let $C^*_\avg$ be the set of clients of $C^*$ at a squared distance at most
$\avg_{C^*, \opt}$ from $c^*$, we refer to these clients
as \emph{leaders} for $C^*$. So $\avg_{C^*, \opt}$ is the \emph{maximum squared distance} from a leader to its center $c^*$. Also note that we have, $|C^*_{avg}|\geq {\eps}|C^*|$. 
We thus have that $\avg_{C^*, \opt}$ is upper bounded by $\frac{1}{1-\eps}$ times the average cost $\frac{1}{|C^*|} \sum_{p\in C^*} \dist^2(p, c^*)$, and at the same time $C^*_{avg}$ contains a constant fraction of the clients. 

\paragraph{Partition matroids.}
In \Cref{sec:submodularopt}, we formulate a submodular function that we then maximize over a partition matroid using known results. We define those concepts and also state the result that we will use. First, recall the definition of a matroid and partition matroid. 
A \emph{matroid} is a tuple $(E, \mathcal{I})$ defined on a ground set $E$ with a family of independent sets $\mathcal{I}$ that satisfy: (1) if $A \in \mathcal{I}$ and $A' \subseteq A$, then $A' \in \mathcal{I}$; and (2) if $A, B \in \mathcal{I}$ and $|A| < |B|$, then there exists an element $e \in B - A$ such that $A \cup \{e\} \in \mathcal{I}$.
  In the special case of  a \emph{partition matroid}, the ground set $E$ is partitioned into disjoint subsets $E_1, E_2, \dots, E_k$, and there are non-negative integers $r_1, r_2, \dots, r_k$ such that:
\begin{itemize}
    \item A subset $I \subseteq E$ is independent if and only if $|I \cap E_i| \leq r_i$ for each $i = 1, 2, \dots, k$.
\end{itemize}
In other words, each subset $E_i$ has a capacity limit $r_i$, and an independent set $I$ cannot contain more than $r_i$ elements from $E_i$.
To work with matroids (to get a running time that is polynomial in the size of the ground set $E$), we often use \emph{independence queries}, which allow us to determine if a given subset $A \subseteq E$ is independent (i.e., whether $A \in \mathcal{I}$). Specifically, an independence query on a subset $A$ returns {true} if $A$ is independent and {false} otherwise.

\paragraph{Submodular functions.} Having defined matroids, we proceed to define non-negative monotone submodular functions. For a finite ground set $E$, a set function $f: 2^E \to \mathbb{R}$ is \emph{submodular} if it satisfies the diminishing returns property: for every pair of sets $A \subseteq B \subseteq E$ and any element $e \in E - B$, it holds that
\[
f(A \cup \{e\}) - f(A) \geq f(B \cup \{e\}) - f(B).
\]
Intuitively, this means that adding an element $e$ to a smaller set $A$ provides at least as much additional value as adding $e$ to a larger set $B$.
In addition, $f$ is called \emph{monotone} if for any pair of sets $A \subseteq B \subseteq E$, we have $f(A) \leq f(B)$. In other words, adding elements to a set does not decrease the function value. Finally, $f$ is \emph{non-negative} if $f(A) \geq 0$ for all $A \subseteq E$. 

A celebrated result~\cite{CalinescuCPV11} gives a polynomial-time $(1-1/e)$-approximation algorithm for the problem of maximizing a non-negative monotone submodular function over a matroid constraint. While the original result was a randomized algorithm, a recent striking result~\cite{BuchbinderFeldman24}, building upon the work of~\cite{FilmusW14}, gives a deterministic algorithm with the same guarantee. We remark that polynomial time here refers to a running time that is polynomial in the size of the ground set $E$.   To summarize, they show 
the following theorem.
\begin{theorem}
    Let $f: 2^E \to \mathbb{R}$ be a non-negative monotone submodular function that we can evaluate in polynomial time, and let $(E, \mathcal{I})$ be a matroid on the same groundset, for which we can answer independence queries in polynomial time. Then, for every $\zeta>0$, there is a polynomial-time algorithm that outputs a set $X\subseteq E$ with $X \in \calI$ so that
    \[
       f(X) \geq (1-1/e - \zeta) \cdot \max_{X^*\in \calI} f(X^*)\,.  
    \]
    \label{thm:submodularmatroidoptimization}
\end{theorem}

\subsection{Properties of a Locally Optimal Solution}
\label{sec:localSearchAnalysis}
Let $S$ be a {locally optimal} solution output by 
the following standard local search algorithm. We remark that its running time is polynomial since by \Cref{lem:aspectratio} the cost of a solution is a polynomially bounded integer and, at each improving step, the cost decreases by at least one.

\begin{mdframed}[hidealllines=true, backgroundcolor=gray!15]
\vspace{-5mm}
\paragraph{Classic Local Search Procedure for $k$-Means.}\ \\
\begin{enumerate}
    \item $S \gets $ Arbitrary solution for $k$-Means.
    \item While there exists a solution $S'$ such that $|S \Delta S'|\le 2$ and $\clcost(S') <\clcost(S)$:~~
    \begin{enumerate}
        \item $S \gets S'$.
    \end{enumerate}
    \item Output $S$.
\end{enumerate}
\end{mdframed}
Local search achieves a $\apxLSkmeans$ approximation for
$k$-Means, see \cite{LSGupta}.
\begin{lemma}
    $\clcost(S) \leq \apxLSkmeans\cdot \sopt$.
    \label{lemma:localsearchis5approximation}
\end{lemma}

In the remaining part of the algorithm and its analysis we fix $S$, and 
 recall that we fixed the reference optimal solution $\opt$. We distinguish different types of clusters of $\opt$ and analyze their properties in the next two sections. We let $S_p:=\dist^2(p,S)$ be the cost paid
by $p$ in $S$ and recall that $\opt_p$ denotes the cost $p\in \clients$ pays in $\opt$.

\subsection{Pure Clusters}
\label{sec:pure}
 We say that a cluster $C^*$ of $\opt$ is \emph{pure} 
(w.r.t. $S$) if there exists a cluster $C'$ of $S$ such that $|C' \Delta C^*| \le {3\eps} \min(|C^*|,|C'|)$. In which
case we also say that $C'$ is pure w.r.t. $\opt$. We say that $C'$ and $C^*$ are associated.
The following lemma follows from the fact that the input is $\beta$-stable
for $\beta := \eps/\log n$ and that the solution $S$  is a $\apxLSkmeans$-approximation by \Cref{lemma:localsearchis5approximation}. Given a client or center $a$ and a distance $r$, we let $B(a,r)$ be the set of clients and centers at distance at most $r$ from $a$. This lemma is originally proved in~\cite{Cohen-AddadS17}, we include a proof for completeness.
\begin{lemma}[Restatement from~\cite{Cohen-AddadS17}, Lemma IV.4]
\label{lem:numnonpure}
The number $k_{imp}$ of clusters of $\opt$  that are not pure w.r.t. $S$ (and thus the number of clusters of $S$ which are not pure w.r.t $\opt$)
    is at most {$\frac{\log n}{\eps^3}$}. 
\end{lemma}
\begin{proof} We let $C^*_o$ denote the optimal cluster for each optimal center $o\in \opt$ and we let $C_c$ denote the cluster in $S$ corresponding to center $c\in S$. 
     We refine  $\opt$  in steps so that the final set only consists of centers of pure clusters. 

    We start by removing expensive clusters in $\opt$. Specifically, let 
    \[
        \widetilde{\opt} = \{ o \in \opt \mid \sum_{p\in C^*_o} \dist^2(p, o) \geq  \frac{5\eps^3}{\log n}\sopt\}\,.
    \]
    Clearly $|\widetilde{\opt}|\leq \frac{\log n}{5\eps^3}$. Let $\opt_1:=\opt-\widetilde{\opt}$.

     For the next step,  for every $o \in \opt_1$, let $\pi(o)$ be the closest center in $\opt - \{o\}$ to $o$, and let $d_o = d(o, \opt - \{o\}) = d(o, \pi(o))$ be the respective distance.  As the instance is $\beta{=\frac{\eps}{\log n}}$ stable we have for every $o\in \opt_1$:
    \begin{gather*}
      |C^*_o| \dist^2_o =  |C^*_o| \dist^2(o, \pi(o)) \geq  \sum_{p\in C^*_o} \left(\frac{1}{2}\dist^2(p, \pi(o)) - \dist^2(p,o)\right)  \geq \frac{1}{2}\beta {\sopt}-\frac{5\eps^3}{\log n}\sopt\geq \frac{1}{3}\beta\sopt, 
    \end{gather*}
    where the first inequality holds by \Cref{lem:apxTriangleInequality2} with $\gamma=2$. 
    Therefore, as the cluster $C^*_o$ has cost less than $5\eps^3 \sopt/\log n$ for $o\in \opt_1$, we have 
    \[
        |C^*_o \cap B(o, d_o\sqrt{15\eps})| \geq (1-\eps) |C^*_o|\,. 
    \]
    Indeed, otherwise, we would have the contradiction
    \[
         \sum_{p \in C^*_o} \dist^2(p, o) \geq |C^*_o - B(o, d_o\sqrt{15\eps})|\cdot 15\eps\cdot d^2_o > |C^*_o| 15\eps^2 d^2_o  \geq 5\eps^2\beta \sopt = 5\eps^3 \frac{\sopt}{\log n}\,.
    \]
    Now let $\widetilde{\opt}_1$ be the subset of $\opt_1$ so that for every $o\in \widetilde{\opt}_1$, we have $B(o, d_o\sqrt{40\eps}) \cap S = \emptyset$. We claim that $|\widetilde{\opt}_1| \leq \frac{\log n}{5\eps^3}$. 
    Indeed, for each $o\in \widetilde{\opt}_1$, the cost of each point $p\in C^*_o \cap B(o, d_o\sqrt{15\eps})$ in the solution $S$ would be at least $\frac{40\eps}{2}d^2_o-15\eps\cdot d^2_o=5\eps \cdot d^2_o$ by \Cref{lem:apxTriangleInequality2} with $\gamma=2$. Hence the total cost of the points in $C^*_o \cap B(o, d_o\sqrt{15\eps})$ in the solution $S$ would be at least $5\eps 
    d^2_o (1-\eps) |C^*_o| 
 \geq 5(1-\eps)\frac{\eps^2}{3} \frac{\sopt}{\log n}$. Therefore, using the fact that local search gives a $\apxLSkmeans$-approximate solution, we must have $|\widetilde{\opt}_1| \leq \frac{3\cdot \apxLSkmeans \log n}{5(1-\eps)\eps^2} \leq  \frac{\log n}{5\eps^3}$.

    We define $\opt_2 = \opt_1 - \widetilde{\opt}_1$ and so now we have $B(o, d_o\sqrt{40\eps}) \cap S \neq \emptyset$ for every $o\in \opt_2$. Now let $\widetilde{\opt}_2$ be the subset of  $\opt_2$ so that  every $o\in \widetilde{\opt}_2$ is such that $|B(o, d_o/3) \cap S| > 1$. By the pigeon hole principle we must have that $|\widetilde{\opt}_2| \leq |\opt| - |\opt_2| =|\widetilde{\opt}|+|\widetilde{\opt}_1| \leq \frac{2\log n}{5\eps^3} $.     
    
    Let $\opt_3:=\opt_2-\widetilde{\opt}_2$. For each center $o\in \opt_3$ there is thus a center $c(o)\in S$ such that $d(o, c(o)) \leq d_o\sqrt{40\eps}$ and any other center $c'\in S$ with $c'\neq c(o)$  satisfies $d(o, c') \geq d_o/3$. It follows that all points in $B(o,  d_o/4)$ are assigned to $c(o)$ {in $S$}. In other words, at least a  $(1-\eps)$ fraction of the points of $C^*_o$ are assigned to the cluster with center $c(o)$ in $S$. Now let $\widetilde{\opt}_3$ be those centers in $\opt_3$ so that for every $o \in \widetilde{\opt}_3$ the associated center $c(o) \in S$ is assigned more than $\eps |C^*|$ points from outside the ball $B(o, d_o/4)$. The cost of the cluster of $c(o)$ in $S$ is then at least
\[
    \left(\frac{d_o}{5}\right)^2 \eps |C^*| \geq \frac{\eps^2}{75 \log n} \sopt.
\]
It follows that $\widetilde{\opt}_3$ has cardinality at most $75\cdot \apxLSkmeans \frac{\log n}{\eps^2} \leq \frac{\log n}{5\eps^3}$. Our final set is $\opt_4:=\opt_3-\widetilde{\opt}_3$.

Now for each center in $c \in \opt_4$, {with associated cluster $C^*$}, we have that it  together with its associated center $c(o) \in S$, {with associated cluster $C'$}, satisfies the following:

\begin{itemize}
    \item All clients in $B(o,  d_o/4)$, i.e. at least $(1-\eps) |C^*|$ many clients, are both assigned to $o$ and $c(o)$ in $\opt$ and $S$, respectively. Notice that this implies $\min\{|C^*|,|C'|\}\geq (1-\eps)|C^*|$.
    \item The additional clients assigned to $o$ and not to $c(o)$ is at most $\eps |C^*|$.
    \item The additional clients assigned to $c(o)$ and not $o$ is at most $\eps |C^*|$.
\end{itemize}
Notice that $|C^*\Delta C'|\leq 2\eps |C^*|\leq 3\eps\min\{|C^*|,|C'|\}$, in particular $C^*$ is pure. It follows that all centers in $\opt_4$ correspond to pure clusters. Furthermore, $\opt_4$ was obtained by removing at most $\frac{\log n}{\eps^3}$ centers. The claim follows.   
\end{proof}

Let $C^*$ be a pure cluster of $\opt$ and $C'$ be the associated cluster of $S$. Let
$c^*$ be the center of $C^*$ and ${c'}$ be the center of ${C'}$. We define $t(c^*) := {c'}$.
For any client
$p \in C^*$, let $r(p) := \dist^2(p, {c'}) = \dist^2(p, t(c^*))$.
We have the following lemma which is the only lemma
that requires that the solution $S$ is obtained via the local search (the proof of the previous lemma only used that it was a constant-factor approximation). In words, it says that $S$ approximates the connection cost of clients belonging to pure clusters of the optimal solution almost perfectly.

\begin{lemma}[Adapted from~\cite{Cohen-AddadS17}]
\label{lem:purecost}
    Let $\clientspure \subseteq \clients$ be the subset of clients that belong to pure clusters  of the optimal solution $\opt$. We have
    \begin{gather*}
        \sum_{p\in \clientspure} r(p) \leq \sum_{p\in \clientspure} \opt_p + O(\sqrt{\eps} \cdot \sopt)\,.
    \end{gather*}
\end{lemma}
\begin{proof}
Recall that $S_p = \dist^2(p,S)$. Consider a pair of pure clusters $C^*$ and $C'$ with the same notation as above. We prove the following:
     \begin{enumerate}
         \item $\sum_{p \in C^* \cap C'} r(p) = \sum_{p \in C^* \cap C'} \dist^2(p,S) \le  \sum_{p \in C^* \cap C'} \left((1+12\sqrt{\eps}) \opt_p + 12\sqrt{\eps} S_p\right)+\sqrt{\eps}\sum_{p\in C'}S_p$; and
         \item  
         $\sum_{p \in C^* {-} C'} r(p) {=} \sum_{p \in C^*{-} C'} \dist^2(p, c') \le \sum_{p \in C^*{-} C'} (1+\sqrt{\eps})\opt_p + 12\sqrt{\eps}\sum_{p \in C^*} (\opt_p + {S_p})$.
     \end{enumerate}
     The statement of the lemma then follows by summing up the above bounds and using that $\clcost(S) \leq \apxLSkmeans \cdot {\sopt}$ by \Cref{lemma:localsearchis5approximation}.

For proving the first bullet, consider the swap that swaps in $c^*$ and swaps out the center $c'$. The cost change for the clients in $C^* \cap C'$ is exactly $\sum_{p \in C^*\cap C'} (\opt_p - S_p)$.
For any $p \not\in C^*$, if $p$ is not served by ${c'}$ in $S$ the cost is unchanged. Otherwise the cost of 
$p \in C' - C^*$ in the new 
solution $S - \{{c'}\} \cup \{c^*\}$ 
is, by \Cref{lem:apxTriangleInequality3} with $\gamma=1+\sqrt{\eps}$, at most $(1+\sqrt{\eps})S_{p} + \frac{2+2/\sqrt{\eps}}{|C^*\cap C'|} \sum_{p' \in C^*{\cap C'}} (S_{p'} + \opt_{p'})$.
Summing over all clients $p \in C'-C^*$, this is at most 
$\sum_{p \in C'-C^*} (1+\sqrt{\eps})S_{{p}} + \frac{3}{\sqrt{\eps}}\frac{|C'-C^*|}{|C^*\cap C'|} \sum_{p' \in C^*{\cap C'}} (S_{p'} + \opt_{p'})$. Thus, the cost difference $\clcost(S- \{c'\} \cup \{c^*\}) - \clcost(S)$
for these clients is at most 
$$
\sum_{p\in C'-C^*}\sqrt{\eps}S_p+\frac{3}{\sqrt{\eps}}\frac{|C'-C^*|}{|C^*\cap C'|} \sum_{p' \in C^*{\cap C'}} (S_{p'} + \opt_{p'})\leq \sum_{p\in C'}\sqrt{\eps}S_p+ 12\sqrt{\eps} \sum_{p' \in C^*{\cap C'}} (S_{p'} + \opt_{p'}),$$
where in the inequality we used the fact that $C^*$ is pure, {hence $\frac{|C'-C^*|}{|C^*\cap C'|}\leq \frac{{3\eps}|C^*|}{(1-\eps)|C^*|}\leq 4\eps$.} The first bullet follows since, by local optimality, $\clcost(S- \{c'\} \cup \{c^*\}) - \clcost(S) \ge {0}$. 

We turn to the second bullet. {By triangle inequality, for each $p\in C^*-C'$ and each $p'\in C^*\cap C'$, one has $\dist(p,c')\leq \dist(p,c^*)+\dist(p',c^*)+\dist(p',c')=\sqrt{\opt_p}+\sqrt{\opt_{p'}}+\sqrt{S_{p'}}$. Hence,} by \Cref{lem:apxTriangleInequality3} with $\gamma=1+\sqrt{\eps}$ {and by averaging over $p'\in C^*\cap C'$}, one has 
$$
\dist^2(p,c')\leq (1+\sqrt{\eps})\opt_p + \frac{3/\sqrt{\eps}}{|C^* \cap C'|} \sum_{p' \in C^* {\cap} C'} (S_{p'} + \opt_{p'}).
$$
Summing up over all clients
$p \in C^* - C'$, we have that
\begin{align*}
\sum_{p \in C^*{-} C'} \dist^2(p, c') & \le  \sum_{p \in C^*- C'} (1+\sqrt{\eps})\opt_p + \frac{3}{\sqrt{\eps}}\frac{|C^*{-} C'|}{|C^*\cap C'|} \sum_{p' \in C^*{\cap}C'} (\opt_{p'} + S_{p'})\\
& \leq  \sum_{p \in C^*- C'} (1+\sqrt{\eps})\opt_p + 12\sqrt{\eps} \sum_{p' \in C^*} (\opt_{p'} + S_{p'}),
\end{align*}
where the second inequality follows from $\frac{|C^*{-} C'|}{|C^*\cap C'|}\leq {4\eps}$ similarly to the first bullet. The second bullet follows.
\end{proof}

\subsection{Cheap Clusters and $D$-Sample Process to Hit Expensive Clusters}
\label{sec:dsampleproc}
In the previous section, we showed that the cost of pure clusters in $\opt$ is close to to their optimal cost in our local search solution $S$. Our goal in this section is to ``hit'' those non-pure clusters of $\opt$
 that have a high cost in $S$.
We say that a non-pure cluster $C^*$ of $\opt$ with center $c^*$ is \emph{basic-cheap} if the total cost in 
$S$ of the clients in $C^*_\avg$ is less than $\eps^{5}\sopt/\log n$ or 
if there exists a center of $S$ at squared distance at most
{$(1+\eps)\avg_{C^*,\opt} + \eps \gamma_{C^*}$} from $c^*$, where $\gamma_{C^*} := \frac{1}{|C^*|} \sum_{p\in C^*} (\opt_p + S_p)$; in the latter case we also say that $C^*$ is \emph{covered} by $S$.

The following procedure,  that aims to ``hit'' all clusters of $\opt$ that are non-pure and non-basic-cheap,  is inspired by the classic $k$-means++ algorithm and its
    so called $D^2$-sampling procedure to
    sample clients proportional to their cost in solution $S$. 
\begin{mdframed}[hidealllines=true, backgroundcolor=gray!15]
\vspace{-5mm}
\paragraph{$s$-$D^2$-Sample Process. }\ \\
\begin{enumerate}
    \item Let $W$ be the set obtained by taking $s$ independent samples of clients where, in each sample, client  $p$ is sampled with probability $\frac{\dist^2(p, S)}{\clcost(S)}$.
    \item Output $W$. 
\end{enumerate}
\end{mdframed}
Our algorithm obtains a set $W$ of clients by sampling $s^*$ using the $s^*$-$D^2$-sample process, where
\[
    s^* = \apxLSkmeans \cdot \frac{\log n}{\eps^{5}}\ln\left( \frac{{81}}{\eps^{2}}\right)\,.
\]
We have the following lemma that says that we hit all non-pure and non-basic-cheap clusters except for a set $V$ of clusters that have an insignificant cost in $S$.

For a {non-pure} cluster $C^*$ of $\opt$ with center $c^*\in \opt$, we let $t(c^*)$ be the closest center in $S$ to $c^*$, and for any $p\in C^*$, we let {the replacement cost of $p$ be} $r(p):=\dist^2(p,t(c^*))$. 
\begin{lemma}
\label{lem:probcost}
Consider the random set $W$ of clients obtained from the $s^*$-$D^2$-sample process.  Let $V$ be the set of clusters $C^*$ of $\opt$ that are non-pure and non-basic-cheap  such that $C^*_\avg  \cap W = \emptyset$. 
With probability at least $1-\eps$,
\begin{align}\sum_{C^* \in V} \sum_{p \in C^*} {r(p)} \le \eps \cdot \sopt\,.
\label{eq:probcost}
\end{align}
\end{lemma}

\begin{proof}
Recall that $S_p = \dist^2(p,S)$. Consider a non-pure and non-basic-cheap cluster $C^*$ of $\opt$.
We say that a sampled point $p\in W$ \emph{hits} $C^*$ if $p\in C^*_{avg}$ (i.e., $p$ is a leader of $C^*$). 
As $C^*$ is not basic-cheap, we have $\sum_{p\in C^*_{avg}}S_p\geq \eps^{5} {\sopt}/\log n$. Thus the probability that a sampled $p\in W$ hits $C^*$ is at least $\frac{\sum_{p\in C^*}S_p}{\clcost(S)}\geq \frac{\eps^{5}}{{\apxLSkmeans} \log n}$, where we used that $\clcost(S) \leq \apxLSkmeans \sopt$ (\Cref{lemma:localsearchis5approximation}). Thus the probability that $C^*$ is not hit by any point in $W$ is at most $(1-\frac{\eps^{5}}{{\apxLSkmeans} \log n})^{s^*}$ which by the selection of $s^*$ is at most  $\eps^{2}/{81}$. {For a point $p\in C^*$ which is assigned to $c'$ in $S$, by triangle inequality one has
\begin{align*}
\sqrt{r(p)} & =\dist(p,t(c^*))\leq \dist(p,c^*)+\dist(c^*,t(c^*))\leq \dist(p,c^*)+\dist(c^*,c') \\
& \leq \dist(p,c^*)+\dist(p,c^*)+\dist(p,c') =2\sqrt{\opt_p}+\sqrt{S_p}.  
\end{align*}
Thus, by applying Lemma \ref{lem:apxTriangleInequality3} with $\gamma=3$, one gets
$$
r(p)\leq 6\opt_p+3S_p.
$$
}
Therefore, by linearity of expectation,
$$
E[\sum_{C^*\in V}\sum_{p\in C^*}{r(p)}] \leq \sum_{p\in \clients} {(6\opt_p+3S_p)}\cdot \frac{\eps^{2}}{{81}} \leq \eps^{2}\cdot \sopt\,,
$$
where we again used that $\clcost(S) \leq \apxLSkmeans \sopt$ by \Cref{lemma:localsearchis5approximation}.
The total cost of the clusters in $V$ is thus at most $ \eps\cdot {\sopt}$ with probability at least $1-\eps$ by Markov's inequality. 
\end{proof}

\paragraph{Successful $W$.}
We say that the set of clients $W$ obtained from the $s^*$-$D^2$-sample process is a 
\emph{successful} sample  if~\eqref{eq:probcost} holds.
We define the following two types of $\opt$ clusters. We say that 
a non-pure cluster $C^*$ of $\opt$ with center $c^*$ is \emph{cheap} if it is basic-cheap 
or is in $V$; a  non-pure cluster is 
\emph{expensive} otherwise.  In other words, a cluster $C^*$ of the optimal solution is expensive if it is a non-pure and non-basic-cheap cluster such that $C^*_{avg} \cap W \neq \emptyset$. 

Assuming that $W$ is successful, we have a good bound on the cost of the clients assigned to cheap centers in the optimal solution. 
\begin{lemma}
\label{lem:cheapcost}
    Let $\clientscheap \subseteq \clients$ be the subset of clients that belong to cheap clusters of the optimal solution $\opt$. If $W$ is a successful sample, 
    \begin{gather*}
        \sum_{p\in \clientscheap} r(p) \leq \sum_{p\in \clientscheap} 4\,\opt_p + O(\sqrt{\eps} \cdot \sopt)\,.
    \end{gather*}
\end{lemma}
\begin{proof}
We claim that for any cluster $C^*$ of $\opt$ that is {basic-}cheap, and   any $p \in C^*$:
    \begin{enumerate}
    \item If $C^*$ is not covered, then
      $r(p) = \dist^2(p, t(c^*)) \le (2+\sqrt{\eps})(\opt_p + \avg_{C^*,\opt}) + \frac{3}{\eps^{1.5} |C^*|} \sum_{p \in C^*_{\avg}}  S_p$. 
     \item If $C^*$ is covered, then
        $r(p) \le 2\opt_p + 2(1+\eps) \avg_{C^*,\opt} + 2\eps \gamma_{C^*}$.
    \end{enumerate}
If $C^*$ is covered, the bound 2. follows immediately by \Cref{lem:apxTriangleInequality2} with $\gamma=2$. 
Otherwise, observe that for any such cluster and for any $p'\in C^*_{\avg}$, one has $\sqrt{r(p)}=d(p,t(c^*))\leq \dist(p,c^*)+d(c^*,p')+d(p',t(c^*))=\sqrt{\opt_p}+\sqrt{\opt_{p'}}+\sqrt{S_{p'}}$. Thus applying \Cref{lem:apxTriangleInequality3} with $\gamma=1+\frac{2}{\sqrt{\eps}}\leq \frac{3}{\sqrt{\eps}}$, one gets $d^2(p,t(c^*))\leq (2+\sqrt{\eps})(\opt_p+\opt_{p'})+\frac{3}{\sqrt{\eps}}S_{p'}$. Averaging over $p'\in C^*_{\avg}$, we obtain
\begin{align*}
d^2(p,t(c^*)) & \leq (2+\sqrt{\eps})\opt_p+\frac{2+\sqrt{\eps}}{|C^*_{\avg}|}\sum_{p'\in C^*_{\avg}}\opt_{p'}+\frac{3}{\sqrt{\eps}|C^*_{\avg}|}\sum_{p'\in C^*_{\avg}}S_{p'}\\
& \leq (2+\sqrt{\eps})(\opt_p+\avg_{C^*,\opt})+\frac{3}{\eps^{1.5}|{C^*}|}\sum_{p'\in C^*}S_{p'},
\end{align*}
where in the second inequality we used $|C^*_{\avg}| \ge \eps |C^*|$. The bound 1. follows.

Now, to prove the above lemma, we sum up the cost of all the clients of cheap clusters $C^*$. For a covered cluster $C^*$, {by 2.} we have  
    \begin{align*}
        \sum_{p\in C^*} r(p)& \leq 2\sum_{p\in C^*} \left(\opt_p + (1+\eps)\avg_{C^*, \opt} + \eps \gamma_{C^*}\right)\\
        &\leq 2\left( 1+ (1+2\eps) (1+\eps)\right) \sum_{p\in C^*} \opt_p + 2\eps \sum_{p\in C^*} (\opt_p + S_p)\,,
    \end{align*}
    where we used the definition of $\gamma_{C^*} = \frac{1}{|C^*|} \sum_{p\in C^*} (\opt_p + S_p)$ and that $|C^*| \cdot \avg_{C^*, \opt} \leq (1+2\eps) \sum_{p\in C^*} \opt_p$, which holds because $C^*$ has at least $(1-\eps) |C^*|$ clients $p\in C^*$ so that $\opt_p \geq \avg_{C^*, \opt}$.
    For a basic-cheap cluster $C^*$ {which is not covered}, {by 1. one has}
    \begin{align*}
        \sum_{p\in C^*} r(p) & \leq \sum_{p\in C^*} \left((2+\sqrt{\eps})(\opt_p + \avg_{C^*,\opt}) + \frac{3}{\eps^{1.5} |C^*|} \sum_{p \in C^*_{\avg}}  S_p\right)\\
        &\leq (4+O(\sqrt{\eps})) \sum_{p\in C^*} \opt_p +  \frac{3\eps^{3.5}\sopt}{\log n}\,,
    \end{align*}
    where we again used that $|C^*| \cdot \avg_{C^*, \opt} \leq (1+2\eps) \sum_{p\in C^*} \opt_p$ and we additionally used that, by the definition of basic-cheap clusters, $\sum_{p\in C^*_{avg}} S_p \leq \eps^{5} \sopt/\log n$.

    Finally, we have that if we sum up all clusters in $V$, i.e., clusters that are cheap but {not basic-cheap}, then \Cref{lem:probcost} says
    \begin{gather*}
        \sum_{C^* \in V} \sum_{p \in C^*} {r(p)} \le \eps \cdot \sopt\,.
    \end{gather*}
    The statement of the lemma now follows by summing up the above bounds for all cheap clusters and using that $\clcost(S) \leq \apxLSkmeans\cdot \sopt$ by \Cref{lemma:localsearchis5approximation} and the property that there are at most $\log(n)/\eps^3$ basic-cheap clusters. {The latter claim} is true because \Cref{lem:numnonpure} says that there are at most $\log(n)/\eps^3$ non-pure clusters and basic-cheap clusters are non-pure clusters by definition.
\end{proof}

Given $S$ and $W$, the set of clusters of the optimal solution $\opt$ is thus partitioned into pure clusters, cheap clusters, and expensive clusters. We let $\ho$ be the centers of $\opt$ that are expensive. We further define the set $\clientsexpensive \subseteq \clients$  to be the subset of clients that belong to the expensive clusters in $\opt$. We have thus partitioned $D$ into sets $\clientspure, \clientscheap$, and $\clientsexpensive$ depending on the type of optimal cluster they belong to. \Cref{lem:purecost} bounds the cost of the clients in $\clientspure$ in $S$ and \Cref{lem:cheapcost} bounds the cost of the clients in $\clientscheap$. For future reference, we summarize these two lemmas  (by weakening the upper bound for clients in $\clientspure$).
\begin{lemma}
Assuming  $W$ is a successful sample, 
    \[
    \sum_{p\in D - \clientsexpensive} r(p) \leq \sum_{p \in D - \clientsexpensive} {4} \opt_p + O({\sqrt{\eps}} \cdot \sopt)\,.
    \]
    \label{lemma:convenience_upper_bound}
\end{lemma}

It follows that the connection cost in solution $S$ of clients in $\clients - \clientsexpensive$ is {roughly} within a factor ${4}$ of their cost in the optimal solution. The remaining part of the algorithm is thus devoted to modifying $S$ to obtain a small connection cost of the clients in $\clientsexpensive$ without significantly increasing the cost of the other clients.

\subsection{Identifying Balls of Expensive Clusters}
\label{sec:ballguesses}

By the definition of expensive clusters, each expensive cluster $C^*$ of $\opt$ satisfies $C^*\cap W \neq \emptyset$ where $W$ is the sample obtained by running the sampling procedure of the last subsection. Here, our goal is to guess a subset $\calB$ of balls so that for each expensive cluster $C^*$ there is a ball $B(\ell, {\sqrt{\rho}}) \in \calB$ so that $\ell$, which we refer to as a leader, is in  $C^*_{avg}$ and $\rho$ is a very close approximation to $\avg_{C^*, \opt}$. This guarantees that the center $c^* \in \ho$ of cluster $C^*$ is in $B(\ell, {\sqrt{\rho}})$ and that no other center in $B(\ell, {\sqrt{\rho}})$ is "too" far from $c^*$ since the radius of the ball is $\approx {\sqrt{\avg_{C^*, \opt}}}$. The balls $\calB$ will then be used when we approximate the centers in $\ho$ via submodular function maximization subject to a partition matroid constraint (the balls will correspond to the partitions of the matroid). 
As the algorithm does not know the set of expensive clusters, it enumerates all subsets of $W$, and for each guessed point, it enumerates a constant number of radii.

\begin{mdframed}[hidealllines=true, backgroundcolor=gray!15]
\vspace{-5mm}
\paragraph{Ball Procedure. 
}\ \\

\begin{enumerate}
\item $\calL_\texttt{bal} \gets \emptyset$. 
 \item For each subset $\{p_1, p_2, \ldots, p_q\} \subseteq W$:
 \begin{enumerate}
     \item For all integers $i_1, i_2, \ldots, i_q$ such that
     \[\eps^3  {S_{p_j}} \le (1+\eps^3)^{i_j} \le {S_{p_j}}/\eps^3
     \qquad \mbox{for $j=1, 2,\ldots, q$}\]
     add $\calB = \left\{B(p_1, {\sqrt{(1+\eps^3)^{i_1}}}), B(p_2, {\sqrt{(1+\eps^3)^{i_2}}}), \ldots, B(p_q, {\sqrt{(1+\eps^3)^{i_q}}}\right\}$ to $\calL_{\texttt{bal}}$.
 \end{enumerate}
\end{enumerate}
\end{mdframed}

We show that the ball procedure runs in polynomial time and that  one of the subsets $\calB$ of balls that the procedure constructs is a correct guess: every expensive cluster $C^*$ has a leader, and $\avg_{C^*, \opt}$ is guessed approximately correct.
We say that a set of balls $\calB$ is a \emph{valid} set of balls if it satisfies the following: For each expensive cluster $C^*$ of $\opt$,  we have a ball $B(\ell, {\sqrt{\rho}}) \in \calB$ such that $\ell\in C^*_{avg}$ is a leader and $$\avg_{C^*,\opt} \leq \rho \leq  \avg_{C^*,\opt} + \eps \gamma_{C^*}, $$
where we recall that $\gamma_{C^*} = \frac{1}{|C^*|} \sum_{p \in C^*} (\opt_p + S_p)$.

\begin{lemma}
\label{lem:ballguesses} 
The ball procedure runs in polynomial time $n^{\eps^{-O(1)}}$ and produces a collection $\calL_{\texttt{bal}}$ of sets of balls such that {at least} one set $\calB\in \calL_{\texttt{bal}}$ is a valid set of balls.
\end{lemma}
\begin{proof}

We start by showing that the running time is bounded by $n^{\eps^{-O(1)}}$. Indeed, we have $|W| = s^*$, and so the number of subsets of $W$ is $2^{s^*}= n^{\eps^{-O(1)}}$ since $s^* = {\apxLSkmeans} \cdot \frac{\log n}{\eps^{5}}\ln\left( \frac{{\apxLSkmeans}}{\eps^{2}}\right)$. Moreover, for a fixed subset $\{p_1, p_2, \ldots, p_q\}$, we consider a constant number  $\alpha = \log_{1+\eps^3}(1/\eps^6)$ of balls (i.e., values of $i_j$) for each $p_j$. So the total number of balls considered added for each subset $\{p_1, p_2, \ldots, p_q\}$ is $\alpha^{q}$ which is $n^{\eps^{-O(1)}}$ since $q \leq |W| \leq s^*$.

We now turn our attention to showing that there is a set of balls $\calB \in \calL_{\texttt{bal}}$ that satisfies the properties of the lemma.
Let $C^*(1), \ldots, C^*(|\ho|)$ be the expensive clusters of \opt. 
By definition, we have $C^*_{avg}(j) \cap W\neq \emptyset$ for each such cluster. Now arbitrarily fix one client $\ell_j\in C^*_{avg}(j) \cap W$  for $j=1,2, \ldots, |\ho|$, which we refer to the leader of $C^*(j)$. Since the algorithm enumerates over all subsets
    of $W$, there exists an execution of the algorithm where
    the subset is the set $\{\ell_1, \ldots, \ell_{|\ho|}\}$ of leaders for all the expensive clusters.
    We consider the set  $$\calB = \left\{B(\ell_1, {\sqrt{(1+\eps^3)^{i_1}}}), B(\ell_2, {\sqrt{(1+\eps^3)^{i_2}}}), \ldots, B(\ell_{|\ho|}, {\sqrt{(1+\eps^3)^{i_{|\ho|}}}})\right\}$$ of balls, where $i_j$ is the smallest integer so that $(1+\eps^3)^{i_j} \geq \max\{\avg_{C^*, \opt}, \eps^3 S_{\ell_j}\}$. 

    It remains to prove that, for each leader $\ell_j$, we have $(1+\eps^3)^{i_j} \leq S_{\ell_j}/\eps^3$  and $\avg_{C^*(j),\opt} \leq (1+\eps^3)^{i_j}\leq  \avg_{C^*(j),\opt} + \eps \gamma_{C^*(j)}$. The first property guarantees that  $\calB \in \calL_{\texttt{bal}}$ and the second property is the bounds on $\rho$ for $\calB$ to be a valid set of balls.

     We now verify these two properties for each leader $\ell_j$. To simplify notation, we let $\ell = \ell_j$, $C^* = C^*(j)$ and $\rho = (1+\eps^3)^{i_j}$.
    We have
    ${d^2}(\ell, S) = S_{\ell} > \eps \avg_{C^*,{\opt}}$ since otherwise
    $C^*$ would be covered by $S$, which would contradict that $C^*$ is an expensive cluster. By the selection of $i_j$,  $\rho = (1+\eps^3)^{i_j}\leq (1+\eps^3) \max\{\avg_{C^*, \opt}, \eps^3 S_\ell\} \leq (1+\eps^3) S_\ell/\eps$ and thus $(1+\eps^3)^{i_j} \leq S_\ell/\eps^3$ as required. 
    
    We further claim that $\rho \in [\avg_{C^*,\opt},  \avg_{C^*,\opt} + {\eps^2} S_\ell)$. 
    Assume first that $\avg_{C^*,\opt}\geq \eps^3 S_\ell$. Then $\rho = (1+\eps^3)^{i_j}\geq\avg_{C^*,\opt}$, for the smallest possible integer $i_j$ which satisfies this condition. The claim then follows  since $\rho \leq (1+\eps^3)\avg_{C^*,\opt}\leq \avg_{C^*,\opt}+\eps^2 S_\ell$ as required (where we used that $\avg_{C^*, \opt} \leq S_\ell/\eps$). In the complementary case, namely $\avg_{C^*,\opt}< \eps^3 S_\ell$, one has that $i_j$ was selected to be the smallest integer so that $(1+\eps^3)^{i_j} \geq \eps^3 S_\ell$. Moreover,  the interval $[\avg_{C^*,\opt},  \avg_{C^*,\opt} + \eps^2 S_\ell)$ contains the interval $[\eps^3 S_\ell,\eps^2 S_{\ell}]$, and the latter interval contains at least one power of $(1+\eps^3)$. We thus have
    \[
        \avg_{C^*, \opt} \leq \rho \leq \avg_{C^*, \opt} + \eps^2 S_{\ell}\, .
    \]
    
    We conclude the proof by showing $S_\ell \le ({6}+{6}\eps)\gamma_{C^*}{\leq \gamma_{C^*}/\eps}$. Let $\ell$ be assigned to $c'$ in $S$, and consider any $p\in C^*$ which is assigned to $c''$ in $S$. Then, {by \Cref{lem:apxTriangleInequality3} with $\gamma=3$,} one has $S_\ell={d^2}(\ell,c')\leq {d^2}(\ell,c'')\leq {3d^2}(\ell,c^*)+{3d^2}(c^*,p)+{3d^2}(p,c'')={3d^2}(\ell,c^*)+{3}\opt_p+{3}S_p$. Hence, averaging over $C^*$, one gets $S_\ell \le {3d^2}(\ell, c^*) + \frac{{3}}{|C^*|} \sum_{p \in C^*} (S_p + \opt_p)$. Since $\ell$ is a leader, i.e., $\ell \in C^*_{avg}$, ${d^2}(\ell,c^*)\leq \avg_{C^*,\opt}\leq  \frac{1+2\eps}{|C^*|}\sum_{p \in C^*}\opt_p$, where we used that there is at least $(1-\eps)|C^*|$ clients $p\in C^*$ so that ${d^2}(p, c^*) \geq \avg_{C^*, \opt}$. Therefore, $S_\ell \leq ({6}+{6}\eps) \gamma_{C^*}$, which, as aforementioned, concludes the proof of the lemma. 
\end{proof}

We thus have that the ball procedure produces a family $\calL_{\texttt{bal}}$ such that {at least one} set $\calB \in \calL_{\texttt{bal}}$ satisfies the conditions of the lemma.  While our algorithm will try all possible sets in $\calL_{\texttt{bal}}$ (since it does not know which one is valid), we do our analysis  by considering the run of the algorithm when it selects this set  $\calB$ of valid balls.

\subsection{\emph{Dummy} Centers and Mixed Solutions $M_O$ and $M_D$}
\label{sec:dummycenters}

    Next, the algorithm creates, for each ball  $B(\ell, {\sqrt{\rho}}) \in \calB$ with center $\ell$ and
    radius ${\sqrt{\rho}}$, 
    a \emph{dummy} center $\delta$. The center $\delta$ is at distance ${\sqrt{\rho}}$ from $\ell$ and and at distance 
   ${\sqrt{\rho}} + d(\ell, p)$ from any other input point $p$.

    Let $\Lambda$ be the set of dummy centers, so $|\Lambda| = |\calB|$. Recall that  $\clientsexpensive$ is the set of clients that are in an expensive cluster of $\opt$, i.e., served by a center
in $\ho$ in the optimal solution.

\begin{lemma}
\label{lem:structSminusS0}
Suppose the sample $W$ is successful, and that the set of balls $\calB$ is valid, then there exists a set of centers $S_0$ of $S$ that satisfies the following properties:
  \begin{enumerate}
  \item $|S_0| = |\ho| \le \frac{\log n}{\eps^3}$, and;
  \item $\forall c \in S_0$, there is no $c^* \in \opt-\ho$ such that $t(c^*) = c$, and;
  \item $\clcost(S - S_0 \cup \dummyset) \le
  {9} \sum_{p \in \clientsexpensive}\opt_p + \sum_{p \in \clients {-} \clientsexpensive}
  r(p) +
  O(\eps \cdot{\sopt})) \le ({9}+O({\sqrt{\eps}}))\sopt$.
  \end{enumerate}
  \label{lemma:S0properties}
\end{lemma}
\begin{proof}
By definition, the centers in $\hat{\opt}$ are not pure, hence $|\hat{\opt}|\leq \frac{\log n}{\eps^3}$ by \Cref{lem:numnonpure}.
Furthermore, since there are $|\opt - \ho|$ centers of $\opt$ not
in $\ho$, there are (at least) $|\ho|$ centers $c$ of $S$ such that there is no $c^* \in \opt-\ho$ such that $t(c^*) = c$, and we let $S_0$
refer to these centers: the first 2 bullets follow immediately.

Let us turn to the last bullet. The bound on the cost 
of the clients that are not in $\clientsexpensive$ 
follows from the second bullet: for 
each client $p$ in a cluster of the optimum solution
whose center is $c^* \in \opt-\ho$, we have
that $t(c^*) \in S-S_0$ and so its cost
is at most $r(p)$.
We finally consider the points in $\clientsexpensive$.
Since the balls are valid, \Cref{lem:ballguesses} implies that 
for each expensive cluster $C^*$ of $\opt$, there is a 
ball $B(\ell,{\sqrt{\rho}})$  where $\rho \in 
[\avg_{C^*,\opt},  \avg_{C^*,\opt} + \eps \gamma_{C^*})$ and 
$\ell \in C^*_\avg$. Therefore, if we let $c^*$ be the center of $C^*$ 
{and using \Cref{lem:apxTriangleInequality3} with $\gamma=3$},
\begin{align*}
\sum_{p\in C^*} {d^2}(p, \dummyset) & \leq \sum_{p\in C^*} \left({3}d^2(p, c^*) + {3}d^2(c^*, \ell) +  {3}d^2(\ell, \dummyset) \right)\\
&\leq\sum_{p\in C^*} \left({3}\opt_p  + {3}\avg_{C^*,\opt}+ {3}\rho \right)\\
     & \leq \sum_{p\in C^*} \left({3}\opt_p + {6}\avg_{C^*,\opt} + {3}\eps \gamma_{C^*}) \right) \\
    & = \sum_{p\in C^*} {3}\opt_p + |C^*|({6}\avg_{C^*,\opt} + {3}\eps \gamma_{C^*})\\
    & \le  ({9}+O(\eps)) \sum_{p\in C^*} \opt_p
    + O(\eps \sum_{p \in C^*} {S_p}).
\end{align*}
The lemma follows from the fact
 $\clcost(S) \leq {\apxLSkmeans} \cdot {\sopt}$ (\Cref{lemma:localsearchis5approximation}) and by
invoking \Cref{lemma:convenience_upper_bound} to bound the sum
of the $r(p)$ values (where we use the assumption that the sample $W$ is successful).
\end{proof}

\paragraph{Definition of the mixed solutions $M_O$ and $M_D$.}
Next, let us analyze
the cost of the swap $(\ho, S_0)$. Namely, of the solution $M_O = S - S_0 \cup 
\hat{\opt}$, where $M$ in $M_O$ stands for  ``mixed solution''  and the subscript $O$ stands for that the mixed solution  is obtained by swapping in some elements from $\opt$  (and removing $S_0$). We will also analyze the ``dummy'' version of $M_O$ where instead of swapping in $\ho$, we swap in the dummy centers ${\dummyset}$. We denote that solution by $M_D = S - S_0 \cup \dummyset$, where subscript $D$ stands for ``dummy''.  
\begin{lemma}
    We have  $d^2(p, \ho) \leq \opt_p$ if $p\in \clientsexpensive$ and $d^2(p, S- S_0)\leq r(p)$ if $p\in\clients - \clientsexpensive$.
    \label{lemma:costsMO}
\end{lemma}
\begin{proof}

  We have $d^2(p, \ho) \leq \opt_p$ when $p\in \clientsexpensive$ because the center serving $p$ in the optimal solution is in $\ho$. Consider the remaining case and let $p\in \clients - \clientsexpensive$ be a client whose serving center in $\opt$ is say $c^*$. 
  By the second condition on $S_0$ in \Cref{lemma:S0properties}, the center $t(c^*)$ is in $S - S_0$  and
so  $d^2(p, S- S_0) \leq d^2(p, t(c^*)) = r(p)$.  
\end{proof}

Notice that $M_O$ contains $S- S_0$ and $\ho$. Therefore, if we sum up the above bounds for all clients we get
\begin{align*}
    \clcost(M_O) & \leq \sum_{p\in \clientsexpensive} \opt_p  + \sum_{p\in \clients {-} \clientsexpensive} r(p)\,, 
\end{align*}
which by \Cref{lemma:convenience_upper_bound} is at most $\sum_{p\in \clientsexpensive} \opt_p  + \sum_{p\in \clients {-}\clientsexpensive} {4} \opt_p + O({\sqrt{\eps}}) \cdot \sopt$ (assuming that that the sample $W$ was successful). By swapping $S_0$ with $\ho$ we thus get a $1$-approximation on the clients  $\clientsexpensive$ and a ${4}$-approximation of the remaining clients (plus a small error term). Our  goal in the next sections will thus be to approximate this swap. The next few steps will be aimed at simplifying (or guessing parts of) $S_0$. The last step will then approximate $\ho$ and the remaining part of $S_0$ by submodular function maximization subject to a partition matroid constraint. For that part, it will be helpful to have a bound on the dummy solution as well; specifically, a modified version of it (see \Cref{lemma:costboundofMOandMDwithassignments} and \Cref{claim:Mprimebounds}). For intuition of those statements, let us here say that one can show 
\begin{align*}
        \clcost(M_D) & \leq \sum_{p\in \clientsexpensive} {9}\,\opt_p  + \sum_{p\in \clients - \clientsexpensive} r(p)  + O({\sqrt{\eps}}) \cdot \sopt\,\,
\end{align*}
by observing that only the clients $\clientsexpensive$ have a different connection cost in $M_D$ than in $M_O$. 
Moreover, each such client $p\in \clientsexpensive$ that was previously assigned to a center $c^* \in \ho$, and belongs to cluster $C^*$ of $\opt$ {with leader $\ell$}, has an associated dummy center $\delta$ at {squared} distance 
$$
{d^2(p,\delta)\leq 3d^2(p, c^*)+3d^2(c^*,\ell)+3d^2(\ell,\delta)\leq} {3}d^2(p, c^*) +  {6}\avg_{C^*, \opt} + {3}\eps \gamma_{C^*}
$$ 
(assuming the set $\calB$ of balls is valid {and using \Cref{lem:apxTriangleInequality3} with $\gamma=3$}). Simplifications then give the stated inequality as $|C^*| \avg_{C^*, \opt} \leq (1+2\eps) \sum_{p\in C^*} \opt_p$ and $|C^*| \eps \gamma_{C^*} = \eps \sum_{p\in C^*} ( \opt_p + S_p)$.

\subsection{Removal of Expensive Centers of $S_0$}
\label{sec:removalofExpensive}

Thus, it remains to show that our algorithm can yield a good approximation to
the gain that the swap $(\ho, S_0)$ would provide. Of course, it will not 
necessarily be optimum, but we can show it
is enough to conclude 
the proof of our theorem.
We focus on the iteration of the procedure
that produces valid ball guesses. In particular, at this point given that $W$ is successful and the ball guesses are valid, the number of balls in our ball guess is equal to $|S_0|=|\ho|{=|\Lambda|}\leq \frac{\log n}{\eps^3}$. Our algorithm then
makes use of the following procedure that takes as input the local search solution $S$ and the set $\dummyset$ of dummy centers.
\begin{mdframed}[hidealllines=true, backgroundcolor=gray!15]
\vspace{-5mm}
\paragraph{$expRem(S,\dummyset)$}\ \\
\begin{algorithmic}[1]
\State $\calLexp\leftarrow \emptyset$
\For{$(\frac{{20}}{\eps})^{|{\Lambda}|+1}\ln n$ many iterations}
\State $\calQ\leftarrow \emptyset$
\For{$|{\Lambda}|+1$ many iterations}
\State With probability 1/2 do the following:
\State Consider the clustering induced by $\dummyset\cup S-\calQ$ and, for each $c\in S-\calQ$, let $\clcost(c)$ be the cost of the cluster associated with $c$ in the considered solution
\State Sample one $c\in S-\calQ$ with probability 
$\frac{\clcost(c)}{\sum_{c'\in S-\calQ}\clcost(c')}$ 
\State Set $\calQ\leftarrow \calQ\cup \{c\}$
\EndFor
\State $\calLexp\leftarrow \calLexp\cup \{\calQ\}$ 
\EndFor
\State \textbf{return} $\calLexp$
\end{algorithmic}
\end{mdframed}
The goal of the above procedure is to guess a subset $\calQ$ of all ``costly'' centers in $S_0$ (the set $\calLexp$ contains all guesses). Specifically, if we consider the solution $(S - \calQ) \cup \dummyset$, then we wish that the cost of the clusters corresponding to the remaining centers in $S_0 - \calQ$ is at most $\eps \cdot {\sopt}$. Formally, the clusters corresponding to $S_0 - \calQ$ 
 in $S - \calQ \cup \dummyset$ consists of all the clients whose closest center in $S - \calQ \cup \dummyset$ is from $S_0 - \calQ$. If we let $\clients' \subseteq \clients$ be those clients then the total cost in solution  $S- \calQ \cup \Lambda$ of the clusters {with centers in} $S_0 - \calQ$ is defined as 
 \[
 \sum_{p\in \clients'} d^2(p, S- \calQ \cup \Lambda) = \sum_{p\in \clients'} d^2(p, S_0 - \calQ)\,,
 \]
 where the equality holds because of the definition $\clients'$.

    \begin{lemma}
    \label{lem:successguessprocess}
{$expRem()$ runs} in $n^{1/\eps^{O(1)}}$-time and produces a collection $\calLexp$ of at most $n^{1/\eps^{O(1)}}$ subsets of $S$ such that with probability at least $1-1/n$ there exists $\calQ\in \calLexp$ satisfying
    \begin{enumerate}
    \item $\calQ \subseteq S_0$, and
    \item The total cost in the solution $S - Q \cup \dummyset$ of the clusters with centers {in} $S_0 - \calQ$ is at most $\eps \cdot \sopt$.
    \end{enumerate}
\end{lemma}
\begin{proof}
The claim on the running time and on the size of $\calLexp$ is trivially satisfied since $|\Lambda|=|S_0|\leq \frac{\log n}{\eps^3}$. 

It remains to show that with probability at least $1-\frac{1}{n}$ at least one set $\calQ\in \calLexp$ satisfies the desired properties. Let us consider some iteration of the external for loop, and let $\calQ$ be the corresponding set. We say that the $j$-th iteration of the corresponding inner for loop is successful if the following happens:
\begin{enumerate}
    \item If condition (2) of the lemma is satisfied considering the current value of $\calQ$, the event from line 5 does not happen (hence in particular $\calQ$ is not updated in this $j$-th iteration).
    \item Otherwise, the event from line 5 happens and furthermore the sampled $c$ belongs to $S_0$.
\end{enumerate}
Let $A_j$ denote the event that the considered $j$-th iteration is successful.
Let us condition on the event $A_{<j}$ that the previous iterations $A_1,\ldots,A_{j-1}$ are successful. In particular, one has $\calQ\subseteq S_0$ at the beginning of the $j$-th iteration. 
Let $B_j$ denote the event that condition (2) is satisfied at the beginning of the iteration. Then trivially $Pr[A_j|B_j,A_{<j}]\geq 1/2$. Suppose next that $B_j$ does not hold. In that case, with probability $1/2$, the procedure samples a center $c$. Since by assumption $\sum_{c\in S_0-\calQ}\clcost(c)>\eps \sopt$, the probability that the procedure samples some $c\in S_0-\calQ$ is at least
$$
\frac{\eps\sopt}{\sum_{c'\in S-\calQ}\clcost(c')}\geq \frac{\eps \sopt}{\clcost(\dummyset\cup S-\calQ)}\geq \frac{\eps \sopt}{\clcost(\dummyset\cup S-S_0)}\geq \frac{\eps}{{10}},
$$
where in the last inequality we used the fact that $\clcost(\dummyset\cup S-S_0)\leq {10}\cdot\sopt$ by \Cref{lem:structSminusS0}.

Altogether $Pr[A_j|\overline{B}_j,A_{<j}]\geq \eps/{20}$. Thus $Pr[A_j|A_{<j}]\geq \eps/{20}$. Chaining the obtained inequalities we obtain that all the events $A_1,\ldots,A_{|S_0|+1}$ are simultaneously true with probability at least $(\eps/{20})^{|S_0|+1}$. Notice that when the latter event happens, the corresponding $\calQ$ satisfies all the properties. Indeed, $\calQ\subseteq S_0$ and furthermore the event $B_j$ must happen for some $j\leq |S_0|+1$ (hence for the next iterations if any) since otherwise $\calQ$ would contain more than $|S_0|$ elements.

Thus, the probability that the overall procedure fails is at most 
$(1-(\frac{\eps}{{20}})^{|S_0|+1})^{({20}/\eps)^{|S_0|+1}\ln n}\leq \frac{1}{n}$.
\end{proof}

\subsection{Consistently Assigning Clients in Mixed Solutions and the sets  $\calU$ and $\calR$}
\label{sec:consistentlyassigning}
While our algorithm tries all possible $\calQ \in \calLexp$, we focus  on the execution path of a set $\calQ$ of centers satisfying the properties of \Cref{lem:successguessprocess}, i.e., 
\begin{enumerate}
    \item 
      $\calQ \subseteq S_0$, and
      \item  the cost in the solution $S - \calQ \cup \dummyset$ of the clusters with centers in $S_0 - \calQ$ is at most $\eps \cdot \sopt$.
\end{enumerate}
    We define the solution $S_{\calQ} := S- \calQ \cup \dummyset$ and for a center $c\in S_{\calQ}$ we let $S_{\calQ}(c)$ be the subset of clients closest to $c$ (i.e., assigned to $c$) in the solution $S_{\calQ}$. The second property above allows us to consider each cluster $S_{\calQ}(c)$, $c\in S_0 - \calQ$, as contracted by paying a small extra cost of $\eps \cdot \sopt$. In particular, this allows us to give a ``consistent''  assignment $\mu_O$ of clients in the mixed solutions. Here, consistent means that all clients in $S_Q(c)$, for $c\in S_0 - \calQ$, are assigned to the same center in the mixed solutions where all of $S_0$ is removed.  
    
    \paragraph{Definition of $\mu_O$ and $\mu_D$.}
    We modify how he clients are assigned in the solution $M_O$ to obtain the assignment $\mu_O$.
    For every $c \in S_0 - \calQ$, all the clients in $S_{\calQ}(c)$ are reassigned to the same 
    center of $M_O$ as follows. Let $c_1$ be the 
    cluster of $\ho$ that is the closest to 
    $c$ and let $c_2$ be the center of $M_O-\ho = S - S_0$ that
    is the closest to $c$. The clients of $S_{\calQ}(c)$ are all assigned to either $c_1$ or $c_2$:
    \begin{itemize}
        \item If $d^2(c,c_1) \le d^2(c,c_2)/5$, assign
        $S_{\calQ}(c)$ to $c_1$; $S_{\calQ}(c)$ is 
        called a \emph{type-1} cluster.
        \item Otherwise ($d^2(c,c_1) > d^2(c,c_2)/5$), assign
        $S_{\calQ}(c)$ to $c_2$; $S_{\calQ}(c)$ is 
        called a \emph{type-2} cluster.
    \end{itemize}
    Let  $T_1$ be the set of clients in a type-1 cluster and $T_2$ the set of clients in a type-2
    cluster. The remaining clients in $\clientsexpensive - (T_1 \cup T_2)$ are assigned to their closest centers in $\ho$, and the remaining clients in $(\clients - \clientsexpensive)- (T_1 \cup T_2)$ are assigned to their closest centers in $S- S_0$. {Notice that, even neglecting the extra cost due to the consistent assignment, the above assignment of clients in $T_2$ is suboptimal when $d(c,c_2)>d(c,c_1)>d(c,c_2)/5$. The motivation for this reassignment is technical. More specifically, it will be used in Lemma \ref{lemma:costboundofMOandMDwithassignments} to have a better upper bound on the cost of clients in $C^*\cap T_1$ for a non-expensive cluster $C^*$ (case b2). This leads to a larger upper bound on the cost of clients in $C^*\cap T_2$ for an expensive cluster $C^*$ (case a3), which is however tolerable.}

    We further obtain an assignment $\mu_D$ of clients in $M_D$ by modifying $\mu_O$ as follows:  each client $p$ with $\mu_O(p) \in \ho$ is assigned to the associated dummy center by $\mu_D$. So for each client $p$ we have $\mu_D(p) = \mu_O(p) $ unless $\mu_O(p)\in \ho$ in which case $\mu_D(p)$ equals the dummy center associated with $\mu_O(p)$.

    \paragraph{Sets $\calR$ and $\calU$.} By the definition of type-1 and type-2 cluster, we can split the
    set of centers of $S_0 - \calQ$ into two groups $\calR$ and $\calU$. Let $\calR$
    be the set of centers of $S_0- \calQ$ whose set of clients is completely assigned
    to a center of $\ho$ in the assignment $\mu_O$, and let $\calU = S_0 - \calQ- \calR$ be the remaining ones that are assigned to centers in $S- S_0$.

    We end this section by analyzing the costs of the assignments $\mu_O$ and $\mu_D$. {Recall that} we use the notation $\clcost(M_O,\mu_O)$ to denote the cost of $M_O$ equipped with assignment $\mu_O$ and, similarly $\clcost(M_D, \mu_D)$ for the cost of $M_D$ equipped with assignment $\mu_D$. We say that the selection of $(W, \calB, \calQ)$ is successful, if the sample $W$ selected in \Cref{sec:dsampleproc} is successful, the set of balls $\calB$ from \Cref{sec:ballguesses} is valid, and $\calQ$ selected in this section satisfies the properties of \Cref{lem:successguessprocess}.
\begin{lemma}
    If the selection of $(W, \calB, \calQ)$ is successful,
    \[
        \clcost(M_O, \mu_O) + \clcost(M_D, \mu_D) \leq 10\cdot \sopt + O(\sqrt{\eps} \cdot \sopt)\,. 
    \]
    \label{lemma:costboundofMOandMDwithassignments}
\end{lemma}
\begin{proof}
Throughout the proof of the lemma, we will repeatedly upper bound $d^2(p,\mu_D(p))$ in terms of $d^2(p,\mu_O(p))$ using the following bounds:
\begin{claim}\label{clm:costboundofMOandMDwithassignments}
Consider a client $p\in S_{\calQ}(c)$ with $\mu_O(p) \neq \mu_D(p)$ and so $\mu_O(p)\in \ho$ is the center of an expensive cluster $C^*$ of $\opt$. Then
\begin{enumerate}
\item $d^2(p, \mu_D(p))  \leq 3d^2(p, \mu_O(p)) + 6  ( \avg_{C^*, \opt} + \eps \gamma_{C^*})$, and
\item $d^2(p,\mu_D(p)) \leq \frac{2}{\sqrt{\eps}}d^2(p, S_\calQ) + 9(1+\sqrt{\eps})d^2(c, \mu_O(p))$.
\end{enumerate}
\end{claim}
\begin{proof}[Proof of Claim \ref{clm:costboundofMOandMDwithassignments}]
Let $\mu_O(p)  = c^*$. By the assumption that $\calB$ is a valid set of balls, we have by \Cref{lem:ballguesses} that there is a ball $B(\ell, \sqrt{\rho}) \in \calB$ such that $c^* \in B(\ell, \sqrt{\rho})$ and 
\[
    \avg_{C^*, \opt} \leq \rho \leq \avg_{C^*, \opt} + \eps \gamma_{C^*}\,.
\]
Thus, if we let $\delta \in \Lambda$ be the dummy center associated with this ball (at a distance $\sqrt{\rho}$ from $\ell$), then $\mu_D(p) = \delta$ and, by \Cref{lem:apxTriangleInequality3} with $\gamma=3$, 
\[
    d^2(p, \mu_D(p)) \leq 3 d^2(p, c^*) + 3 d^2(c^*, \ell) + 3 d^2(\ell, \delta) \leq 3 d^2(p, \mu_O(p)) + 6  ( \avg_{C^*, \opt} + \eps \gamma_{C^*})\,.
\]
For the second bound, one has:
\begin{align*}
 \notag   d^2(p, \mu_D(p)) & \leq \frac{2}{\sqrt{\eps}}d^2(p, c)  + (1+\sqrt{\eps})(d(c, c^*) + d(c^*, \ell) + d(\ell, \delta))^2 \\
\notag    &\leq \frac{2}{\sqrt{\eps}}d^2(p, S_\calQ) + 3(1+\sqrt{\eps})(d^2(c, c^*) + 2( \avg_{C^*, \opt} + \eps \gamma_{C^*})) \\
    &\leq \frac{2}{\sqrt{\eps}}d^2(p, S_\calQ) + 9(1+\sqrt{\eps}) d^2(c, c^*)\,.
    \label{eq:lastinequality}
\end{align*}
In the first inequality above we used \Cref{lem:apxTriangleInequality3} with $\gamma=1+\sqrt{\eps}$ (hence $\frac{\gamma}{\gamma-1}\leq \frac{2}{\sqrt{\eps}}$), in the second inequality above we used \Cref{lem:apxTriangleInequality3} with $\gamma=3$, and the last inequality above follows from the fact that we must have $d^2(c, c^*) > (1+\eps) \avg_{C^*, \opt} + \eps \gamma_{C^*}$ since otherwise we would have the contradiction that $c^*\in \ho$ is covered by $S$; recall that $\ho$ only contains expensive centers of $\opt$. 
\end{proof}
Having proven the above claim, we prove the lemma by distinguishing between expensive and non-expensive clusters $C^*$ of $\opt$. 

\paragraph{Case a: $C^*$ is an expensive cluster of $\opt$.} Let $c^*\in \ho$ be the center of $C^*$. Let $p\in C^*$ be a client in $C^*$ and let $c$ be its closest center in $S_\calQ$, i.e.,   $p\in S_\calQ(c)$.  We bound $d^2(p, \mu_O(p)) + d^2(p, \mu_D(p))$  by considering three cases:
\begin{description}
    \item[Case a1: $p\in C^* - (T_1 \cup T_2)$.]  By definition, we have  $\mu_O(p)$ is the closest center in $\ho$. Moreover, $\mu_O(p)$ equals the cluster center $c^*$ of $C^*$ since $p$ is closest to $c^*$ among all centers in $\opt \supseteq \ho$. Claim \ref{clm:costboundofMOandMDwithassignments}.1 together with \Cref{lemma:costsMO} thus gives us,
    \[
        d^2(p,\mu_O(p)) + d^2(p, \mu_D(p)) \leq 4\opt_p +  6  ( \avg_{C^*, \opt} + \eps \gamma_{C^*})\,.
    \]
\item[Case a2: $p\in C^* \cap T_1$.]  By definition of type-1 clusters,  $\mu_{O}(p) = c_1$ is the closest center in $\ho$ to $c$. Now, by Claim \ref{clm:costboundofMOandMDwithassignments}.2,
\begin{align*}
d^2(p, \mu_O(p))  + d^2(p, \mu_D(p)) & \leq d^2(p,\mu_O(p)) + \frac{2}{\sqrt{\eps}} d^2(p, S_{\calQ}) + 9(1+\sqrt{\eps}) d^2(c, c_1) \\
& \leq \frac{2}{\sqrt{\eps}}d^2(p, S_\calQ)+(1+\sqrt{\eps})d^2(c,c_1) + 9(1+\sqrt{\eps}) d^2(c, c_1)\\
&  \leq \frac{2}{\sqrt{\eps}}d^2(p, S_\calQ)+10(1+\sqrt{\eps})(\frac{2}{\sqrt{\eps}}d^2(p, S_\calQ)+(1+\sqrt{\eps})d^2(p,c_1))\\
& \leq \frac{23}{\sqrt{\eps}}d^2(p, S_\calQ)+(10+O(\sqrt{\eps}))\opt_p\,.
\end{align*}
where the second and third inequalities above uses \Cref{lem:apxTriangleInequality2} with $\gamma=1+\sqrt{\eps}$.
\item[Case a3: $p\in C^*\cap T_2$.] By the definition of type-2 clusters, $\mu_O(p) = c_2$ is the closest center to $c$ in $S- S_0$, and $d^2(c, c_2) \leq 5 d^2(c, \ho)$. So, by \Cref{lem:apxTriangleInequality2} with $\gamma=1+\sqrt{\eps}$, $d^2(c, c_2) \leq 5(\frac{2}{\sqrt{\eps}} d^2(c,p) + (1+\sqrt{\eps})\opt_p)$ where we additionally used that $d^2(p, \ho) \leq \opt_p$ by \Cref{lemma:costsMO}. As $\mu_O(p) = \mu_D(p)$ in this case, and using again \Cref{lem:apxTriangleInequality2} with $\gamma=1+\sqrt{\eps}$,
\begin{align*}
{d^2}(p, \mu_O(p))  + d^2(p, \mu_D(p)) & = 2d^2(p,\mu_O(p))  \\
& \leq 2(\frac{2}{\sqrt{\eps}}d^2(p, c) + (1+\sqrt{\eps})d^2(c,c_2)) \\
& \leq 2(\frac{2}{\sqrt{\eps}}d^2(p, c) + (1+\sqrt{\eps})5(\frac{2}{\sqrt{\eps}}d^2(c,p) + (1+\sqrt{\eps})\opt_p)) \\
& \leq \frac{25}{\sqrt{\eps}}d^2(p, S_\calQ) + (10+O(\sqrt{\eps})) \opt_p\,.
\end{align*}
\end{description}
By the above bounds, we have that \begin{align*}
& \sum_{p\in C^*} (d^2(p, \mu_O(p)) + d^2(p, \mu_D(p)))\\
\leq & \sum_{p\in C^* - (T_1 \cup T_2)} \left(4\opt_p +  6(\avg_{C^*, \opt} + \eps \gamma_{C^*})\right) + \sum_{p\in C^*\cap (T_1 \cup T_2)}\left((10+O(\sqrt{\eps}))\opt_p + \frac{25}{\sqrt{\eps}} d^2(p, S_\calQ)\right)\,.  
\end{align*}
We complete the analysis of the expensive clusters by upper bounding the terms $6\eps\gamma_{C^*}$ and $6\avg_{C^*, \opt}$. By definition we have $|C^* - (T_1 \cup T_2)| \cdot 6\eps \gamma_{C^*}\leq  6\eps \sum_{p\in C^*} (\opt_p + S_p)$. For the other term, we claim that 
\[
    6 |C^* - (T_1 \cup T_2)|\cdot \avg_{C^*, \opt} \leq \sum_{p\in C^* - (T_1 \cup T_2)} 6\opt_p + 7\eps \cdot \sum_{p\in C^*} \opt_p\,. 
\]
This holds because there are at least $(1-\eps)|C^*|$ clients  in $C^*$ such that $\opt_p \geq \avg_{C^*, \opt}$, which implies 
\begin{align*}
    \sum_{p\in C^* - (T_1 \cup T_2)} 6\opt_p + 7\eps \cdot \sum_{p\in C^*} \opt_p &\geq 6(|C^* - (T_1 \cup T_2)|-\eps|C^*|)\cdot \avg_{C^*, \opt} + 7\eps(1-\eps)|C^*|\cdot \avg_{C^*, \opt} \\
    & \geq 6 |C^*- (T_1 \cup T_2)|\cdot \avg_{C^*, \opt}\,. 
\end{align*}
In summary, for an expensive cluster $C^*$ we have that
\begin{align}
& \sum_{p\in C^*} (d^2(p, \mu_O(p)) + d^2(p, \mu_D(p)))\nonumber\\
\leq &
\sum_{p\in C^* - (T_1 \cup T_2)} 10\opt_p  + \sum_{p\in C^*\cap (T_1 \cup T_2)}(10\opt_p + \frac{25}{\sqrt{\eps}} d^2(p, S_\calQ)) + O(\sqrt{\eps})\sum_{p\in C^*} (\opt_p  + S_p)
\label{eq:expensive_cluster_bound}
\end{align}
.

\paragraph{Case b: $C^*$ is a non-expensive cluster of $\opt$.} We proceed with a similar analysis. Indeed, consider a client $p\in C^*$  and let $c$ be its closest center in $S_\calQ$, i.e.,   $p\in S_\calQ(c)$.  We again bound $d^2(p, \mu_O(p)) + d^2(p, \mu_D(p))$  by considering three cases:

\begin{description}
    \item[Case b1: $p\in C^* - (T_1 \cup T_2)$.] By the definition of $\mu_O$, we have that $\mu_O(p)$ is the closest center in $S - S_0$ to $p$. So $\mu_D(p) = \mu_O(p)$ and  $d^2(p,\mu_O(p)) + d^2(p, \mu_D(p)) \leq 2r(p)$, by \Cref{lemma:costsMO}. 
\item[Case b2: $p\in C^* \cap T_1$.]  By the definition of type-1 clusters,  $\mu_{O}(p) = c_1$ is the closest center in $\ho$ to $c$. Now, by Claim \ref{clm:costboundofMOandMDwithassignments}.2
 and \Cref{lem:apxTriangleInequality2} with $\gamma=1+\sqrt{\eps}$, 
\begin{align*}
d^2(p, \mu_O(p))  + d^2(p, \mu_D(p)) & \leq d^2(p,\mu_O(p)) + \frac{2}{\sqrt{\eps}}d^2(p, S_{\calQ}) + 9(1+\sqrt{\eps}) d^2(c, c_1) \\
& \leq \frac{2}{\sqrt{\eps}}d^2(p, S_{\calQ})+(1+\sqrt{\eps})d^2(c, c_1)+\frac{2}{\sqrt{\eps}}d^2(p, S_{\calQ}) + 9(1+\sqrt{\eps}) d^2(c, c_1)\\
& \leq \frac{4}{\sqrt{\eps}} d^2(p, S_\calQ) + 10(1+\sqrt{\eps}) d^2(c,c_1)\,.
\end{align*}
We proceed by upper bounding $d^2(c, c_1)$. 
{Here we critically use the fact that, by the definition of $T_1$, one has} $d^2(c, c_1) \leq d^2(c, c_2)/5$ where $c_2$ is the closest center to $c$ in $S- S_0$. Let $c_3$ be the closest center to $p$ in $S - S_0$. Then, using \Cref{lem:apxTriangleInequality2} with $\gamma=1+\sqrt{\eps}$, we obtain that 
$$
d^2(c, c_1) \leq \frac{1}{5}d^2(c, c_2)\leq \frac{1}{5}d^2(c,c_3) \leq \frac{2}{5\sqrt{\eps}}d^2(p,S_\calQ) + \frac{1+\sqrt{\eps}}{5}d^2(p, c_3).
$$ 
Moreover, we have $d^2(p, c_3) = d^2(p, S- S_0)\leq r(p)$ by \Cref{lemma:costsMO} and so
$
    d^2(c,c_1)  \le \frac{1}{\sqrt{\eps}}d^2(p, S_\calQ) + \frac{1+\sqrt{\eps}}{5}r(p)$, which gives us the bound
    \begin{align*}
    d^2(p, \mu_O(p))  + d^2(p, \mu_D(p)) \leq \frac{15}{\sqrt{\eps}} d^2(p, S_\calQ) + (2+O(\sqrt{\eps})) r(p)\,.
    \end{align*}
\item[Case b3: $p\in C^*\cap T_2$.] By the definition of type-2 clusters, $\mu_O(p) = c_2$ is the closest center to $c$ in $S- S_0$, and so $\mu_O(p) = \mu_D(p)$. Now by \Cref{lemma:costsMO} and applying twice \Cref{lem:apxTriangleInequality2} with $\gamma=1+\sqrt{\eps}$,
\begin{align*}
d^2(p, \mu_O(p)) & \leq \frac{2}{\sqrt{\eps}}d^2(p,c) + (1+\sqrt{\eps})d^2(c, c_2)\\
& \leq \frac{2}{\sqrt{\eps}}d^2(p,c ) + (1+\sqrt{\eps})\left(\frac{2}{\sqrt{\eps}}d^2(p,c )+(1+\sqrt{\eps})d^2(p,c_2)\right)\\
& \leq \frac{5}{\sqrt{\eps}}d^2(p,c ) + (1+\sqrt{\eps})^2 d^2(p, S- S_0)\\
& \leq \frac{5}{\sqrt{\eps}}d^2(p, S_\calQ) + (1+O(\sqrt{\eps}))r(p).
\end{align*} 
So we get the bound
\begin{align*}
    d^2(p,\mu_O(p)) + d^2(p, \mu_D(p)) {=2 d^2(p,\mu_O(p))}\leq  \frac{10}{\sqrt{\eps}}d^2(p, S_\calQ) + (2+O(\sqrt{\eps})) r(p)\,.
\end{align*}
\end{description}
Summing up the above bounds, yields
\[\sum_{p\in C^*} (d^2(p, \mu_O(p)) + d^2(p, \mu_D(p))) \leq 
\sum_{p\in C^*} (2+O(\sqrt{\eps})) r(p) + \sum_{p\in C^* \cap (T_1 \cup T_2)} \frac{15}{\sqrt{\eps}} d^2(p, S_\calQ)
\]
for a non-expensive cluster $C^*$ of $\opt$. 
If we sum  up the above inequality for all non-expensive clusters and the bound~\eqref{eq:expensive_cluster_bound} for expensive clusters of $\opt$ we thus get
\[
\clcost(M_O, \mu_O) + \clcost(M_D, \mu_D) \leq \sum_{p\in \clientsexpensive} 10 \opt_p + \sum_{p\in \clients- \clientsexpensive} (2+O(\sqrt{\eps})) r(p) + O(\sqrt{\eps} \cdot \sopt)\,,
\]
where we used that $\sum_{c \in S_0 - \calQ} \sum_{p \in S_{\calQ}(c)} d^2(p, S_\calQ) = \sum_{p\in T_1 \cup T_2} d^2(p, S_\calQ) \leq \eps \cdot \sopt$ (\Cref{lem:successguessprocess}) and $\clcost(S) \leq \apxLSkmeans \sopt$ (\Cref{lemma:localsearchis5approximation}) to bound the error term $O(\sqrt{\eps} \cdot \sopt)$.

The proof is now concluded by \Cref{lemma:convenience_upper_bound}, which says that $\sum_{p\in \clients - \clientsexpensive} r(p) \leq \sum_{p\in \clients - \clientsexpensive} 4\opt_p + O(\sqrt{\eps} \cdot \sopt)$ if the sample $W$ is successful. 
\end{proof}

\subsection{Removal of Cheap Centers in $S_0$}
\label{sec:removalofCheap}

At this stage we have defined a clustering $S_{\calQ} = S - \calQ \cup \dummyset$, where $\dummyset$ is the set of dummy centers. Recall that the mixed solution $M_O$ equals $S-S_0 \cup \ho$. So $S_\calQ - \Lambda$ represents progress compared to $S$ in that we have already removed a subset $\calQ$ of $S_0$. Furthermore, with respect to the assignment of clients $\mu_O$, we have that the remaining centers $S_0 - Q$, which we want to remove, are divided into two sets $\calR$ and $\calU$. The clusters with centers in $\calR$ are completely assigned (by $\mu_O$) to centers in $\ho$, and every cluster with center in $\calU$ is completely assigned to a center in $S - S_0$. 
The task of this section is to make further progress by guessing the centers in $\calU$ and their reassignment to centers in $S- S_0$. This turns out to be a difficult task, and we will instead ``approximately'' guess a set $\bcalU$ and reassignment $\tmu$ of the clients of those clusters. Specifically, we will output a solution $S_{\calQ \cup \bcalU} = S_{\calQ} - \bcalU$ (obtained by removing $\bcalU$ from $S_{\calQ}$) and an assignment of clients $\tmu$ to centers in $S_\calQ-\bcalU$. This is done through the following lemma.   {Recall that for $c\in S_{\calQ}$, $S_{\calQ}(c)$ is the set of clients assigned to $c$ in $S_\calQ$, i.e., closest to $c$ among all centers in $S_\calQ$.}    
\begin{lemma}\label{lem:successcheapremove}
Given $S$ and $\calQ$ as described above, there is a polynomial-time algorithm that produces a collection $\calL_{cheap}$ of subsets $\bcalU$ of $S-\calQ$, and for each such $\bcalU$ and assignment $\tmu$ of clients to centers in $S_{\calQ \cup \bcalU} = S_\calQ-\bcalU$, such that at least one such pair $(\bcalU,\tmu)$ satisfies the following properties:
\begin{enumerate}
    \item $|\bcalU| = |\calU|$.
    \item $\bcalU \cap \calR = \emptyset$, i.e., $\bcalU$ does not contain any center of $S_0 - \calQ$ whose set of clients is completely assigned
    to a center of $\ho$ in the assignment $\mu_O$.
    \item $\tmu$ satisfies: \begin{enumerate}
        \item For every center $c\in S_{\calQ} - \bcalU$, we have $\tmu(p)= c$ for every $p\in S_{\calQ}(c)$.
        \item For a center $c\in \bcalU$, the clients in $S_\calQ(c)$ are reassigned to a center $c' \in S- \calQ - \bcalU - \calR$, i.e., $\tmu(p) = c'$ for every $p\in S_{\calQ}(c)$.
       \item The cost increase of the reassignment $\tmu$ of clients previously assigned to $\bcalU$ compared to that  of the reassignment $\mu_O$ of clients previously assigned to $\calU$ is bounded by $O(\eps \cdot \sopt)$:
       \begin{gather*}
        \sum_{c \in \calU \cup \bcalU}  \sum_{p\in S_\calQ(c)} \left( d^2(p, \tmu(p)) - d^2(p, \mu_O(p)) \right) = O ({\sqrt{\eps}} \cdot \sopt)\,. 
        \end{gather*}
    \end{enumerate}
\end{enumerate}
\end{lemma}

In words, the above properties of $\tmu$ say that we maintain the assignment of clients whose closest center remains the same, and the clients associated to the removed centers $\bcalU$ are reassigned to centers not in $\calR$ at no higher cost (up to $O({\sqrt{\eps}} \cdot \sopt)$) than the cost of the reassignment of $\calU$ in the modified solution $M_O$.

Let $\calX$ be the centers in $S-S_0$ to which the points $p\in S_{\calQ}(c)$ with $c\in \calU$ are assigned according to $\mu_O$. Notice that $|\calX|\leq |\calU|$ and $\calX \cap (\calQ\cup \calU\cup \calR)=\emptyset$.
We assume that the values of $|\calU|$, $|\calR|$, and $|\calX|$ are known. This can be achieved by trying, for each such value, all the integers between $0$ and $|S_0|\leq \frac{\log n}{\eps^3}$. Let $\ell:=|\calU|+|\calR|+|\calX|\leq 2|\calU|+|\calR|\leq 2|S_0 - \calQ|$. Then we run the recursive procedure $cheapRem()$ described in the box with input $(\calU',\calR',\calX',\calN,\bcalU)=(\emptyset,\emptyset,\emptyset,\emptyset,\emptyset)$. {We assume that $cheapRem()$ has access to the quantities $S$, $\calQ$, $|\calU|$, $|\calR|$, $|\calX|$ and $\ell$, as well as to a} global variable $\calL_{cheap}$, which is initialized to $\emptyset$. The intuition for the parameters is as follows. Intuitively $\calU'$, $\calR'$ and $\calX'$ are subsets of $\calU$, $\calR$ and $\calX$, resp., that we have already identified, while $\bcalU$ is the current value of the set $\bcalU$ under construction. Intuitively, $\calN$ are the centers to which the points $p\in S_{\calQ}(c)$ with $c\in \bcalU$ are reassigned, namely $\tmu(p)\in \calN$. At the end of the {root call}, $\calL_{cheap}$ contains the desired collection of sets $\bcalU$ as in \Cref{lem:successcheapremove}. Whenever some $\bcalU$ is added to $\calL_{cheap}$, we define a corresponding $\tmu$ as follows: For each $p\in S_{\calQ}(c)$, $\tmu(p)=c$ if $c\in S_{\calQ}-\bcalU$. Otherwise, i.e., if $c\in \bcalU$, $\tmu(p)=next(c)\in \calN$, where $next(c)$ is defined in the recursive call when $c$ is added to $\bcalU$.
\begin{mdframed}[hidealllines=true, backgroundcolor=gray!15]
\vspace{-5mm}
\paragraph{$cheapRem(\calU',\calR',\calX',\calN,\bcalU)$}\ \\
\begin{algorithmic}[1]
\If{$|S-\calQ|\leq 4\ell$}\label{alg:cheapRem:fewCenters}
\For{All $\bcalU\subseteq S-\calQ$ of size $|\calU|$ and $\bcalR\subseteq S-\calQ-\bcalU$ of size $|\calR|$}
\State Add $\bcalU$ to $\calL_{cheap}$ and, for each $c\in \bcalU$, set $next(c)$ to the closest center to $c$ in $S-\calQ-\bcalU-\bcalR$\label{alg:cheapRem:optimalGuess}
\EndFor
\State \textbf{halt}
\EndIf
\If{$|\bcalU|=|\calU|$}\label{alg:cheapRem:termination}
\State Add $\bcalU$ to $\calL_{cheap}$ and \textbf{halt}
\EndIf
\State Let $c\in S-\calQ-\calR'-\calX'-\calN-\bcalU$ minimize $R(c):=\sum_{p\in S_{\calQ}(c)}d^2(p,next(c))$, where $next(c)$ is the closest center to $c$ in $S-\calQ-\calU'-\calR'-\bcalU-\{c\}$.\label{alg:cheapRem:selectc}
\If{$c\notin \calU'$ and $|\calX'|<|\calX|$} $cheapRem(\calU',\calR',\calX'\cup \{c\},\calN,\bcalU)$\label{alg:cheapRem:discoveredX}
\EndIf
\If{$c\notin \calU'$ and $|\calR'|<|\calR|$} $cheapRem(\calU',\calR'\cup \{c\},\calX',\calN,\bcalU)$\label{alg:cheapRem:discoveredR}
\EndIf
\If{$|\calU'|+|\calR'|=|\calU|+|\calR|$ or $next(c)\in \calN$} 
$cheapRem(\calU',\calR',\calX',\calN\cup \{next(c)\},\bcalU\cup \{c\})$\label{alg:cheapRem:easyAddtU}
\Else
\If{$|\calR'|<|\calR|$} $cheapRem(\calU',\calR'\cup \{next(c)\},\calX',\calN,\bcalU)$ \label{alg:cheapRem:guessNextR}
\EndIf
\If{$|\calU'|<|\calU|$} $cheapRem(\calU'\cup \{next(c)\},\calR',\calX',\calN,\bcalU)$ \label{alg:cheapRem:guessNextU}
\EndIf
\State $cheapRem(\calU',\calR',\calX',\calN\cup \{next(c)\},\bcalU\cup \{c\})$\label{alg:cheapRem:guessNextNotRU}
\EndIf
\end{algorithmic}
\end{mdframed}

    \begin{proof}[Proof of \Cref{lem:successcheapremove}]    
Consider the execution of the above procedure for the correct guess of the values $|\calR|$, $|\calX|$ and $|\calU|$. 
We remark that, when line \ref{alg:cheapRem:selectc} is executed, the set $S-\calQ-\calR'-\calX'-\calN-\bcalU$ is not empty (so that we can choose {an} appropriate $c$). Indeed, at each recursive call we add at most two elements to $\calR'\cup \calX'\cup \calN\cup \bcalU$  
and the value of $|\calU'|+|\calR'|+|\calX'|+|\bcalU|$ grows by at least $1$. The latter value cannot grow more than $2|\calU|+|\calR|+|\calX|\leq 2\ell$ times since at that point we would have necessarily $|\bcalU|=|\calU|$, which makes the condition of line \ref{alg:cheapRem:termination} true. Since the condition of line \ref{alg:cheapRem:fewCenters} cannot hold if line \ref{alg:cheapRem:selectc} is executed at least once, we have that $|S- \calQ|> 4\ell$. Thus $S-\calQ$ contains sufficiently many elements to remove up to $2$ elements for up to $2\ell$ times. 

Concerning the running time of the procedure, each recursive step involves a branching on at most $4$ subproblems, and in each one of them the value of $|\calU'|+|\calR'|+|\calX'|+|\bcalU|$ grows by at least $1$. Thus, by the same argument as before, the depth of the recursion is at most $2\ell$. So the number of recursive calls is at most $4^{2\ell}=n^{{1/\eps^{O(1)}}}$, implying a polynomial running time.

It remains to show that at least one set $\bcalU\in \calL_{cheap}$ (with the associated $\tilde{\mu}$) at the end of the procedure satisfies the claim. If the condition of line \ref{alg:cheapRem:fewCenters} holds, it must happen that in one of the executions of line \ref{alg:cheapRem:optimalGuess} one has $\bcalU=\calU$ and $\bcalR=\calR$. In that case $\bcalU$ and the corresponding $\tmu$ trivially satisfy the claim. In particular we remark that for each $p\in S_{\calQ}(c)$ with $c\in \bcalU=\calU$, $\tmu(p)=next(c)=\mu_{O}(p)$. 

Otherwise, let us focus on the solution $\bcalU\in \calL_{cheap}$ which is generated by the following chain of recursive calls starting from the root call with $(\calU',\calR',\calX',\calN,\bcalU)=(\emptyset,\emptyset,\emptyset,\emptyset,\emptyset)$. Consider the current input parameters $(\calU',\calR',\calX',\calN,\bcalU)$ and let $c$ be selected in line \ref{alg:cheapRem:selectc}. {During the process, we will maintain the following invariants:
\begin{equation}\label{inv:UpRpXp}
\calU'\subseteq \calU;\,\calR'\subseteq \calR; \calX'\subseteq \calX.    
\end{equation}
\begin{equation}\label{inv:N}
\calN \cap (\calU \cup \calR)=\emptyset.    
\end{equation}
\begin{equation}\label{inv:tU}
\bcalU \cap (\calX \cup \calR\cup \calN)=\emptyset. 
\end{equation}
The root call satisfies the mentioned invariants trivially since $\bcalU=\calU'=\calR'=\calX'=\calN=\emptyset$.
}

If $c\in \calX$, we continue with the recursive call of line \ref{alg:cheapRem:discoveredX}. {The invariants are maintained since $\calX'\cup\{c\}\subseteq \calX$ and $c\notin \bcalU$.} Observe that this choice is excluded if $c\in \calU'$ or $|\calX'|=|\calX|$, {however by Invariant \eqref{inv:UpRpXp} the latter conditions cannot happen when $c\in \calX$}.

Similarly, if $c\in \calR$, we continue with the recursive call of line \ref{alg:cheapRem:discoveredR}. {The invariants are maintained since $\calR'\cup\{c\}\subseteq \calR$, $c\notin \calN$, and $c\notin \bcalU$.} Observe that this choice is excluded if $c\in \calU'$ or $|\calR'|=|\calR|$, {however by Invariant \eqref{inv:UpRpXp} the latter conditions cannot happen when $c\in \calR$}.

Otherwise, namely if $c\notin \calR\cup \calX$, if the condition of line \ref{alg:cheapRem:easyAddtU} is true, we continue with the recursive call of the same line. {Let us show that invariants are preserved. Suppose first that $next(c)\in \calN$. In this case the first two invariants are not affected, and the third one is maintained since $c\notin \calN$ (and by assumption $c\notin \calR\cup \calX$). Otherwise one has $|\calU'|+|\calR'|=|\calU|+|\calR|$, which by Invariant \eqref{inv:UpRpXp} implies $\calU'=\calU$ and $\calR'=\calR$. Thus one has $next(c)\notin \calU\cup \calR$ and $c\notin \calR\cup \calX\cup \calN$.}

Otherwise, i.e., if the condition of line \ref{alg:cheapRem:easyAddtU} is false, depending on whether $next(c)$ belongs to $\calR$, $\calU$, or none of the previous cases, we continue with the recursive calls of lines \ref{alg:cheapRem:guessNextR}, \ref{alg:cheapRem:guessNextU}, and \ref{alg:cheapRem:guessNextNotRU}, resp.
{If $next(c)\in \calR$, the invariants are satisfied since $\calR'\cup \{c\}\subseteq \calR$. Invariant \eqref{inv:UpRpXp} guarantees that in this case $|\calR'|<|\calR|$ as required. Similarly, if $next(c)\in \calU$, the invariants are satisfied since $\calU'\cup \{c\}\subseteq \calU$. Invariant \eqref{inv:UpRpXp} guarantees that in this case $|\calU'|<|\calU|$ as required. The remaining case if that $next(c)\notin \calR\cup \calU$. Invariant \eqref{inv:N} is maintained since $next(c)\notin \calR\cup \calU$. Invariant \eqref{inv:tU} is maintained since $c\notin \calX\cup \calR$ by the assumptions of this case, $c\notin \calN\cup \{next(c)\}$ by construction, and $\next{c}\notin \bcalU$ by construction.}

Let us show that the final $\tilde{U}$ obtained with the above procedure with the associated $\tilde{\mu}$ satisfies the claim. Condition (1) is trivially satisfied. {Invariant \eqref{inv:tU} directly implies property (2).} We next focus on property (3). Property (a) is satisfied by definition. {When a center $c$ is added to $\bcalU$, the respective $c':=next(c)=\tmu(p)$ for all $p\in S_{\calQ}(c)$ is added to $\calN$ (if not already there). Invariant \eqref{inv:N} guarantees that $c'\notin \calR$, while Invariant \eqref{inv:tU} guarantess that $c'\notin \bcalU$. Property (b) follows.}

It remains to prove property (c). Consider first any $c\in \calU\cap \bcalU$. We claim that for the corresponding $p\in S_{\calQ}(c)$ one has $\tilde{\mu}(p)=\mu_O(p)$. Indeed, let $c_1,c_2,\dots$ be the centers in $S- \calQ- \{c\}$ in non-decreasing order of distance from $c$. Let also $c_q$ be the first such center belonging to $S- S_0$. In particular $\mu_O(p)=c_q$. By construction all the centers $c_1,\ldots,c_{q-1}$ must have been already added to $\calU'\cup \calR'$ in previous recursive steps when $c$ is added to $\bcalU$. Furthermore, when the latter event happens either $c_q\in \calN$ already or $c_q$ is added to $\calN$. In both cases one has that $\tilde{\mu}(p)=c_q$ as desired. As a consequence
\begin{align}
\sum_{c\in \calU\cap \bcalU}\sum_{p\in S_{
\calQ}(c)}(d^2(p,\tilde{\mu}(p))-d^2(p,\mu_O(p)))=0.\label{lem:successcheapremove:eqn1}
\end{align}
Consider next any $c\in \calU- \bcalU$. For any $p\in S_{\calQ}(c)$ by construction one has $\tilde{\mu}(p)=c$. 
Thus
\begin{align}    
\sum_{c\in \calU- \bcalU}\sum_{p\in S_{
\calQ}(c)}(d^2(p,\tilde{\mu}(p))-d^2(p,\mu_O(p))) & \leq \sum_{c\in S_0- \calQ}\sum_{p\in S_{
\calQ}(c)}d^2(p,c) -\sum_{c\in \calU- \bcalU}\sum_{p\in S_{
\calQ}(c)}d^2(p,\mu_O(p))
\nonumber \\
& \leq O(\eps) opt-\sum_{c\in \calU- \bcalU}\sum_{p\in S_{
\calQ}(c)}d^2(p,\mu_O(p)),\label{lem:successcheapremove:eqn2}
\end{align}
where in the last inequality above we used the assumption that $\calQ$ satisfies property (2) of \Cref{lem:successguessprocess}, and the fact that $\calU- \bcalU\subseteq \calU\subseteq S_0-\calQ$.

Finally consider any $c\in \bcalU- \calU$. Define any bijection between each such $c$ and some distinct $c'\in \calU- \bcalU$ (this is possible since $|\bcalU|=|\calU|$). Consider the recursive call when $c$ is added to $\bcalU$. Let $\calU'$, $\calR'$, $\calX'$, $\calN$, $next(c)$, $R(c)$, $next(c')$, $R(c')$ be the associated quantities in that call. Notice that at that time $c'$ was an available candidate to be added to $\bcalU$ since $c'\in S-\calQ-\calR-\calX-\bcalU\subseteq S-\calQ-\calR'-\calX'-\bcalU$ {(using Invariant \eqref{inv:UpRpXp})}
and $\calN\cap \calU=\emptyset$ {by Invariant \eqref{inv:N}}, {hence $c'\in S-\calQ-\calR'-\calX'-\calN-\bcalU$}. Since we added $c$ instead of $c'$, it must be the case that $R(c)\leq R(c')$. 
For every $p\in S_{\calQ}(c')$, $\mu_O(p)\in \calX$, hence $\mu_O(p)\notin \bcalU$ {by Invariant \eqref{inv:tU}.} It follows that $\mu_O(p)\in S-\calQ-\calU- \calR- \bcalU-\{c'\}\subseteq S-\calQ-\calU'- \calR'- \bcalU-\{c'\}$, where we used {again Invariant \eqref{inv:UpRpXp}}. Thus $d^2(c',next(c'))\leq d^2(c',\mu_O(p))$. 
Hence for every $p\in S_{\calQ}(c')$, using \Cref{lem:apxTriangleInequality2} with $\gamma=1+\sqrt{\eps}$ twice,
\begin{align*}
d^2(p,next(c')) & \leq \frac{2}{\sqrt{\eps}}d^2(p,c')+(1+\sqrt{\eps})d^2(c',next(c')) \leq \frac{2}{\sqrt{\eps}}d^2(p,c')+(1+\sqrt{\eps})d^2(c',\mu_O(p))\\
& \leq \frac{4}{\sqrt{\eps}}d^2(p,c')+(1+\sqrt{\eps})^2d^2(p,\mu_O(p)).
\end{align*}
Summarizing,
\begin{align}
& \sum_{c\in \bcalU- \calU}\sum_{p\in S_{\calQ}(c)}(d^2(p,\tilde{\mu}(p))-d^2(p,\mu_O(p))) \leq \sum_{c\in \bcalU- \calU}\sum_{p\in S_{\calQ}(c)}d^2(p,\tilde{\mu}(p))=
\sum_{c\in \bcalU- \calU}R(c)\leq \sum_{c'\in \calU- \bcalU}R(c')\nonumber \\
& =
\sum_{c'\in \calU- \bcalU}\sum_{p\in S_{\calQ}(c')}d^2(p,next(c'))
 \leq \sum_{c'\in \calU- \bcalU}\sum_{p\in S_{\calQ}(c')}(\frac{4}{\sqrt{\eps}}d^2(p,c')+(1+\sqrt{\eps})^2d^2(p,\mu_O(p)))\nonumber\\
 & \leq O({\sqrt{\eps}}) opt+ \sum_{c\in \calU- \bcalU}\sum_{p\in S_{\calQ}(c)}d^2(p,\mu_O(p)),\label{lem:successcheapremove:eqn3}
\end{align}
where in the last inequality above we used again the assumption on 
$\calQ$, i.e., property (2) of \Cref{lem:successguessprocess}, {and the fact that $\sum_{c\in \calU- \bcalU}\sum_{p\in S_{\calQ}(c)}d^2(p,\mu_O(p))\leq \clcost(M_O, \mu_O)\leq 11\,\sopt$ by \Cref{lemma:costboundofMOandMDwithassignments}}.
Property (c) follows by summing \eqref{lem:successcheapremove:eqn1}, \eqref{lem:successcheapremove:eqn2}, and \eqref{lem:successcheapremove:eqn3}.
\end{proof}

\subsection{Mixed Solutions after the Removal of $\bcalU$}    

We let $(S_{\calQ \cup \bcalU}, \tmu)$ be the output of a successful run of the $cheapRem$ procedure, i.e., where $\bcalU \in \calL_{cheap}$ and its associated assignment $\tmu$ satisfies the properties of \Cref{lem:successcheapremove}. Recall that $S_{\calQ \cup \bcalU}$ consists of the centers $S' = S - \calQ - \bcalU$ and the dummy centers $\dummyset$.
Further recall that $M_O$ consists of centers $S- S_0 = S - \calQ - \calU - \calR$ and $\ho$.  $M_D$ is the same set of centers except that $\ho$ is replaced by $\dummyset$. Notice that some centers of $M_O$ and $M_D$ are now removed if $\bcalU - \calU \neq \emptyset$. To take care of this, we modify these mixed solutions to obtain $M'_O$ and $M'_D$. 
 
 We first define $M'_O$; the definition of $M'_D$ is then very similar. The centers of $M'_O$ are $S- \calQ - \bcalU
 - \calR$ and $\ho$. Hence, the difference between $M_O$  and $M'_O$  is that in $M_O$ we remove $\calU$ from $S$ and in $M'_O$ we remove $\bcalU$. In other words,  $M'_O = M_O \cup \calU - \bcalU$. Similarly, we let $M'_D = M_D \cup \calU - \bcalU$. We update the assignment $\mu_O$ of $M_O$ to an assignment $\mu'_O$ of $M'_O$ (and the assignment $\mu_D$ to $\mu'_D$). 
 Recall that $\mu_O$ assigns each client $p$ to its closest center in $\ho$ (if $p\in \clientsexpensive$) or its closest center in $S- S_0$ (if $p\in \clients- \clientsexpensive$) except for those clients that belong to a cluster $S_{\calQ}(c)$ with $c\in \calU\cup \calR$. Indeed, the clients in $S_\calQ(c)$ are all assigned to the center of $\ho$ that is closest to $c$  if $c\in \calR$, and if $c\in \calU$ they are all assigned to the center in $S - S_0$ that is closest to $c$. We also recall that $\mu_D$ is the same as $\mu_O$ except when $\mu_O(p) \in \ho$ in which case $\mu_O(p)$ is replaced by its corresponding dummy center.

 \paragraph{Definitions of $\mu_O'$ and $\mu_D'$.}
 For clients $p\in S_{\calQ}(c)$ with $c \not\in \calU \cup \bcalU \cup \dummyset$, we define $\mu_O'(p) = \mu_O(p)$. For $p\in S_{\calQ}(c)$  with $c\in \dummyset$, we let $\mu'_O(p)$ be the closest center in $\ho$.
 Finally, for $p\in S_{\calQ}(c)$ with $c\in \calU \cup \bcalU$, we define $\mu'_O(p) =  \tmu(p)$. 

 The assignment $\mu'_D$ is obtained in the same way from $\mu_D$ with the difference that we use $\dummyset$ instead of $\ho$: 
 For clients $p\in S_{\calQ}(c)$ with $c \not\in \calU \cup \bcalU \cup \dummyset$, we define $\mu_D'(p) = \mu_D(p)$. For $p\in S_{\calQ}(c)$  with $c\in \dummyset$, we let $\mu'_D(p)$ be the closest center in $\dummyset$.
 Finally, for $p\in S_{{\calQ}}(c)$ with $c\in \calU \cup \bcalU$, we define $\mu'_D(p) =  \tmu(p)$. 
 
 This completes the definition of $\mu'_O$ and $\mu'_D$.
  We remark that we have $\mu'_O(p) = \mu_O(p)$ and $\mu'_D(p) = \mu_D(p)$  for all $p\in S_\calQ(c)$ with $c\not \in \calU \cup \bcalU \cup \dummyset$.
    
    We continue by arguing that $\mu'_O$ is well-defined, i.e., that $\mu'_O(p) \in M_O'$ for every $p\in \clients$ (the proof for $\mu'_D$ is the same).
 This is immediate for a client $p\in S_\calQ(c)$ with $c \in \dummyset$. For a client $p\in S_{{\calQ}}(c)$ with $c\in \cal U \cup \bcalU$, it holds because, by \Cref{lem:successcheapremove}, $\tmu(p)$ equals $c\in M'_O$ if  $c \in \calU -  \bcalU$ and otherwise $\tmu(p) \in  S- \calQ - \bcalU - \calR\subseteq M'_O$.
 It remains to verify that $\mu_O(p) \in M_O'$ for a client $p\in S_\calQ(c)$ with $c\not \in \calU \cup \bcalU \cup \dummyset$. If $c\in \calR$, we have $\mu_O(p) \in \ho \subseteq M'_O$.  Similarly, if $c \not \in \calU \cup \bcalU \cup \Lambda \cup \calR $ and $p\in \clientsexpensive$ we have $\mu_O(p) \in \ho \subseteq M'_O$.  In the remaining case when $c \not \in \calU \cup \bcalU \cup \Lambda \cup \calR $ and $p\in \clients- \clientsexpensive$, $\mu_O(p)$ equals $p$'s closest center in $S- S_0$. As   $p\in S_\calQ(c)$, we thus have   $\mu_O(p) = c \in S - (S_0 \cup \bcalU) \subseteq M'_O$. 

 The following upper bound on $\clcost(M'_O, \mu'_O) + \clcost(M'_D, \mu'_D)$ is a fairly immediate consequence of the assumption that the cheap-removal process is successful (Property 3c of Lemma \ref{lem:successcheapremove}). 
 Similarly to before, we say that the selection of $(W, \calB, \calQ, \bcalU,\tmu)$ is successful, if the sample $W$ selected in~\ref{sec:dsampleproc} is successful, the set of balls $\calB$ from \Cref{sec:ballguesses} is valid,  $\calQ$ selected in \Cref{sec:removalofExpensive} satisfies the properties of \Cref{lem:successguessprocess}, and $\bcalU,\tmu$  selected in this section satisfies the properties of \Cref{lem:successcheapremove}. 
 \begin{claim}
     If $(W, \calB, \calQ, \bcalU,\tmu)$ is successful,
     $$\clcost(M'_O, \mu'_O) + \clcost(M'_D, \mu'_D) \leq \sum_{p\in \clients} {10} \opt_p + O({\sqrt{\eps}} \cdot \sopt)\,.$$
     \label{claim:Mprimebounds}
 \end{claim}
 \begin{proof}
     For $x\in \{O, D\}$ the only differences between $\mu'_x$ and $\mu_x$ are clients in $S_{\calQ}(c)$ with $c\in \dummyset \cup \calU \cup \bcalU$.  Consider first a client $p \in S_{\calQ}(c)$  with $c \in \dummyset$. Then, as $c$ is the closest center among $(S - \calQ) \cup \dummyset$, which both contains $M_O - \ho$ and $M'_O - \ho$, we have that the closest center to $p$ in both $M'_O$ and $M_O$ is in $\ho$ (because by the definition of dummy centers, $d^2(p, \ho) \leq d^2(p, \dummyset)$). Similarly, the closest center to $p$ in both $M'_D$ and $M_D$ is in $\Lambda$. Hence, the definitions of $\mu'_O$ and $\mu'_D$ to assign $p$ to its closest center in $\ho$ and $\dummyset$, respectively, cannot increase the cost, i.e., $d^2(p, \mu'_O(p)) \leq d^2(p, \mu_O(p))$ and $d^2(p, \mu'_D(p)) \leq d^2(p, \mu_D(p))$ for such a client $p$. 
     
     Finally, for those clients $p\in S_\calQ(c)$ with $c\in\calU \cup \bcalU$ we have $\mu'_x = \tmu$ by definition and $d^2(p, \mu_O(p)) \leq d^2(p, \mu_D(p))$ since $\mu_O(p) = \mu_D(p)$ unless $\mu_O(p)\in \ho$ in which case $\mu_D(p)$ is the dummy center associated with $\mu_O(p)$ that can only be farther away from $p$ than $\mu_O(p)$.
     Hence
       \begin{gather*}
       \sum_{c \in \calU \cup \bcalU}  \sum_{p\in S_\calQ(c)} \left( d^2(p, \tmu(p)) - d^2(p, \mu_D(p)) \right) \leq \sum_{c \in \calU \cup \bcalU}  \sum_{p\in S_\calQ(c)} \left( d^2(p, \tmu(p)) - d^2(p, \mu_O(p)) \right) = O (\sqrt{\eps} \cdot {\sopt})\,,
        \end{gather*}
where in the equality we used Property (c) of \Cref{lem:successcheapremove}). The claim follows since
$$
\clcost(M'_O, \mu'_O) + \clcost(M'_D, \mu'_D)\leq \clcost(M_O, \mu_O) + \clcost(M_D, \mu_D) +O(\sqrt{\eps}\cdot \sopt) \overset{Lem. \ref{lemma:costboundofMOandMDwithassignments}}\leq 10 \opt_p + O(\sqrt{\eps} \cdot \sopt)\,.
$$
 \end{proof}

To better understand $\clcost(M'_O, \mu'_O)$, let us consider the cost $d(p, \mu'_O(p))$ of a single client $p$:
\begin{itemize}
    \item If $p \in S_{\calQ}(c)$ with $c\not\in \bcalU \cup \calR$ then $\tmu(p) = c$ by Property 3a of \Cref{lem:successcheapremove}.  Moreover,  $\mu'_O(p)$ assigns $p$ either to a center in $\ho$ or to one in $S - \calQ - \bcalU - \cal R$ and we have $d^2(p, c) \leq d^2(p, S - \calQ - \bcalU - \cal R)$.  Hence, in either case, we have
    \[
        d^2(p, \mu'_O(p)) \geq  d^2(p, \{\tmu(p)\} \cup \ho) = d^2(p, \{\tmu(p)\} \cup \ho \cup \dummyset) \,.
    \]
    \item If $p\in S_{\calQ}(c)$ with $c\in \calR$ then $\mu'_O(p) = \mu_O(p) =  c^*\in \ho$ where $c^*$ is the center of $\ho$ that is closest to $c$. So 
    \[
        d^2(p, \mu_O'(p)) = d^2(p, c^*)\,. 
    \]
    \item If $p \in S_{\calQ}(c)$ with $c\in \bcalU$ then $\mu'_O(p) = \tmu(p)$ and so
    \[
        d^2(p, \mu_O'(p)) \geq d^2(p, \{\tmu(p)\} \cup \ho \cup \dummyset).
    \]
\end{itemize}
Similarly, we can analyze $d^2(p, \mu_D(p))$ by replacing $\ho$ with the set $\dummyset$ of dummy centers (and $c^*$ by its associated dummy center). 
Summarizing, we have 
\begin{align*}
     \clcost(M'_O, \mu'_O) &\geq \sum_{c\in S_{\calQ} - \calR}\sum_{p\in S_{\calQ}(c)}  d^2\left(p, \{c\} \cup \ho \cup \dummyset\right) + \sum_{c\in \calR} \min_{c' \in \ho} \sum_{p\in S_{\calQ}(c)}  d^2\left(p, c' \right)\\
     \intertext{and}
     \clcost(M'_D, \mu'_D) &\geq \sum_{c\in S_{\calQ} - \calR}\sum_{p\in S_{\calQ}(c)}  d^2\left(p, \{c\}  \cup \dummyset\right) + \sum_{c\in \calR} \min_{c'\in \dummyset}\sum_{p\in S_{\calQ}(c)}  d^2\left(p,   c'\right)
\end{align*}

 Furthermore, by Properties 3a and 3b of \Cref{lem:successcheapremove}, we have that, for every $c\in \calR$, the set $S_{\calQ}(c)$ of clients assigned to $c$ in $S_{\calQ}$ equals the set $\tmu^{-1}(c)$ of clients assigned by $\tmu$. So, if we let $\{C_c\}_{c\in S_{\calQ \cup \bcalU}}$ be the partitioning of the set $\clients$ of clients according to $\tmu$, i.e., $C_c = \{p\in \clients \mid \tmu(p) =c\}$,  we can thus rewrite the above  bounds on the cost as 
\begin{align}
     \clcost(M'_O, \mu'_O) &\geq \sum_{c\not \in \calR}\sum_{p\in C_c}  d^2\left(p, \{c\} \cup \ho \cup \dummyset\right) + \sum_{c\in \calR}\min_{c' \in \ho}\sum_{p\in C_c}  d^2\left(p,  c'\right) \label{eq:MprimeObound}\\
     \intertext{and}
     \clcost(M'_D, \mu'_D) &\geq \sum_{c \not\in  \calR}\sum_{p\in C_c}  d^2\left(p, \{c\}  \cup \dummyset\right) + \sum_{c\in \calR}\min_{c' \in \dummyset}\sum_{p\in C_c}  d^2\left(p,  c'\right)\,.\label{eq:MprimeDbound}
\end{align}
In the next section, we use these bounds to give a polynomial-time algorithm that outputs a solution $S^*$ with $k$ centers so that (see \Cref{lemma:findingSstar})
\begin{align*}
    \clcost(S^*)  \leq
    \frac{(1+{6\sqrt{\eps}})}{2} \left(  \clcost(M'_O, \mu'_O) +  \clcost(M'_D, \mu'_D) \right) + O({\sqrt{\eps}}\cdot \sopt)
     \leq \sum_{p\in \clients} {5} \opt_p  + O({\sqrt{\eps}} \cdot \sopt)
\end{align*}
where the second inequality is by \Cref{claim:Mprimebounds}.
So the proof of \Cref{lemma:findingSstar}  in the next subsection is the final step in the proof of \Cref{thm:mainadditivecenters} (see also \Cref{sec:stable:everythingtogether} where we put everything together to prove \Cref{thm:mainadditivecenters}).

\subsection{Finding $S^*$ via Submodular Optimization}
\label{sec:submodularopt}

In this section, we give a polynomial-time algorithm for finding the solution $S^*$ by reducing the problem to maximizing a submodular function subject to a partition matroid constraint. Specifically, we prove the following lemma:
\begin{lemma}
    We can in polynomial-time find a clustering $S^*$ with $k$ centers of cost at most 
    \[
    \frac{(1+{6\sqrt{\eps}})}{2} \left(  \clcost(M'_O, \mu'_O) +  \clcost(M'_D, \mu'_D) \right) + O({\sqrt{\eps}}\cdot \sopt)\,.
     \]
    \label{lemma:findingSstar}
\end{lemma}
Recall that at this point the algorithm calculated the local search solution $S$, sampled $W$, guessed $\calB$, $\calQ$, and $\bcalU$ with the assignment $\tmu$ of the clustering $S_{\calQ \cup \bcalU}$. We assume that all these choices were successful guesses, i.e., that $(W, \calB, \calQ, \bcalU, \tmu)$ is successful.  Further recall the notation that we partition the clients into the clusters $\{C_c\}_{c\in S_{\calQ \cup \bcalU}}$ where $C_c = \{ p\in \clients : \tmu(p) =c\}$.  

Our goal is to give a polynomial-time algorithm that finds approximations of the sets $\ho$ and $\calR$ that only have slightly worse cost. We start by defining the feasible set of candidates for $\ho$ as a partition matroid. Recall that $\ho$ contains one center from each ball in $\calB$.

\paragraph{Definition of partition matroid $\calM_\calB$.} 
We define the partition matroid that captures the constraint that we wish to open a center in each ball in $\calB$. 
Let $\facilities_\calB$ be the  (multi) subset of facility/center locations containing  $B\cap \facilities$ for each ball  $B\in \calB$. If a center $c$ is in multiple balls in $\calB$, then $\facilities_\calB$ contains one distinct copy of $c$ for each ball. We let $\facilities_{\calB}(B) \subseteq F_\calB$ denote the centers associated with $B\in \calB$. These sets satisfy the following two properties:
\begin{itemize}
    \item The sets $\facilities_{\calB}(B)$ partition $\facilities_\calB$.
    \item For $B \in \calB$, $\facilities_{\calB}(B)$ contains (a copy of) every center in $\facilities\cap B$.
\end{itemize}
The first property holds since we took a unique copy of each center for each ball, and they are thus disjoint: $\facilities_\calB(B) \cap \facilities_\calB(B') = \emptyset$ for distinct $B,B' \in \calB$. Indeed, while it is not better for the cost to open multiple copies of a center, we make the copies to ensure the above two properties. This allows us to define the partition matroid $\calM_\calB = (\facilities_\calB, \calI)$ where
\[
    \calI = \{X \subseteq \facilities_\calB : |X\cap F_\calB(B)| \leq 1 \mbox{ for every } B\in \calB\}\,. 
\]
Moreover, since there is exactly one ball in $\calB$ for each center in $\ho$, we have $\ho \in \calI$ (where we slightly abuse notation as we should take the copy of center $c^*\in \ho$ that belongs to its associated ball). 

\paragraph{Core, concentrated and hit clusters.} Our goal is now to define a submodular function $f$ so that we obtain a good approximation to $\ho$ by maximizing $f$ over the matroid constraint $\calM_\calB$.  In particular, the domain of $f$ is every subset of $\facilities_\calB$. However, we need some additional steps before defining $f$. In particular, we introduce the concept of the core of a cluster $C_c$ and the notions of concentrated and hit clusters, which allow us to simplify the structure of centers in $\calR$. For a cluster $C_c$, we define the \emph{core} of $C_c$ as
\begin{gather*}
    \core_c = \left\{p \in C_c: d^2(p, c) \leq \eps \cdot \frac{\clcost(S)}{|\calR|\cdot |C_c|}\right\}\,.
\end{gather*}
We further say that cluster $C_c$ is \emph{concentrated} if 
\begin{gather*}
    |\core_c| \geq (1-\eps) |C_c|
\end{gather*}
and it is \emph{hit} by a set $X$ of centers if there is a point $p \in \core_c$ such that
\begin{gather*}
    d^2(p, X) < d^2(p,c)\,.
\end{gather*}
{For shortness we will sometimes say that a center is concentrated (resp., hit), if the corresponding cluster is so.}

\paragraph{Guessing the {centers} of $\calR$ that are not concentrated.} We further simplify the task of finding the set of centers $\calR$  by guessing the {centers} of $\calR$ that are not concentrated. Specifically, partition the set $\calR$ into $\calR_0$ and $\calR_1$, where $\calR_1$ contains those {centers} of $\calR$ that are concentrated and $\calR_0$ contains those that are not. We can  correctly guess $\calR_0$ in polynomial time since
the following simple claim shows that it is a subset of the {centers} whose {cluster} costs at least $\eps^2 \cdot \sopt/|\calR|$, of which there are only $O(|\calR|/\eps^2)$ many, 
and so guessing $\calR_0$ can thus be done in time
$2^{O(|\calR|/\eps^2)}$.
\begin{claim}
    Suppose that $C_c$ is not concentrated. Then the cost of $C_c$ is at least $\eps^2 \cdot \sopt/|\calR|$.
\end{claim}
\begin{proof}
    Since $C_c$ is not concentrated we have $|C_c - \core_c| \geq \eps |C_c|$. Moreover, each point $p\in C_c - \core_c$ has $d^2(p,c) > \eps \frac{\clcost(S)}{|\calR| \cdot |C_c|}$ by definition. Hence, as $\clcost(S) \geq \sopt$,
    \[
        \sum_{p\in C_c} d^2(p,c) \geq \sum_{p\in C_c - \core_c} d^2(p,c) \geq |C_c- \core_c|\cdot \eps \frac{\sopt}{|\calR| \cdot |C_c|} \geq \eps^2 \frac{\sopt}{|\calR|}\,.
    \]
\end{proof}

Now let $\calP$ be the potential centers that can be in $\calR$. Specifically, we let $\calP$ be the set that contains a center $c$ in $S -\calQ - \bcalU$ if $C_c$ equals the set of clients in $S_{\calQ}(c)$, where we recall that $S_{\calQ}(c)$ is the set of clients that are closest to $c$ in the clustering $S_{\calQ}$. Notice that the algorithm has all the information $S, \calB, \calQ, \bcalU$ and $\tmu$ to calculate $\calP$. Furthermore, by \Cref{lem:successcheapremove}, we have that no clients from $\bcalU$ were reassigned  by $\tmu$ to a center in $\calR$. So for $c\in \mathcal{R}$ {we} have $C_c = S_{\calQ}(c)$.  Hence, $\cal R \subseteq \calP$ and we can guess $\calR_0$ as follows:
\begin{mdframed}[hidealllines=true, backgroundcolor=gray!15]
\vspace{-5mm}
\paragraph{Guessing $\calR_0$}\ \\
\begin{enumerate}
    \item Let $\calC$ be the clusters of $\calP$ whose cost is at least $\eps^2 \cdot \sopt/|\calR|$.
    \item Output each subset of $\calC$.
\end{enumerate}
\end{mdframed}
By the above claim and the definition of $\calP$, one of the outputs is $\calR_0$.  Furthermore, the above guessing procedure outputs a family of polynomial many subsets. Indeed, we have $|\calR| \leq \log(n)/\eps^3$ by \Cref{lem:numnonpure}. Moreover, the total cost of the clusters in $\calP$ is $O({\sopt})$. To see this notice that the clusters corresponding to $\calP$ is a subset of the clusters of $S_{\calQ}$ and the cost of $S_{\calQ}$ is at most $\clcost(S - S_0  \cup \Lambda)$ (since $S_{\calQ}\supseteq S- S_0 \cup \Lambda$), which has cost at most $O(\sopt)$ by \Cref{lemma:S0properties}.  This implies that $|\calC| = O(\log(n)/\eps^5)$, so the total number of subsets is $n^{\eps^{-O(1)}}$. The algorithm proceeds by trying all possible subsets in the output, and we analyze the algorithm when it takes the correct guess of $\calR_0$.

\paragraph{Definition of the submodular function $f$. }
We first define another function $g$ on the same domain as $f$, i.e., on all subsets $X$ of $F_\calB$. We will then define $f$  by $f(X) = g(\emptyset) - g(X)$. Let the \emph{closed cost} of a cluster $C_c$ be defined as
\[
    \closedclcost_c(X) := 
    \begin{cases}
        \sum_{p\in C_c} d^2(p, \{c\} \cup X \cup \dummyset) & \mbox{if $C_c$ is hit by $X$,} \\
        \min_{c'\in X \cup \dummyset} \sum_{p\in \core_c} d^2(p, c') + \sum_{p\in C_c - \core_c} d^2(p, \{c\} \cup X \cup \dummyset) & \mbox{otherwise.} 
    \end{cases}
\]
Further, let $\calP_1$ be the potential centers for $\calR_1$: it contains each {center} $c\in \calP - \calR_0$ so that $c$ is concentrated, i.e., $|\core_c| \geq (1-\eps) |C_c|$. We remark that the algorithm can calculate this set $\calP_1$ as it only depends on $\calP$, the guessed set $\calR_0$, the value $\clcost(S)$, and the clusters $C_c$ defined by $\tmu$. Moreover, by definition, we have $\calR_1 \subseteq \calP_1$.
For a subset $X \subseteq \facilities_\calB$, we then then define $g(X)$ to be the minimum value of 
\begin{gather*}
     \sum_{c \not \in \calR_0 \cup \calR'_1}\sum_{p\in C_c}  d^2\left(p, \{c\} \cup X \cup \dummyset\right) + \sum_{c\in \calR_0} \sum_{p\in C_c} d^2(p, X \cup \dummyset) + \sum_{c\in \calR'_1} \closedclcost_c(X) 
\end{gather*}
over all subsets $\calR'_1 \subseteq \calP_1$ with $|\calR_1'| = |\calR_1| = |\calR| - |\calR_0|$. 
In words, over the best $\calR_1'$, $g(X)$ is the cost of the solution obtained by removing the centers $\calR_0 \cup \calR_1'$  and assigning clients as follows:
\begin{itemize}
    \item If $p\in C_c$ for a remaining center $c \not\in \calR_0 \cup  \calR_1'$ or $p \not \in \core_c$ with $c\in \calR'_1$, $p$ is assigned to its closest center in $\{c\} \cup X \cup \dummyset$.
    \item If $p \in C_c$ for a removed center $c\in \calR_0$, $p$ is assigned to its closest center in $X \cup \dummyset$.
    \item If $p \in \core_c$ for a removed center $c\in \calR'_1$, $p$ is assigned to its closest center in $\{c\} \cup X \cup \dummyset$ if  $X$ hits $C_c$, and otherwise all clients in $\core_c$ are assigned to the same center $c'\in X \cup \dummyset$ that minimizes the cost.
\end{itemize}

We remark that this assignment is infeasible in the sense that it may assign clients to removed centers in $\calR_1$. Nevertheless, we relate the values of $g(\emptyset)$ and $g(\ho)$ to $\clcost(M'_D, \mu'_D)$ and $\clcost(M'_O, \mu'_O)$, respectively; and, we show that given an $X \subseteq \facilities_\calB$ that is independent in $M_\calB$, i.e., $X\in \calI$, we can in polynomial-time output $k$ centers {$S^*$} whose cost is at most $(1+{6\sqrt{\eps}})g(X) + {O(\sqrt{\eps}}) \sopt$.
Finally, the definition of $g$ allows us to prove that $f$ (defined by $f(X) = g(\emptyset) - g(X)$) is a {non-negative} monotone submodular function.
\begin{lemma}
We have that $g$ and $f$  satisfy the following properties:
\begin{enumerate}
\item We can evaluate $g(X)$ in polynomial time for every {given} $X \subseteq \facilities_\calB$, and we can thus evaluate $f(X)$ in polynomial time. 
\item $f$ is a non-negative monotone submodular function.
\item The value $g(\emptyset)$ is at most $\clcost(M'_D, \mu'_D)$.
\item The value $g(\ho)$ is at most $\clcost(M'_O, \mu'_O)$.
\item  Given an $X \subseteq \facilities_\calB$ that is independent in $M_\calB$, i.e., $X\in \calI$, we can in polynomial-time output $k$ centers {$S^*$} whose {associated} cost is at most $(1+{6\sqrt{\eps}}) g(X) + {O(\sqrt{\eps})} \sopt$.
\end{enumerate}
\label{lemma:propertiesfandg}
\end{lemma}
We give the proof of the lemma in the next subsection. We explain here how it implies \Cref{lemma:findingSstar}.

\begin{proof}[Proof of \Cref{lemma:findingSstar}]

By the first property of \Cref{lemma:propertiesfandg}, we can evaluate $f$ in polynomial time. Moreover, it is easy to see that we can answer independence queries in polynomial time for any partition matroid, particularly for $\calM_\calB$. We can thus apply \Cref{thm:submodularmatroidoptimization} on $f$ and $\calM_\calB$  with $\zeta =  (1-1/e - 1/2)$
to find a solution $X \in \calI$ such that\footnote{We remark that we could select $\zeta$ to be arbitrarily small and get a better guarantee in the statement of \Cref{lemma:findingSstar}. We have chosen this value to simplify the calculations as improving the guarantee of \Cref{lemma:findingSstar} does not improve the overall result. }
\begin{gather*}
    g(\emptyset) - g(X) =  f(X) \geq (1-1/e - \zeta)f(\ho) = (1-1/e - \zeta)( g(\emptyset) - g(\ho)) = \frac{1}{2} ( g(\emptyset) - g(\ho))\,,
\end{gather*}
where we used that $\ho$ is one feasible solution, i.e., $\ho \in \mathcal{I}$.
This in turn implies that 
\[
\frac{1}{2} \left(g(\ho)  +  g(\emptyset) \right) \geq  g(X)\,.
\]

Using the upper bounds on $g(\emptyset)$ and $g(\ho)$ of the above lemma we thus have found a set $X$ such that 
\begin{gather*}
 \frac{1}{2} (\clcost(M'_O, \mu'_O) +  \clcost(M'_D, \mu'_D)) \geq g(X)\,.
\end{gather*}
Now using the last property of the above lemma we can output a solution $S^*$ whose cost is at most
\begin{gather*}
    \frac{(1+{6\sqrt{\eps}})}{2} \left(  \clcost(M'_O, \mu'_O) +  \clcost(M'_D, \mu'_D) \right) + O({\sqrt{\eps}}\cdot \sopt)\,.
\end{gather*}
as required. 

Finally, each of the above steps runs in polynomial time: the algorithm of \Cref{thm:submodularmatroidoptimization} is polynomial time, and the last property of \Cref{lemma:propertiesfandg} used to obtain $S^*$ is polynomial-time. 
Moreover, the number of guesses of $\calR_0$ is at most $n^{\eps^{-O(1)}}$, as argued after the description of that procedure. So we can, in polynomial time, try all possibilities and, among all solutions found (one for each guess of $\calR_0$), return one that minimizes the cost and, in particular, has cost at most that of $S^*$ (which was analyzed assuming the guess of $\calR_0$ was correct). We thus have  a polynomial time algorithm that returns a solution that satisfies the guarantee of the lemma.
\end{proof}

\subsubsection{Proof of \Cref{lemma:propertiesfandg}}

\paragraph{Proof of Property 1.} Given $X \subseteq \facilities_{\calB}$, we argue that we can evaluate $g(X)$ in polynomial time. For a center $c \in \calP_1$, define 
\[
\increase_c(X) = \closedclcost_c(X) - \sum_{p\in C_c} d^2(p, \{c\} \cup X \cup \dummyset)\,.
\]
In other words, using that $d^2(p,c ) \leq d^2(p, X)$ for $p\in \core_c$ if $C_c$ is not hit by $X$,
\begin{align}
\increase_c(X) = 
    \begin{cases}
        0 & \mbox{if $C_c$ is hit by $X$,} \\
        \min_{c'\in X \cup \dummyset} \sum_{p\in \core_c} (d^2(p, c') - d^2(p,\{c\}  \cup \dummyset))  & \mbox{otherwise.} 
    \end{cases}
\label{eq:increase}
\end{align}
We remark that $\increase_c(X)\geq 0$ by definition. 
Let $\clients_0$ denote the clients belonging to clusters $C_c$ with $c\in \calR_0$.
With this notation, $g(X)$ is the minimum value
\begin{align}
    \sum_{p\in \clients - \clients_0} d^2(p, \{\tmu(p)\} \cup X \cup \dummyset) + \sum_{p\in \clients_0} d^2(p, X \cup \dummyset) +  \sum_{c \in \calR'_1} \increase_c(X)
    \label{eq:nicedefofg}
\end{align}
over all subsets $\calR'_1 \subseteq \calP_1$ with $|\calR'_1| = |\calR_1|$. We further have that the value of $\increase_c(X)$ for $c\in \calR'_1$ is independent of other centers in $\calR'_1$. We can thus obtain the best choice of $\calR'_1$ by evaluating $\increase_c(X)$ for every $c\in \calP_1$ and choose the $|\calR_1|$ ones of smallest value. Notice that the algorithm knows $|\calR_1|$ since it equals $|S_0| - |\calQ| - |\bcalU| - |\calR_0| =|\calB| - |\calQ| - |\bcalU| - |\calR_0|$.     After we have obtained $\calR'_1$ in polynomial time, we can evaluate $g(X)$ in polynomial time by simply calculating the sum.

\paragraph{Proof of Property 2.} In the proof of this property it will be convenient to use definition~\eqref{eq:nicedefofg} of $g$.
We start by verifying that $f$ is monotone. I.e., that $f(Y) \geq f(X)$ when $X \subseteq Y$, which is equivalent to verifying that $g(Y) \leq g(X)$.  Let $g(X)$ equal
\[
     \sum_{ p\in \clients- \clients_0} d^2(p, \{\tmu(p)\} \cup X \cup \dummyset) + \sum_{p\in \clients_0} d^2(p, X \cup \dummyset)+ \sum_{c \in \calR'_1} \increase_c(X)
\]
for some set $\calR'_1$. Then $g(Y)$ is at most
\[
     \sum_{p\in \clients- \clients_0} d^2(p, \{\tmu(p)\} \cup Y \cup \dummyset)+ \sum_{p\in \clients_0} d^2(p, {Y} \cup \dummyset) + \sum_{c \in \calR'_1} \increase_c(Y)
\]
As trivially $d^2(p, \{\tmu(p)\} \cup Y \cup \dummyset) \leq d^2(p, \{\tmu(p)\}\cup X \cup \dummyset)$, $d^2(p,  Y \cup \dummyset) \leq d^2(p,  X \cup \dummyset)$ and $\increase_c(Y) \leq \increase_c(X)$ (see~\eqref{eq:increase}),  we have  $g(Y) \leq g(X)$ as required. Moreover, as $f(\emptyset) = 0$ by definition, non-negativity follows.

We proceed to verify that $f$ is submodular, i.e., that for every $X \subseteq Y \subseteq \facilities_B$ and $c_0 \in \facilities_B - Y$, $f(X \cup \{c_0\}) - f(X) \geq f(Y \cup \{c_0\}) - f(Y)$, or equivalently
\begin{align}
    g(X \cup \{c_0\}) - g(X) \leq g(Y \cup \{c_0\}) - g(Y)\,.
    \label{eq:submodularineq}
\end{align}
 Let $\calR_1^X, \calR_1^Y$, and  ${\calR}_1^{Y+c_0}$ be subsets of $\calP_1$ of cardinality $|\calR_1|$ such that
\begin{align*}
g(X) &=  \sum_{p\in \clients- \clients_0} d^2(p, \{\tmu(p)\} \cup X \cup \dummyset) + \sum_{p\in \clients_0} d^2(p, X \cup \dummyset) + \sum_{c \in \calR^X_1} \increase_c(X)\\
g(Y) &=  \sum_{p\in \clients- \clients_0} d^2(p, \{\tmu(p)\} \cup Y \cup \dummyset)+ \sum_{p\in \clients_0} d^2(p, Y \cup \dummyset) + \sum_{c \in \calR^Y_1} \increase_c(Y)\\
g(Y \cup \{c_0\}) & = \sum_{p\in \clients- \clients_0} d^2(p, \{\tmu(p)\} \cup (Y \cup \{c_0\}) \cup \dummyset)+ \sum_{p\in \clients_0} d^2(p, (Y \cup \{c_0\}) \cup \dummyset) \\
& + \sum_{c \in {\calR}^{Y+c_0}_1} \increase_c(Y\cup \{c_0\})
\end{align*}
We shall define  ${\calR}^{X+c_0}_1$  to be a subset of $\calP_1$ of cardinality $|\calR_1|$ so that  the upper bound
\begin{align*}
g(X\cup \{c_0\}) & \leq  \sum_{p\in \clients- \clients_0} d^2(p, \{\tmu(p)\} \cup (X\cup \{c_0\}) \cup \dummyset)+ \sum_{p\in \clients_0} d^2(p, (X \cup \{c_0\}) \cup \dummyset) \\
& + \sum_{c \in {\calR}^{X+c_0}_1} \increase_c(X \cup \{c_0\})
\end{align*}
allows us to verify Inequality~\eqref{eq:submodularineq}. First notice that no matter the definition of ${\calR}^{X+c_0}_1$, we have that
\begin{align}
& \sum_{p\in \clients- \clients_0} (d^2(p, \{\tmu(p)\} \cup (X\cup \{c_0\}) \cup \dummyset) -  d^2(p, \{\tmu(p)\} \cup X \cup \dummyset) )\nonumber \\
& +\sum_{p\in \clients_0} (d^2(p, (X\cup \{c_0\}) \cup \dummyset) -  d^2(p,  X \cup \dummyset) )
 \label{eq:boringsubmodular1}
\end{align}
is upper bounded by
\begin{align}
& \sum_{p\in \clients- \clients_0} (d^2(p, \{\tmu(p)\} \cup (Y \cup \{c_0\}) \cup \dummyset) - 
 d^2(p, \{\tmu(p)\} \cup Y \cup \dummyset)) \nonumber \\
 & + \sum_{p\in \clients_0}(d^2(p, (Y \cup \{c_0\}) \cup \dummyset) - 
 d^2(p, Y \cup \dummyset))\,.
 \label{eq:boringsubmodular2}
\end{align}
Indeed, if we let $\clients' \subseteq \clients - \clients_0$ be the subset of clients for which $d^2(p, c_0) < d^2(p, \{\tmu(p)\} \cup Y \cup \dummyset)$ and $\clients'_0 \subseteq\clients_0$ be the subset of clients for which $d^2(p, c_0) < d^2(p,  Y \cup \dummyset)$, then~\eqref{eq:boringsubmodular2} equals
\[
    \sum_{p\in \clients'} (d^2(p, c_0) - d^2(p, \{\tmu(p)\} \cup Y \cup \dummyset))+ \sum_{p\in \clients'_0} (d^2(p, c_0) - d^2(p,  Y \cup \dummyset))
\]
and we have 
\[
    \eqref{eq:boringsubmodular1} \leq \sum_{p\in \clients'} (d^2(p, c_0) - d^2(p, \{\tmu(p)\} \cup X \cup \dummyset)) + \sum_{p\in \clients'_0} (d^2(p, c_0) - d^2(p,  X \cup \dummyset))\,.
\]
It follows that \eqref{eq:boringsubmodular1} is at most \eqref{eq:boringsubmodular2} because
\[
    \sum_{p\in \clients'} ( d^2(p, \{\tmu(p)\} \cup Y \cup \dummyset)  - d^2(p, \{\tmu(p)\} \cup X \cup \dummyset)) + \sum_{p\in \clients'_0} ( d^2(p,  Y \cup \dummyset)  - d^2(p,  X \cup \dummyset)) \leq 0\,.
\]

To prove that $f$ is submodular, it is thus sufficient to define ${\calR}^{X+c_0}_1$ so that 
\begin{align}
    \sum_{c\in {\calR}^{X+c_0}_1} \increase_c(X \cup \{c_0\}) - \sum_{c\in {\calR}^X_1} \increase_c(X)\leq \sum_{c\in {\calR}^{Y+c_0}_1} \increase_c(Y \cup \{c_0\}) - \sum_{c\in {\calR}^Y_1} \increase_c(Y)\,.
    \label{eq:mainsubmodular}
\end{align}
To this end, let $\Delta_{\text{new}}$ contain the centers in ${\calR}^{Y+c_0}_1 - {\calR}^Y_1$ and those centers  $c \in {\calR}^{Y+c_0}_1 \cap \calR^Y_1$ with $\increase_c(Y \cup \{c_0\}) < \increase_c(Y)$. Notice that we may assume that, for each $c\in \Delta_{\text{new}}$,
\[
\increase_c(Y \cup \{c_0\})= \increase_c(\{c_0\}) =  \begin{cases}
        0 & \mbox{if $C_c$ is hit by $c_0$,} \\
         \sum_{p\in \core_c} (d^2(p, c_0) - d^2(p,\{c\}  \cup \dummyset))  & \mbox{otherwise.} 
    \end{cases}
\]
because the addition of $c_0$  caused a decrease in $\increase(\cdot)$ for these clusters.   Further, let $\Delta_{\text{old}}^Y$ contain $\Delta_{\text{new}} \cap \calR^{Y}_1$ and the centers $\calR^{Y}_1 - \calR^{Y+c_0}_1$. So $\calR^{Y+c_0}_1 = (\calR^{Y}_1 - \Delta_{\text{old}}^Y) \cup \Delta_{\text{new}}$ and $|\Delta_{\text{new}}| = |\Delta_{\text{old}}^Y|$. We now define $\calR^{X+c_0}_1 = (\calR^{X}_1 - \Delta_{\text{old}}^X) \cup \Delta_{\text{new}}$ where $\Delta^X_{\text{old}}$ is obtained as follows:
\begin{itemize}
    \item Initialize $\Delta_{\text{old}}^X =\calR^X_1\cap \Delta_{\text{new}}$. 
    \item While $|\Delta_{\text{old}}^X|< |\Delta_{\text{new}}|$, add a  center according to the following priorities:
    \begin{itemize}
        \item If there is a center $c\in \calR^X_1 - \Delta_{\text{old}}^X$ so that $c\not \in \calR^{Y}_1$, add $c$ to $\Delta_{\text{old}}^X$.
        \item Else add a center $c\in \calR^X_1- \Delta_{\text{old}}^X$ with $c \in \Delta_{\text{old}}^Y$. 
    \end{itemize}
\end{itemize}
We claim that the above is well-defined, i.e., that if $|\Delta_{\text{old}}^X|< |\Delta_{\text{new}}|$ and there is no center 
$c\in \calR^X_1 - \Delta_{\text{old}}^X$ so that $c\not \in \calR^{Y}_1$, then there must be a center $c\in \calR^X_1- \Delta_{\text{old}}^X$ with $c \in \Delta_{\text{old}}^Y$. This is because $|\Delta_{\text{old}}^X| < |\Delta_{\text{new}}| = |\Delta_{\text{old}}^Y|$ and all remaining centers $\calR_1^X - \Delta_{\text{old}}^X$  are in $\calR^Y_1$ in this case, so at least one of them must be in $\Delta_{\text{old}}^Y$ (recall that $|\calR_1^X| = |\calR_1^Y|$).

With this notation we have 
\begin{align}
   \notag \sum_{c\in {\calR}^{Y+c_0}_1} \increase_c(Y \cup \{c_0\}) - \sum_{c\in {\calR}^Y_1} \increase_c(Y ) & =  \sum_{c\in \Delta_{\text{new}}} \increase_c(Y \cup \{c_0\}) - \sum_{c\in \Delta_{\text{old}}^Y}\increase_{c}(Y) \\
    &=  \sum_{c\in \Delta_{\text{new}}} \increase_c(\{c_0\}) - \sum_{c\in \Delta_{\text{old}}^Y}\increase_{c}(Y)\,. \label{eq:firstsubmodular}
\end{align}
Similarly,
\begin{align}
    \sum_{c\in {\calR}^{X+c_0}_1} \increase_c(X \cup \{c_0\}) - \sum_{c\in {\calR}^X_1} \increase_c(X ) 
    & \leq \sum_{c\in \Delta_{\text{new}}} \increase_c(X \cup \{c_0\}) - \sum_{c\in \Delta_{\text{old}}^X} \increase_{c}(X) \notag\\ 
    & = \sum_{c\in \Delta_{\text{new}}} \increase_c(\{c_0\}) - \sum_{c\in \Delta_{\text{old}}^X} \increase_{c}(X)\,, \label{eq:secondsubmoular}
\end{align}
where we used that $\increase_c(X \cup \{c_0\}) \leq \increase_c(X)$ for any center $c$.

To analyze this let, $\pi$ be a bijection from $\Delta_{\text{old}}^Y$ to $\Delta_{\text{old}}^X$ so that $\pi(c) =c$ for all $c\in \Delta_{\text{old}}^Y \cap \Delta_{\text{old}}^X$. 
\begin{claim}
    For every $c\in \Delta_{\text{old}}^Y$ we have
    \[
    \increase_c(Y) \leq \increase_{\pi(c)}(X)\,.
    \]
\end{claim}
\begin{proof}
    If $c\in \Delta_{\text{old}}^Y \cap \Delta_{\text{old}}^X$, the claim holds because $\increase_c(Y) \leq \increase_c(X)$ for every $c\in \calP_1$. Otherwise, we have $\pi(c) \not \in \calR^Y_1$. This is because $(\Delta_{\text{old}}^X - \Delta_{\text{new}}^Y) \cap \calR_1^Y = \emptyset$. Indeed, in the construction of $\Delta_{\text{old}}^X$, we initialized with $\calR^X_1\cap \Delta_{\text{new}}$ and, as $\Delta_{\text{old}}^Y$ contains $\Delta_{\text{new}} \cap \calR^Y_1$, all elements in  $\calR^X_1\cap \Delta_{\text{new}}$ are either in $\Delta_{\text{old}}^Y$ or not in $\calR^Y_1$. The property is then maintained by definition of the while-loop in the construction of $\Delta_{\text{old}}^X$. 

    It follows that $\calR^Y_1$ where we remove $c$ and add $\pi(c)$ is another subset of $\calP_1$ of cardinality $|\calR_1|$. Hence, by the selection of $\calR^Y_1$ (to be the collection of centers with smallest $\increase(\cdot)$) we have
    \[
        \increase_c(Y) \leq \increase_{\pi(c)}(Y) \leq \increase_{\pi(c)}(X)\,.
    \]
\end{proof}
By the above claim,
\begin{align*}
\sum_{c\in \Delta_{\text{old}}^Y}\increase_{c}(Y) - \sum_{c\in \Delta_{\text{old}}^X} \increase_{c}(X)  = \sum_{c\in \Delta_{\text{old}}^Y}(\increase_{c}(Y) - \increase_{\pi(c)}(X)) \leq 0
\end{align*}
and thus by the above arguments (see~\eqref{eq:firstsubmodular} and~\eqref{eq:secondsubmoular})
\[
    \sum_{c\in {\calR}^{X+c_0}_1} \increase_c(X \cup \{c_0\}) - \sum_{c\in {\calR}^X_1} \increase_c(X)\leq \sum_{c\in {\calR}^{Y+c_0}_1} \increase_c(Y \cup \{c_0\}) - \sum_{c\in {\calR}^Y_1} \increase_c(Y)\,.
\]
We have thus verified~\eqref{eq:mainsubmodular},  which concludes the proof that $f$ is submodular.

\paragraph{Proof of Property 3.} We have
\begin{align*}
    g(\emptyset) &\leq 
     \sum_{c \not \in \calR}\sum_{p\in C_c}  d^2\left(p, \{c\}  \cup \dummyset\right) + \sum_{p\in \clients_0} d^2(p, \dummyset) + \sum_{c\in \calR_1} \closedclcost_c(\emptyset)  \\
     &\leq \sum_{c \not \in \calR_1}\sum_{p\in C_c}  d^2\left(p, \{c\}  \cup \dummyset\right)+ \sum_{p\in \clients_0} d^2(p, \dummyset) + \sum_{c\in \calR_1} \min_{c'\in \dummyset} \sum_{p\in C_c} d^2(p, c') \\
     & \leq \sum_{c \not \in \calR}\sum_{p\in C_c}  d^2\left(p, \{c\}  \cup \dummyset\right) + \sum_{c\in \calR} \min_{c'\in \dummyset} \sum_{p\in C_c} d^2(p, c')\,\overset{\eqref{eq:MprimeDbound}}{\leq}\clcost(M'_D, \mu'_D).
\end{align*}

\paragraph{Proof of Property 4.} Similarly to the third property, we have
\begin{align*}
    g(\ho) &\leq 
     \sum_{c \not \in \calR}\sum_{p\in C_c}  d^2\left(p, \{c\} \cup \ho \cup \dummyset\right) + \sum_{p\in \clients_0}d^2(p, \ho \cup \dummyset)+  \sum_{c\in \calR_1} \closedclcost_c(\ho)  \\
     &\leq \sum_{c \not \in \calR}\sum_{p\in C_c}  d^2\left(p, \{c\}   \cup \ho \cup \dummyset\right)+ \sum_{p\in \clients_0} d^2(p, \ho \cup \dummyset) + \sum_{c\in \calR_1} \min_{c'\in \ho} \sum_{p\in C_c} d^2(p, c') \\
     & \leq \sum_{c \not \in \calR}\sum_{p\in C_c}  d^2\left(p, \{c\} \cup \ho \cup \dummyset\right) + \sum_{c\in \calR} \min_{c'\in \ho} \sum_{p\in C_c} d^2(p, c')\overset{\eqref{eq:MprimeObound}}{\leq}\,\clcost(M'_O, \mu'_O)\,.
\end{align*}

\paragraph{Proof of Property 5.}

    Let  $X \subseteq \facilities_\calB$ be independent in $M_\calB$, i.e., $X\in \calI$. We give a polynomial-time algorithm that outputs $k$ centers $S^*$ whose cost is at most $(1+{6\sqrt{\eps}})g(X) +{O(\sqrt{\eps})} \sopt$. Let $\calR'_1$ be a subset of $\calP_1$ of cardinality $|\calR_1|$ so that
\begin{gather*}
   g(X) =   \sum_{c \not \in \calR_0 \cup \calR'_1}\sum_{p\in C_c}  d^2\left(p, \{c\} \cup X \cup \dummyset\right)+ \sum_{p \in \clients_0} d^2(p, X \cup \dummyset) + \sum_{c\in \calR'_1} \closedclcost_c(X) 
\end{gather*}
    By the arguments in the proof of the first property, we can calculate $\calR'_1$ in polynomial time. Now we obtain $S^*$ from $S_{\calQ \cup \bcalU}$ as follows:
    \begin{itemize}
    \item Initialize $S^* = S_{\calQ \cup \bcalU} = S \cup \dummyset - \calQ - \bcalU$.
    \item Remove centers $\calR_0$ and $\calR'_1$ from $S^*$ to obtain a set of cardinality $k$. 
    \item Finally, for each ball in $\calB$, if $X$ contains a center $c$ in that ball, replace the associated dummy center by $c$; otherwise, replace the dummy center with an arbitrary center in the ball.  
    \end{itemize}
    By the definition of the dummy centers $\dummyset$, the cost of $S^*$ is at most
    \[
     \sum_{c \not \in \calR_0 \cup \calR'_1}\sum_{p\in C_c}  d^2\left(p, \{c\} \cup X \cup \dummyset\right)+ \sum_{p \in \clients_0} d^2(p, X \cup \dummyset) + \sum_{c\in \calR'_1}  \sum_{p\in C_c} d^2(p, X \cup \dummyset)\,.
    \]
    To relate this to $g(X)$ we need to compare $\closedclcost_c(X)$ to $\sum_{p\in C_c} d^2(p, X\cup \dummyset)$ for every $c\in \calR'_1$.
    Recall that 
\[
    \closedclcost_c(X) = 
    \begin{cases}
        \sum_{p\in C_c} d^2(p, \{c\} \cup X \cup \dummyset) & \mbox{if $C_c$ is hit by $X$,} \\
        \min_{c'\in X \cup \dummyset} \sum_{p\in \core_c} d^2(p, c') + \sum_{p\in C_c - \core_c} d^2(p, \{c\} \cup X \cup \dummyset) & \mbox{otherwise.} 
    \end{cases}
\]
Consider a cluster $C_c$ with $c\in \calR'_1$ and partition $C_c- \core_c$  into sets $A$ and $B$ where $A$ contain those clients $p$ so that $d^2(p, \{c\} \cup X \cup \dummyset) =d^2(p, c)$ and $B$ contain the remaining clients $p$ with $d^2(p, \{c\} \cup X \cup \dummyset) = d^2(p, X \cup \dummyset)$. Suppose first that there is a client $p_0 \in \core_c$ so that $d^2(p_0, X \cup \dummyset) < d^2(p_0, c)$, i.e., $C_c$ is hit by $X \cup \dummyset$.
By the definition of $\core_c$, we have  $d^2(p_0, c) \leq \eps \cdot \frac{\clcost(S)}{|\calR| \cdot |C_c|}\leq {\apxLSkmeans} \eps \cdot \frac{\sopt}{|\calR| \cdot |C_c|}$ where we used that $\clcost(S) \leq {\apxLSkmeans} \cdot \sopt$ by \Cref{lemma:localsearchis5approximation}. So,  by {\Cref{lem:apxTriangleInequality3} with $\gamma=1+\sqrt{\eps}$},
\begin{align*}
\sum_{p\in C_c} d^2(p, X \cup \dummyset)& \leq  \sum_{p\in C_c - B} ({1+\sqrt{\eps}})d^2(p, c) + {\frac{3}{\sqrt{\eps}}}(d^2(c, p_0) + d^2(p_0, X\cup \dummyset))) + \sum_{p\in B} d^2(p, X \cup \dummyset)\\
&\leq \sum_{p\in A} ({1+\sqrt{\eps}})d^2(p, c) + \sum_{p\in B} d^2(p, X \cup \dummyset)  + {O(\sqrt{\eps})} \cdot \frac{\sopt}{|\calR|}\,,
\end{align*}
where we used $d^2(p_0, X\cup \dummyset) \leq d^2(p_0,c) \leq {\apxLSkmeans}\eps \cdot \frac{\sopt}{|\calR| \cdot |C_c|}$ and $d^2(p, c) \leq {\apxLSkmeans}\eps \cdot \frac{\sopt}{|\calR| \cdot |C_c|}$ for any $p\in \core_c$ for the last inequality. Moreover, by definition of $A$ and $B$, we have $d^2(p, c) = d^2(p,  \{c\} \cup X \cup \dummyset)$ for any $p\in A$ and $d^2(p, X \cup \dummyset) = d^2(p, \{c\} \cup X \cup \dummyset)$ for any $p\in B$. Hence, in this case,
\begin{align*}
    \sum_{p\in C_c} d^2(p, X \cup \dummyset) & \leq \sum_{p\in C_c - \core_c} ({1+\sqrt{\eps}})d^2(p, \{c\} \cup X \cup Y) + {O(\sqrt{\eps})} \cdot \frac{\sopt}{|\calR|} \\
    & \leq   ({1+\sqrt{\eps}})\closedclcost_c(X) + {O(\sqrt{\eps})} \cdot \frac{\sopt}{|\calR|}\,.
\end{align*}

Let us now consider a cluster $C_c$ with $c\in \calR'_1$ that is not hit by $X \cup \dummyset$, particularly not by $X$.
Then 
\begin{align*}
    \closedclcost_c(X) &= \min_{c'\in X \cup \dummyset} \sum_{p\in \core_c} d^2(p, c') + \sum_{p\in C_c - \core_c} d^2(p, \{c\} \cup X \cup \dummyset) \\
    & \geq \sum_{p\in \core_c \cup B} d^2(p, X \cup \dummyset) + \sum_{q\in A} d^2(q, c)
\end{align*}
{Using again \Cref{lem:apxTriangleInequality3} with $\gamma=1+\sqrt{\eps}$,} we further have 
\begin{align*}
  &  \sum_{p\in C_c} d^2(p, X \cup \dummyset)  = \sum_{p\in \core_c \cup B} d^2(p, X \cup \dummyset) + \sum_{q\in A} d^2(q, X \cup \dummyset)\\
 \leq & \sum_{p\in \core_c \cup B} d^2(p, X \cup \dummyset) + \sum_{q\in A} \frac{1}{|\core_c|}\sum_{p\in \core_c} ({(1+\sqrt{\eps})}d^2(q, c) + {\frac{3}{\sqrt{\eps}}}(d^2(c, p) + d^2(p, X \cup \dummyset))).
\end{align*}
Now using that $d^2(c,p) \leq {\apxLSkmeans} \eps \cdot \frac{\sopt}{|\calR| \cdot |C_c|}$ {for $p\in \core_c$} we obtain the upper bound
\begin{align*}
     \sum_{p\in C_c} d^2(p, X \cup \dummyset) & \leq \left(1 + {\frac{3}{\sqrt{\eps}}}\frac{|A|}{|\core_c|}\right) \sum_{p\in \core_c\cup B} d(p, X \cup \dummyset) + \sum_{q\in A} ({1+\sqrt{\eps}})d^2(q,c) + {O(\sqrt{\eps})} \frac{\sopt}{|\calR|} \\
     &\leq \left(1 + \frac{{3\sqrt{\eps}}}{(1-\eps)}\right) \sum_{p\in \core_c \cup B} d^2(p, X \cup \dummyset) + \sum_{q\in A} ({1+\sqrt{\eps}})d^2(q,c) + {O(\sqrt{\eps})} \frac{\sopt}{|\calR|} \\
     & \leq (1+ {6\sqrt{\eps}}) \closedclcost_c(X) + {O(\sqrt{\eps})} \frac{\sopt}{|\calR|}
\end{align*}
where for the second inequality we used that every cluster in $\calR_1' \subseteq \calP_1$ is concentrated. {In more detail $|\core_c|\geq (1-\eps)|C_c|$, hence $|A|\leq |C_c-\core_c|\leq \eps|C_c|$}.

Hence, for any $c\in \calR_1'$ we proved that $\sum_{p\in C_c} d^2(p,X \cup \dummyset) \leq (1+{6\sqrt{\eps}}) \closedclcost_c(X) + {O(\sqrt{\eps})} \frac{\sopt}{|\calR|}$. We thus get that the cost of $S^*$ is at most

    \begin{align*}
    & \sum_{c \not \in \calR_0 \cup \calR'_1}\sum_{p\in C_c}  d^2\left(p, \{c\} \cup X \cup \dummyset\right)+ \sum_{p \in \clients_0} d^2(p, X \cup \dummyset) + \sum_{c\in \calR'_1}\left( (1+{6\sqrt{\eps}}) \closedclcost_c(X) + {O(\sqrt{\eps})} \frac{\sopt}{|\calR|}\right) \\
    \leq &  (1+{6\sqrt{\eps}}) \left(\sum_{c \not \in \calR_0 \cup \calR'_1}\sum_{p\in C_c}  d^2\left(p, \{c\} \cup X \cup \dummyset\right)+ \sum_{p \in \clients_0} d^2(p, X \cup \dummyset) + \sum_{c\in \calR'_1} \closedclcost_c(X) \right) + {O(\sqrt{\eps})} \sopt
    \end{align*}
    which equals $(1+{6\sqrt{\eps}}) g(X) + {O(\sqrt{\eps})} \sopt$. 

\subsection{Putting Everything Together 
-- Proof of \Cref{thm:mainadditivecenters}}
\label{sec:stable:everythingtogether}
We now turn to proving \Cref{thm:mainadditivecenters}. We first describe the algorithm we analyse.
\begin{mdframed}[hidealllines=true, backgroundcolor=gray!15]
\vspace{-5mm}
\paragraph{A $({5+O(\sqrt{\eps})})$-Approximation for $k$-Means on $\eps/\log n$-
Stable Instances $k, \clients, \facilities, \dist$.}\ \\
\begin{enumerate}
    \item $S \gets $ Local Search on $(k, \clients, \facilities, \dist)$, see \Cref{sec:localSearchAnalysis}
    \item Let $W$ be the set produced by the
    $s^*$-$D^2$-Sample
    process (\Cref{sec:dsampleproc}) on $S$.
    \item For each $\calB{\in \calL_{\texttt{bal}}}$ {produced as in \Cref{lem:ballguesses}} on $W$ (\Cref{sec:ballguesses})   
    \begin{enumerate}
    \item Compute the set of dummy centers $\dummyset$.
        \item For each {$Q\in \calLexp$ produced by $expRem()$} as in \Cref{lem:successguessprocess} 
        (\Cref{sec:removalofExpensive}):
        \begin{enumerate}
            \item \label{step:finalalg:submodular} For each $\bcalU{\in \calL_{cheap}}$, with the associated $\tmu$, {produced by $cheapRem()$ as in \Cref{lem:successcheapremove} }           (\Cref{sec:removalofCheap}):\newline 
\quad            solve the Submodular Instance associated with $(W,\calB,\calQ,\bcalU,\tmu)$ and compute the resulting solution $S_{(W,\calB,\calQ,\bcalU,\tmu)}$ for the $k$-Means problem (\Cref{sec:submodularopt})
        \end{enumerate}
    \end{enumerate}
    \item Output the solution $S_{(W,\calB,\calQ,\bcalU,\tmu)}$ of smallest cost. 
\end{enumerate}
\end{mdframed}

Our approach consists of running the above algorithm
$\log n$ times and take the minimum cost solution output.
In the remaining, we argue that the above algorithm yields a ${(5+O(\sqrt{\eps}))}$-approximation with probability at least $(1-\eps)(1-1/n)$. \Cref{thm:mainadditivecenters} then follows as the probability that all $\log n$ executions of the above algorithm would fail is at most $\left( 1- (1-\eps)(1-1/n)\right)^{\log n} < 1/n$, and so with high probability we output a ${(5+O(\sqrt{\eps}))}$-approximation. 

We first discuss the algorithm's running time.
The local search algorithm runs in polynomial time. 
The submodular optimization step also runs in polynomial time by \Cref{lemma:findingSstar}.
The output of the $s^*$-$D^2$-Sample process (which is clearly a polynomial time procedure) is 
a subset of size $s^*$. {By \Cref{lem:ballguesses}, we construct $\calL_{\texttt{bal}}$ in polynomial time the and $|\calL_{\texttt{bal}}| \leq n^{\eps^{-O(1)}}$. Similarly, by \Cref{lem:successguessprocess}, we construct $\calLexp$ in polynomial time and  $|\calLexp| \leq n^{\eps^{-O(1)}}$. 
{Finally, by \Cref{lem:successcheapremove}, we construct in polynomial-time 
 $n^{\eps^{-O(1)}}$ candidate pairs $\bcalU,\tmu$.} It follows that we solve $n^{\eps^{-O(1)}}$ submodular function optimization problems, each in polynomial time, and thus the total running time is polynomial.

We turn to proving the approximation guarantee.
We aim to show that there exists a 
$(W,\calB, \calQ, \bcalU, \tmu)$ that is successful with probability at least $(1-\eps/2)$.
Assuming this, the $({5+O(\sqrt{\eps})})$-approximation follows by combining
\Cref{claim:Mprimebounds} and \Cref{lemma:findingSstar}.

Our algorithm runs the local search algorithm and
thus provides us with a $\apxLSkmeans$ approximation $S$ that satisfies \Cref{lem:purecost}. Our algorithm next applies
the $s^*$-$D^2$-Sample process and by \Cref{lem:probcost} 
finds a successful $W$ with probability at least
$1-\eps$. Condition on having a successful
$W$, we then have that \Cref{lem:cheapcost} holds, and so does \Cref{lemma:convenience_upper_bound} by combining \Cref{lem:purecost} and \Cref{lem:cheapcost}.

Next, \Cref{lem:ballguesses} implies that one 
set of balls $\calB$ is a valid set of balls. This
set of balls induces a set of dummy centers 
$\dummyset$. From there, by \Cref{lem:successguessprocess} we have that 
with probability at least $1-1/n$, the 
{$expRem()$} procedure produces a set
of center $\mathcal{Q}$ such that     \begin{enumerate}
    \item $\calQ \subseteq S_0$, and
    \item The total cost in solution $S {-} Q \cup \dummyset$ of the clusters in $S_0 {-} \calQ$ is at most $\eps \cdot \sopt$.
    \end{enumerate}
and so we have obtained a successful 
$(W,\calB,\calQ)$ and \Cref{lemma:costboundofMOandMDwithassignments}
applies. Condition on the event that we obtain
a successful $(W, \calB, \calQ)$, then, applying \Cref{lem:successcheapremove} we have that 
the {$cheapRem()$} procedure outputs a pair
$(\bcalU, \tmu)$ satisfying the properties of 
\Cref{lem:successcheapremove}.
We thus have a $(W, \calB, \calQ, \bcalU,\tmu)$ such that $W$ is successful, the set of balls $\calB$ from \Cref{sec:ballguesses} is valid,  $\calQ$ selected in \Cref{sec:removalofExpensive} satisfies the properties of \Cref{lem:successguessprocess}, and $\bcalU,\tmu$  selected in this section satisfies the properties of \Cref{lem:successcheapremove} and
so $(W, \calB, \calQ, \bcalU,\tmu)$ is successful. Finally, note that the only probabilistic steps were the sampling of $W$ and of $\calQ$, and both steps are successful with probability at least $(1-\eps) \cdot (1-1/n)$, as required. 



\printbibliography
\end{document}